\documentclass[11pt,reqno]{article}
\usepackage{amsmath,amssymb,amsfonts,amsthm,bbm,mathrsfs,verbatim} % Typical maths resource packages
\usepackage[misc]{ifsym}
\usepackage{bbm}
\usepackage{graphics}                 % Packages to allow inclusion of graphics
\usepackage{graphicx}
\usepackage{float}
\usepackage{dsfont}
\usepackage{array}
\usepackage{mathtools}
\usepackage{blkarray}
\usepackage{makecell}
\usepackage{color}                    % For creating coloured text and background
\usepackage{hyperref}                 % For creating hyperlinks in cross references
\usepackage{makeidx}
\usepackage{caption}
\usepackage{bbm}
\usepackage{algorithmic}
\usepackage{authblk}
\usepackage{amsmath} % for \boxed and \smash[b] macros
\usepackage{booktabs}% for \midrule and \cmidrule macros
\DeclarePairedDelimiter\abs{\lvert}{\rvert}

\newcommand\headercell[1]{%
	\smash[b]{\begin{tabular}[t]{@{}c@{}} #1 \end{tabular}}}

\newtheorem{theorem}{Theorem}[section]
\newtheorem{proposition}[theorem]{Proposition}
\newtheorem{corollary}[theorem]{Corollary}
\newtheorem{lemma}[theorem]{Lemma}
\theoremstyle{definition}
\newtheorem{remark}[theorem]{Remark}
\newtheorem{definition}[theorem]{Definition}
\newtheorem{example}[theorem]{Example}

\newtheorem{notation}[theorem]{Notation}

\begin{document}
	\title{Probability of Default modelling with L\'evy-driven Ornstein-Uhlenbeck processes and applications in credit risk under the IFRS 9}
	
	\author[1]{K. Georgiou}
	\author[1]{A.N. Yannacopoulos}
	\affil[1]{\small Department of Statistics and Stochastic Modeling and Applications Laboratory, Athenss University of Economics and Business, 76 Patission Str. 10434, Athens}
	
	\date{}
	\maketitle

	\begin{abstract}
		In this paper we develop a framework for estimating Probability of Default (PD) based on stochastic models governing an appropriate asset value processes. In particular, we build upon a L\'evy-driven Ornstein-Uhlenbeck process and consider a generalized model that incorporates multiple latent variables affecting the evolution of the process. We obtain an Integral Equation (IE) formulation for the corresponding PD as a function of the initial position of the asset value process and the time until maturity, from which we then prove that the PD function satisfies an appropriate Partial Integro-Differential Equation (PIDE). These representations allow us to show that appropriate weak (viscosity) as well as strong solutions exist, and develop subsequent numerical schemes for the estimation of the PD function. Such a framework is necessary under the newly-introduced International Financial Reporting Standards (IFRS) 9 regulation, which has imposed further requirements on the sophistication and rigor underlying credit modelling methodologies. We consider special cases of the generalized model that can be used for applications to credit risk modelling and provide examples specific to provisioning under IFRS 9, and more.
	\end{abstract}

	\tableofcontents

	\section{Introduction}
	One of the main issues currently concerning financial institutions is the implementation of the new International Financial Reporting Standards (IFRS) 9. Due to the financial crisis, the purpose of the updated standards is to introduce a framework under which institutions forecast credit losses (for loan provisioning purposes). Specifically, “under the impairment approach in IFRS 9 it is no longer necessary for a credit event to have occurred before credit losses are recognised. Instead, an entity always accounts for expected credit losses, and changes in those expected credit losses. The amount of expected credit losses is updated at each reporting date to reflect changes in credit risk since initial recognition and, consequently, more timely information is provided about expected credit losses”. Furthermore, "the objective of the impairment requirements is to recognise lifetime expected credit losses for all financial instruments for which there have been significant increases in credit risk since initial recognition — whether assessed on an individual or collective basis — considering all reasonable and supportable information" (IFRS 9 Red Book). Further details and research are given in \cite{beerbaum2015significant} and \cite{xu2016estimating}. Hence, loan provisioning regulations under IFRS 9, require financial institutions to consider expected losses based on the current credit state of each loan and the possible future losses. This estimation requires knowledge of the lifetime Probability of Default (PD) for all loan exposures, as well as additional risk parameters such as the Loss Given Default (LGD) and the Exposure at Default (EAD), and finally being able to update these quantities dynamically under changing market conditions. The estimation of future losses (i.e., forward looking provisions) can therefore be tackled by employing the theory of stochastic processes and their dynamics.
	\par Under IFRS 9, it is now mandatory for financial institutions to classify loans into three distinct categories, known as the IFRS 9 Stages. Specifically, Stage 1 loans are considered performing, Stage 2 contains loans which have displayed a significant increase in credit risk and Stage 3 contains all Non-Performing loans (NPLs), considered to have defaulted. As mentioned, the institutions must forecast future losses specifically for the Stage 2 loans, which are considered to be at risk. The resulting Stage 2 provisions are referred to as the Expected Lifetime Provisions (ECL). In the present work, we will focus on portfolios of corporate and small business loans, where it is common practice to consider the company's assets to be governed by stochastic process (see e.g., \cite{barndorff2001non} and \cite{barndorff2001modelling}). Under this assumption, the PD associated with each loan depends on the underlying asset process and we can define the PD as the probability that the asset process falls below a fixed threshold. 
	
\par Perhaps one of the most influential changes due to the IFRS 9 is the requirement for financial institutions to consider Expected Lifetime Provisions (ECL), whereby future losses must be forecast using mathematically robust and rigorous methods, for all Stage 2 credit exposures, which are considered to have displayed a significant increase in risk. This estimation requires knowledge of the lifetime Probability of Default (PD) for all loan exposures, as well as additional risk parameters such as the Loss Given Default (LGD) and the Exposure at Default (EAD), and finally being able to update these quantities dynamically under changing market conditions. There exist recent papers detailing and studying the ECL calculation, such as \cite{beerbaum2015significant} and \cite{xu2016estimating}. We extend this modelling framework by employing the theory of stochastic processes and their dynamics for the estimation of lifetime PDs and future losses (i.e., forward looking provisions). In this paper, we aim to develop a stochastic modelling framework under which the PD process can be considered mathematically and practically. Using this approach we can address in a robust and efficient manner the challenging provisioning, forecasting and pricing tasks under IFRS 9.
\par In general, calculating default probabilities both analytically and numerically is of paramount importance in risk management and a broad range of financial applications. However, particularly under IFRS 9, credit loss forecasting has introduced the need for robust structural models, that can be used for pricing and provisioning purposes. To this end, we will consider stochastic models for the evolution the underlying asset value process, whose default will be studied as an appropriate firs-time-hitting problem. Therefore, it can be assumed that we are working mainly within portfolios of corporate and small business loans, where it is common practice to consider the company's assets to be governed by stochastic process (see e.g., \cite{barndorff2001non} and \cite{barndorff2001modelling}). Under this assumption, the PD associated with each loan depends on the underlying asset process and we can define the PD as the probability that the asset process falls below a fixed threshold. More specifically, we will assume that the asset process is governed by an Ornstein - Uhlenbeck (OU) process with a jump component, a member of the family of jump diffusion processes. We note that practitioners may consider the evolution of asset-dependent processes instead, e.g., returns; such processes can still be described by similar stochastic models, rendering the methods proposed in this paper applicable in these cases, as well. For brevity, hereinafter we will refer to this underlying process as the asset process, with the understanding that is can be replaced with any related dynamics considered appropriate by practitioners. Important theoretical background of such processes and their properties are given in \cite{applebaum2009levy} and \cite{oksendal2007applied}. The jump process will account for abrupt changes in the asset processes, which are very common in practice and are closely related to loan defaults. Obtaining the evolution of the PD values, based on the stochastic asset model, will allow us to then tackle various modelling tasks which are currently open problems for financial institutions under the IFRS 9 framework. 
\par It is common in the literature to consider two separate cases for default probabilities:
\begin{itemize}
	\item A variable starting time and constant time horizon, defined by the process.
	\item A variable time horizon and constant starting time.
\end{itemize}
Our results in this paper constitute a generalization that combines these two cases (as analyzed in \cite{mishura2016ruin} and \cite{moller1995stochastic}). Specifically, we consider a generalized "Probability of Default function" and prove that, under certain homogeneity assumptions, we can obtain Integral Equations (IEs) and Partial Integro-Differential Equations (PIDEs) for both aforementioned PD cases.
Finally, using the IE formulation we will prove the existence of the PD values and the solvability of the PIDEs in the viscosity sense (details are given in the corresponding Section), in order to obtain estimates that can be used for the aforementioned modelling tasks, without having to always assume and/or prove strict regularity conditions, and further consider the conditions under which these solutions can be considered strong, with the required smoothness.
\par The above methodology will be used to consider real-life examples of loan provisioning calculations, scenario analysis and pricing, exemplifying the wide range of credit risk modelling tasks the proposed methodology can address. Finally, we note that, even though motivated by credit risk, the approaches detailed in this paper can find applications in other areas of financial mathematics, such as derivatives pricing, where the use of stochastic modelling remains prevalent, e.g., in the pricing of barrier options. 
 \par This paper is structured as follows. In Section \ref{background} we recall the background of similar stochastic processes in the literature and define the corresponding Probability of Default functions, in Section \ref{OU-Section} we present the generalized model for the asset value process and discuss special cases which are applicable in the IFRS 9 framework. In Section \ref{integral-section} we obtain IEs for the PD functions under each model considered and use them to prove that the PD functions satisfy certain mathematical properties, and in Section \ref{pides-visc} we obtain the PIDEs for the PD functions. The remainder of the paper is dedicated to numerical approximations and applications: in Section \ref{num-schemes} we construct appropriate Finite Difference numerical schemes to approximate the PD functions and, finally, in Section \ref{ifrs-section}, we present examples of the proposed methods applied to IFRS 9 provisioning, credit derivatives pricing and credit optimization problems.

\section{Aims and modelling framework}\label{background}
\par We start by discussing the PD process which, in its most general form, can be written as a function of both the starting time and maturity, as well as of the initial position of the corresponding asset value process. Furthermore, to accurately model real-life dynamics, it is necessary to account for the dependence on latent variables which affect the PD. Incorporating such processes, which in practice are e.g., macroeconomic variables or different market regimes, is of paramount importance as it largely affects PD estimation and subsequent  modelling results. Rigorously accounting for these exogenous variables is therefore necessary, and will allow us to consider a large family of stochastic processes that can be used in practice. To begin, consider compact and bounded sets $\mathcal{D}, \mathcal{D}_i \subset \mathbb{R}$, for $i=1,\dots,d$. Then, define the PD function:      

\begin{definition}\label{PD_def}
	Consider $x \in \mathcal{D} $ and the vector of (discrete or continuous) stochastic processes $(X_t^i)_{t \geq 0}$ with corresponding state spaces $\mathcal{D}_i$, for $i=1,\dots, d$. Furthermore, consider the stochastic asset value process $(G_t)_{t\geq s}$, with initial value $G_s=x$ and which depends on $(X_t^i)_{t \geq 0}$, for $i= 1,\dots, d$. Then, we define the Probability of Default function $\Psi:\mathcal{D} \times \mathcal{D}_1 \times \dots \times \mathcal{D}_d \times [0,T] \times [0,T] \rightarrow [0,1]$, for some fixed $T>0$, by:
	\begin{flalign}\label{pd-1}
		\Psi(x,x_s^1,x_s^2, \dots, x_s^d, s,t)= \mathbb{P}\Big(\inf_{\substack{s \leq r \leq t}} G_r \leq 0 | G_s=x, X^1_s = x^1_s, X^2_s = x^2_s, \dots, X^d_s = x^d_s \Big),
	\end{flalign} 
	and the corresponding survival probability $\Phi:\mathcal{D} \times \mathcal{D}_1 \times \dots \times \mathcal{D}_d \times [0,T] \times [0,T]  \rightarrow [0,1]$ by:
	\begin{eqnarray}
		\Phi(x,x_s^1,x_s^2, \dots, x_s^d, s,t) = 1-\Psi(x,x_s^1,x_s^2, \dots, x_s^d, s,t).
	\end{eqnarray}
\end{definition}
To motivate this definition and its usefulness, notice that by fixing $s$ we obtain the standard finite-horizon ruin probability (see e.g., \cite{mishura2016ruin}), whereas by fixing $t$ we obtain the ruin probability with variable starting time, as defined in \cite{moller1995stochastic}, which can be used to define a martingale. Finally, allowing $t \rightarrow \infty$ we obtain the infinite-horizon ruin probability. 
\par Modelling the evolution of PD functions has become even more important under IFRS 9, due to the increased complexity of provision calculations and Staging criteria. In general, all aforementioned PD functions, corresponding to variable maturity or starting times find many applications and have been considered in the field of credit risk, such as in \cite{ballotta2015counterparty}, \cite{schoutens2010levy} and \cite{zhou1997jump}. For example, the case of a variable maturity is often referred to as the Lifetime Probability of Default and is used extensively for provisioning and pricing purposes. Particularly in the context of IFRS 9 modelling, the Lifetime Probability of Default is used to assess credit risk at origination, as well as for Expected Lifetime Provisions for Stage 2 loans. We give detailed examples of such calculations in Section \ref{ifrs-section}. 

\par As mentioned, it is standard in the field of financial mathematics to consider the evolution of a debtor's assets to be governed by a stochastic process. A well-documented process that is used in various such applications is the Ornstein-Uhlenbeck (OU) process. In Appendix \ref{ou-app} we recall important properties of the OU process, a generalized version of which we will consider in this paper. OU models have been considered in past research and many applications. For example, a well-known special case is the Vasicek model \cite{vasicek1977equilibrium}. Further work has explored  the Merton model for default with underlying dynamics given by the continuous OU process, and has been extended to cases incorporating jumps. These find important applications particularly in credit risk modelling and pricing; see e.g., \cite{hull2004merton} and  \cite{barndorff2001non}, respectively. By including a jump process to the continuous OU asset process, we obtain the L\'evy-driven (jump) Ornstein-Uhlenbeck process, defined in (\ref{hom-gou}). 
\begin{eqnarray} \label{hom-gou}
	\label{gou1}
	dG_u=k(\theta - G_u)du + \sigma dB_u +  \int_{\mathbb{R}} z N(du,dz), \,\,\,\ G_0 = x.
\end{eqnarray} 
\par This is a natural generalization, as significant changes in credit events are often abrupt and unpredictable (particularly a deterioration in creditworthiness), corresponding to a discontinuous component in the driving stochastic process. Indeed, the goal of credit risk requirements under IFRS 9 is to ensure that financial institutions and their customers are protected against such rare and unexpected events and the subsequent losses. It is therefore important to capture the effect of such events mathematically, which is why this model will form the basis of our analysis and will be used to construct more sophisticated models in the next section. We will employ the fact that the jump OU process is time homogeneous so that, rather than considering a starting time $s$ and initial position $G_s$, we can define the time until maturity by $u:=t-s$ and consider $(G_u)_{u\geq 0}$, as above, equivalently. This is an important property that we will take advantage of to simplify the PD estimation.
\par To conclude, we note that the use of L\'evy processes for financial modeling is well-documented and established. \cite{schoutens2010levy} gives an extensive analysis of L\'evy processes and their use for asset process modelling, credit derivatives pricing and more. In \cite{luciano2006multivariate} and \cite{onalan2009financial} the authors consider a L\'evy-driven OU process, and L\'evy multivariate models for assets processes. The former fits the model parameters to the General Motors stock price, while the latter considers many different indices, obtaining suprisingly accurate results. We also refer the interested reader to \cite{ballotta2016multivariate} for a detailed analysis of the properties of the multivariate model. Seminal work has also been done in the study of L\'evy-driven OU processes in \cite{barndorff2001non}. Finally, well-documented numerical methods exist for the calibration of stochastic models with jumps, such as the Yuima framework for stochastic differential equations in R statistical language \cite{brouste2014yuima}.

\section{The generalized asset value model and PD function}\label{OU-Section}

In this section we develop a stochastic model that incorporates the exogenous variables required when considering asset value processes. To incorporate such effects, we build upon the family of regime switching and stochastic volatility models, as described below. We combine these to produce a generalized model, which we will use to construct a framework that encapsulates a large family of stochastic processes that can be used for asset value modelling and subsequent credit risk calculations. In addition to the mathematical results presented in this paper, we highlight that the framework developed using the generalized model addresses the strict requirements under IFRS 9, whereby credit risk modelling is required to incorporate multiple appropriate latent variables, whilst adhering to mathematical rigor.       

\subsection{Regime switching and stochastic volatility models}
\par First recall that loan exposures under the IFRS 9 framework are now classified into three Stages. Each of these Stages correspond to a given level of risk, with the most noteworthy change being the introduction of Stage 2 loans, i.e., credit exposures which have exhibited a significant increase in credit risk (SICR event, which can be defined by the institution, e.g., as a statistically significant increase in PD, a delinquency warning flag etc.). By definition, changes in the risk profile of an exposure will correspond to changes in the dynamics of the underlying asset process. For example, a debtor may request restructuring, or may be 30 days delinquent. This will trigger a SICR event, which can then affect the underlying asset value process. To capture this dependency we consider a regime switching model for the asset process, whereby the parameters of the stochastic process vary according to the underlying rating (Stage) of the exposure. We can do this by considering the Continuous Time Markov Chain (CTMC) $(R_t)_{t \geq 0}$ describing the rating at time $t$, where the set of all loan ratings is denoted by $\mathcal{R}$, with cardinality $|\mathcal{R}| = R$. Therefore, we obtain the following jump diffusion with Markov switching model: 
\begin{eqnarray}\label{model-rs-def}
	dG_u=k(R_u)\big(\theta(R_u) - G_u\big)dt + \sigma(R_u) dB_u+  \int_{\mathbb{R}} z N(du,dz), \,\,\,\ G_0 =x, R_0 = \rho,
\end{eqnarray}
with $\rho \in \mathcal{R}$. Note that in subsequent sections we adopt the notation $k_{\rho}, \theta_{\rho}, \sigma_{\rho}$, for brevity. For a reminder of CTMC processes and their properties see Appendix \ref{Markov}.
\par %In general, we can consider the vector of parameters ${\bf \Theta}_{R_t}=$ $(\theta_{R_t}, k_{R_t}, \sigma_{R_t}, \lambda_{R_t}, N_{R_t}(t,z) )$, defining the asset process, where $R_t$ is the rating of the credit exposure. 
\par In the sequel, to develop a realistic model we want to capture the effects of macroeconomic variables, which naturally affect the evolution of the asset process; this is necessary for the modelling tasks we will consider under IFRS 9, as previously discussed. Typically, such latent variables are incorporated by considering stochastic volatility models, whereby the diffusion term of the asset process also evolves according to a stochastic process, as described by the coupled process: 
\begin{eqnarray}
	\begin{cases}
		dG_t = \mu_x(G_t, Y_t)dt +\sigma_x(G_t, Y_t)dB_t + \int_{\mathbb{R}} z N(dt,dz), \,\,\,\ G_s = x, \\
		dY_t =  \mu_y(Y_t)dt +\sigma_y(Y_t)dW_t,  \,\,\,\ Y_s = y,\\
	\end{cases}
\end{eqnarray}
for $y \in \mathcal{V}$ and where $B_t$ and $W_t$ are independent Brownian motions. Standard cases are the Bates' model, introduced in \cite{bates1996jumps}, as well as the Heston model (see \cite{benhamou2010time}), a version of which we consider below. In particular, letting $\mu_x(Y_t) = k(\theta -Y_t)$ and $\sigma_x(Y_t)=\sqrt{Y_t}$, we obtain the asset process driven by a stochastic volatility process, which follows the well-established Cox–Ingersoll–Ross (CIR) model, developed in \cite{cox2005theory} (note that both processes are time-homogeneous):  
\begin{eqnarray} \label{model-sv-def}
	\begin{cases} 
		dG_{u} = k(\theta -G_u)dt +\sqrt{Y_t}dB_t + \int_{\mathbb{R}} z N(dt,dz), \,\,\,\ G_0 = x,\\
		dY_{u} = \kappa(\mu- Y_u)dt  + \xi \sqrt{Y_u}dW_t, \,\,\,\ Y_0=y.\\
	\end{cases}
\end{eqnarray}

\par The above models are widely used in mathematical finance and stochastic modelling. Regime switching is a well-documented approach in financial modelling (see \cite{hamilton2010regime}), with applications ranging from macroeconomics (e.g., \cite{aristidou2018meta}) to option pricing (e.g., \cite{duan2002option}, \cite{hainaut2014intensity}) and interest rate modelling (\cite{goutte2011conditional}). In the case of credit risk, the underlying Markov chain is considered as an indication of the market conditions, which significantly impacts credit exposures and ratings. In subsequent sections we add to the multitude of applications by using the PD function that arises from the regime switching model to estimate Lifetime provisions and scenario analysis under IFRS 9. When considering regime switching in asset processes it is important to note that financial institutions may currently have various credit rating systems, which are not compatible with the IFRS 9 staging (which requires three district Stages for exposure ratings). However, recent work has shown that this is not restrictive and IFRS 9 - compatible transition matrices can be obtained from the existing internal ratings (see \cite{georgiou2021markov}). Finally, we refer the reader to \cite{zhu2015feynman} for a detailed analysis of more general regime switching jump diffusion processes, where the authors also consider the dynamics of the underlying Markov process to be a function of the initial position of the jump diffusion.
\par The stochastic volatility model is a natural extension, as it can be seen as the limit process of the regime switching model, as $\mathcal{R} = \mathbb{R}_{+}$. Such models, in the case of both continuous and jump processes, have also been considered for numerous applications in mathematical finance, particularly in pricing and hedging, such as in \cite{toivanen2010componentwise} and \cite{goutte2013pricing}.

\subsection{The generalized model}
Both the aforementioned models have important applications in credit risk modelling under the IFRS 9 framework. Combining the two, we obtain a generalized model that captures all the observable or latent variables required to estimate the PD evolution and subsequently tackle the IFRS 9 modelling tasks. This generalized model is of the form:
\begin{eqnarray} 
	\begin{cases} 
		dG_u=k(Y_u, R_u)\big(\theta(Y_u,R_u) - G_u\big)dt + \sigma(Y_u, R_u) dB_u+  \int_{\mathbb{R}} z N(du,dz), \,\,\,\ G_0 = x, \\
		dY_{u} = \kappa(\mu- Y_u)dt  + \xi \sqrt{Y_u}dW_t. \\
	\end{cases}
\end{eqnarray}
Specifically, we will be considering a combination of (\ref{model-rs-def}) and (\ref{model-sv-def}), which gives rise to the following.
\begin{definition} Under the generalized model, the asset value process is defined by the triple $(G_t, R_t, Y_t)_{t\geq 0}$, capturing both the switching and volatility processes, and is given by:
	\begin{eqnarray} \label{model-gen-def}
		\begin{cases} 
			dG_u=k(R_u)\big(\theta(R_u) - G_u\big)dt + \sigma(R_u) \sqrt{Y_t} dB_u+  \int_{\mathbb{R}} z N(du,dz), \,\,\,\ G_0 = x,\\
			dY_{u} = \kappa(\mu- Y_u)dt  + \xi \sqrt{Y_u}dW_t, \\
		\end{cases}
	\end{eqnarray}
	with $G_0 =x, R_0 = \rho$ and $Y_0 =y$. 
\end{definition}

Before moving on to define the appropriate Probability of Default functions, it will be useful to define some notation.

\begin{remark}\label{gen-unif-cont}
	An important note is that the transition probability of the generalized OU process is also uniformly continuous. This follows from the fact that its transition probability, for given $(R_0, Y_0)=(\rho, y)$, $p(x',x,t;\rho, y)$, is simply the coupling of the transition densities of the corresponding stochastic volatility models, which are continuous functions (see e.g., \cite{ackerer2018jacobi}), and is therefore uniformly continuous on all closed and bounded intervals $\mathcal{D}$ we will consider. 
\end{remark}

\begin{notation}\hfill
	\begin{itemize}
		\item[1.] Throughout the remainder of this paper, we employ the notation $Z_u^x$ to represent the stochastic process $(Z_u)_{u\geq 0 }$, with $Z_0 = x$, where appropriate. We also generalize this notation to incorporate cases with additional underlying variables $X_t^1, X_t^2, \dots, X_t^n$, by writing $Z_u^{(x_1,x_2,\dots,x_n)}$ to represent $(Z_u)_{u\geq 0 }$ with $X_0^i=x_i$ for $i = 1,2,\dots, n$, (the superscripts are to be understood as indices, i.e., the $i-$th underlying variable is $(X^i_t)_{t \geq 0}$).
		\item[2.] In the following sections, when referring to the transition densities of the regime switching, stochastic volatility and generalized OU models, we will omit the dependence on the latent variables for brevity, as it will be obvious from the context.
	\end{itemize}
	
\end{notation}

\subsection{The Probability of Default function}
\par Following the definition of the PD function, as given in (\ref{pd-1}), under the generalized model (\ref{model-gen-def}) we will condition on the initial state of the regime switching and stochastic volatility processes, i.e., $\rho$ and $y$, to obtain:
\begin{eqnarray}
	\Psi(x,\rho, y, s,t) := \mathbb{P}\Big(\inf_{\substack{s \leq r \leq t}} G_r \leq 0 | G_s=x, R_s = \rho, Y_s = y\Big).
\end{eqnarray}
Under this assumption, we can utilize the time homogeneity property to write the PD function more succinctly, whilst still being able to obtain the evolution of the PD, both in the case of a variable maturity and variable starting time. We describe this in the lemma below.
\begin{lemma}
	Under the generalized model (\ref{model-gen-def}), the PD function with variable maturity $\Tilde{\Psi}(x,\rho, y, s;t)$ and variable starting time $\Tilde{\Psi}(x,\rho, y,s;t)$, can be retrieved from the generalized function $\Psi(x,\rho, y, u)$, where $u = t-s$ represents the remaining time until maturity.
\end{lemma}
\begin{proof}
	As mentioned, this observation follows immediately from the time homogeneity of the asset process, since: 
	\begin{flalign}
		\Psi(x,s,t) &= \mathbb{P}\Big(\inf_{\substack{s \leq r \leq t}} G_r \leq 0 | G_s = x, R_s = \rho, Y_s = s\Big) \notag \\ &=\mathbb{P}\Big(\inf_{\substack{0 \leq r \leq t-s}} G_r \leq 0| G_0=x, R_0 = \rho, Y_0=y \Big) = \Psi(x,\rho, y, 0,t-s).
	\end{flalign}
	We can now write $\Phi(x,\rho, y, u)$ with $u:=t-s$ representing the remaining time until maturity. 
\end{proof}
The above shows that, by fixing the appropriate time and with a simple change of variables, we can obtain the evolution of both PD processes. It easily follows that this approach can be generalized to any time homogeneous stochastic processes. Throughout the remainder of this paper, we will therefore use the following formulation of the PD and corresponding survival process: %In a few necessary cases, we will employ the notation $\Psi(x,t;s), \Psi(x,s;t)$, for clarity. 
\begin{eqnarray}
	\Psi(x,\rho, y, u):= \mathbb{P}\Big(\inf_{\substack{r \leq u}} G_r \leq 0| G_0=x, R_0 = \rho, Y_0 = y\Big) \equiv \mathbb{P}\Big(\inf_{\substack{r \leq u}} G_r^{(x,\rho, y)} \leq 0 \Big), \label{gen-pd}\\
	\Phi(x,\rho, y,u) := 1- \Psi(x,\rho, y, u)=\mathbb{P}\Big(\inf_{\substack{ r \leq u}} G_r^{(x,\rho, y)} > 0\Big).\label{gen-surv}
\end{eqnarray}

\begin{remark}
	It is worth emphasizing that we use the general term "time until maturity" purposefully, as it captures both PD cases, under the homogeneity assumption. We will continue to use this term throughout this and subsequent papers, to instill the importance of this generalization. Furthermore, the homogeneity assumption is strong, yet fair. Particularly in the case of corporate and/or small business loans, it is natural to consider such asset processes, since credit risk modelling is often done across complete financial/business cycles (e.g., years or quarters), over which the evolution of the asset process (or related return processes) will have similar dynamics, regardless of the exact point in time. However, even without this assumption, the approaches developed in this paper can be used by fixing the either the starting or maturity time in order to obtain whichever case of the PD process the modeller requires. Hence, this framework is useful for PD modelling under any asset value process.  
\end{remark}

\par Using the generalized model and the corresponding PD (or survival) process (\ref{gen-pd}) (or (\ref{gen-surv})), we can prove certain mathematical properties that are required to ensure the existence of appropriate solutions for the Partial Integro-differential Equations (PIDEs) we will obtain for the PD functions. This creates a complete and robust mathematical framework which can be applied even without assuming or proving regularity, and is therefore applicable to a wide range of asset value models. At the same time, it is important to discuss the practical implications and applicability of the approaches described in this and subsequent sections. When considering real-life credit risk modelling tasks the state of the regime (e.g., the IFRS 9 Stage), and/or the value of any underlying macroeconomic factors may be observable and can therefore be inserted explicitly into the generalized model (\ref{model-gen-def}), thereby obtaining the regime switching or stochastic volatility model, with PD functions given by:
\begin{eqnarray}
	\Psi(x,\rho, u)= \mathbb{P}\Big(\inf_{\substack{r \leq u}} G_r^{(x,\rho)} > 0 \Big), \label{pd-rs}\\
	\Psi(x,y, u)= \mathbb{P}\Big(\inf_{\substack{r \leq u}} G_r^{(x,y)} > 0 \Big), \label{pd-sv}
\end{eqnarray}
respectively, and corresponding PD (or survival) functions $\Psi(x,\rho,u)$ and $\Psi(x,y,u)$ (or $\Phi(x,\rho,u)$ and $\Phi(x,y,u)$). These models can then be used for the tasks we consider under IFRS 9, and credit risk more generally. For this reason, in Section \ref{num-schemes} we develop numerical schemes for such simplified versions of the generalized model, starting with the one-dimensional L\'evy-driven OU asset and its corresponding PD function $\Psi(x,u) =\mathbb{P}\Big(\inf_{\substack{r \leq u}} G_r^{x} \leq 0 \Big)$ (and survival $\Phi(x,u)$), and continuing with the regime switching and stochastic volatility models, which will be used for related applications. %All the solutions obtained from the numerical schemes will be understood in an appropriate weak sense, under the framework developed using the generalized model in Section \ref{pides-visc}.

\begin{remark}
	We refer the reader to Appendix \ref{gens-cou} for a brief overview of useful PDEs that the survival processes and transition densities of the non-jump versions of the above models satisfy (i.e., the versions which do not contain the L\'evy jump process). These, known as Kolmogorov backward equations, can be written in terms of the infinitesimal generators of the stochastic processes, and will be referred to in the subsequent sections, where we will obtain similar equations for the L\'evy-driven models. We also note that we adopt the notation for the generator operators used in these equations for the remainder of the paper, for convenience. 
\end{remark}

\section{Integral characterization and properties of the PD function}\label{integral-section}

As previously mentioned, our approach ultimately relies on deriving and solving (PIDEs) for the PD functions. To obtain these equations, we will first consider an integral equation (IE) characterization of the PD under the generalized model, from which we can obtain similar representations for the simplified models. We develop these IEs in this section, which will allow us to prove that the PD functions enjoy the properties required so as to be considered appropriate (either weak or strong) solutions to the PIDEs. 

\subsection{Required notation for the Integral Equations}
To ease the calculations presented in this section, we first introduce some notation, which will be used for the integral equations.

\begin{definition} %define operator
	Consider $x \in \mathcal{D}$ and the vector of stochastic processes $(X_t^i)_{t \geq 0}$, with corresponding state space $\mathcal{D}_i$, for $i= 0,1,\dots, d$. For a fixed time $u \in [0,T]$, we define the family of operators $\big(\mathcal{T}_s, s \in [0,u]\big)$, acting on a function $\phi: \mathbb{R}^{d+1} \times [0,T] \rightarrow [0,1]$, by:
	\begin{flalign}
		\mathcal{T}_s \phi(x,x^1_0,x^2_0,\dots, x^d_0,u) = \mathbb{E}[\phi(x,X^1_s, X^2_s, \dots, X^d_s, u-s)| X^1_0 = x^1_0, X^2_0 = x^2_0, \dots, X^d_0 = x^d_0].
	\end{flalign}
	In our setting, $x$ corresponds to the initial position of the asset value process, i.e., $G_0=x$, and each $X^i_t$ represents a latent variable, such as the CTMC in the regime switching model or the volatility in the stochastic volatility model. To give analytic forms that will be used in these two models, respectively, we specify the following cases:
	\begin{itemize}	
		\item[$(i)$] Let $(X_t^i)_{t \geq 0}$, for $i= 1,2,\dots, d$ be discrete space and independent stochastic processes, with state space $\mathcal{X}^i$, such that $|\mathcal{X}^i| < \infty$, and with transition probabilities $\pi_i(k,x^i_0,t):=\mathbb{P}(X^i_t = k | X^i_0=x^i_0)$ for $k \in \mathcal{X}^i$. Then: 
		\begin{flalign}
			\mathcal{T}_s \phi(x,& x_0^1,x_0^2,\dots, x_0^d,u) \notag \\ =& \sum_{x^1_s \in \mathcal{X}^1} \cdots \sum_{x^d_s \in \mathcal{X}^d} \phi(x, x^1_s, x^2_s,\dots, x^d_s, u-s) \pi_1(x^1_s, x^1_0,s) \cdots \pi_n(x^d_s, x^d_0,s).
		\end{flalign}
		
		\item[$(ii)$] Let $(X_t^i)_{t \geq 0}$, for $i= 1,2,\dots, d$, be continuous and independent random variables, with supports $\mathcal{D}^i$ and transition densities $q_i(k,x^i_0,t):=\mathbb{P}(X^i_t = k | X^i_0=x^i_0)$ for $k \in \mathcal{D}^i$. Then:
		\begin{flalign}
			&\mathcal{T}_s \phi(x, x_0^1,x_0^2,\dots, x_0^d,u) \notag \\&= \int_{x^1_s \in \mathcal{D}^1} \cdots \int_{x^d_s \in \mathcal{D}^d} \phi(x, x^1_s, x^2_s,\dots, x^d_s, u-s) q_1(x^1_s, x^1_0,s)\cdots q_n(x^d_s, x^d_0,s)dx^1_s \cdots dx^d_s.
		\end{flalign}
	\end{itemize}	
	It is easy to see that in the case where $(X^i_t)_{t \geq 0}$ contains both discrete and continuous stochastic processes the analytical expression will contains both summation and integral terms. Hereinafter, we will refer to $\mathcal{T}_s$, for a fixed $s\in[0,u]$ as the \emph{$s-$operator}. In the calculations that follow we prefer to express the relevant integral equations in terms of the $s-$operator notation for convenience and simplicity.
\end{definition}

\subsection{Integral Equation formulations of the PD functions}
The first step in our methodology entails deriving the IEs for the PD functions. As is convention in most of the literature, we perform our calculations with the survival probability and it is easy to see that the same steps can be used for the corresponding PDs. These equations prove to be very useful, as they will allow us to prove continuity and existence results for the PD function. We prove the result under the generalized model, from which analogous results for regime switching and stochastic volatility cases are easily obtained.

\begin{proposition}\label{integral_eqn_gen}
	Consider the asset value process under the generalized model (\ref{model-gen-def}) with jump rate $\lambda = \nu(\mathbb{R})$ and jump size distribution $F(z)$. Furthermore, let $Q(x,\rho,y,u)$ and $p(x',x,s)$ be the survival probability and transition density, respectively, of the non-jump generalized OU process. Then, the survival probability $\Phi(x,\rho,y, u)$ satisfies the integral equation:
	\begin{flalign}\label{int-gen-prev}
		\Phi(x,\rho, y, u)  =  \int_{0}^{u} \lambda e^{-\lambda s} \int_{0}^{\infty} \int_{\mathbb{R}} \mathcal{T}_s \Phi(x'+z, \rho, y, u) p(x',x,s) dF(z)dx' ds+ e^{-\lambda u}Q(x,\rho, y, u).
	\end{flalign} 
	%where we use the notation $E_R[\Phi(\cdot, R_t^{\rho}, \cdot)]$ to represent the expectation with respect to the random variable $R_t$.
\end{proposition}
\begin{proof}
	Consider the natural filtration $\mathcal{F}_t$ generated by the by the tri-variate process $(G_t,R_t, Y_t)$. Recall that we use the notation $G_t^{(x,\rho,y)}$ to represent the asset value process depending on the regime CTMC and volatility process $Y_t$, with $G_0= x$, $R_0=\rho$ and $Y_0=y$. By definition, $\Phi(x,\rho,y,u)= \mathbb{P}\Big(\inf_{\substack{r \leq u}} G_r^{(x,\rho,y)} > 0 \Big)$ and it will therefore be useful to define the martingale:
	\begin{eqnarray}
		M_s = \mathbb{E}[\mathds{1}\big(\inf_{\substack{r \leq u}} G_r^{(x,\rho,y)} >0  \big)|\mathcal{F}_s],
	\end{eqnarray}
	for $s < u$. Furthermore, let $\tau$ be the time of the first jump and we will also define the stopping time:
	\begin{eqnarray}
		\tau^* = \inf\{ t < u : \{\Delta G_t^{(x,\rho,y)}\neq 0 \} \cap \{ G_s^{(x,\rho,y)} >0, \,\, \forall s \in[0,t] \} \},
	\end{eqnarray}
	i.e., the time the process first jumps, having not yet defaulted. It is easy to check that $\tau^*$ is indeed an $\mathcal{F}_t$-stopping time.
	\par On the event $\{\tau>u\}$ no jumps occur within the examined time horizon and therefore $\Phi(x,\rho,y,u) = Q(x,\rho,y,u)$, where recall that $Q(x,\rho,y,u)$ is the survival probability of the non-jump generalized OU process. On $\{\tau \leq u\}$, we have:
	\begin{eqnarray}
		\Phi(x,\rho,y,u) = \mathbb{E}[\mathds{1}\big(\inf_{\substack{r \leq u}} G_r^{(x,\rho,y)} > 0 \big)|\mathcal{F}_0] = M_0 = \mathbb{E}[M_{\tau^*}],
	\end{eqnarray} 
	where the last step follows from the Optional Stopping Theorem. Notice that $\mathbb{P}(\tau^* = \infty) >0$ and therefore the above is to be understood in an almost sure sense. Now, by the strong Markov property and the time homogeneity of the OU process, it follows that:
	\begin{flalign}\label{mark}
		%	\Phi(x,\rho,u) = \mathbb{E}[\Phi(G^{(x,\rho)}_{\tau^*}, R_{\tau^*}^{\rho} ,u-\tau^*)].
		\Phi(x,\rho,y,u) =& \mathbb{E}[\Phi(G^{(x,\rho,y)}_{\tau^*}, R_{\tau^*}, Y_{\tau^*} ,u-\tau^*)]\notag\\ = & \mathbb{E}[\Phi(G_{\tau^*}, R_{\tau^*}, Y_{\tau^*} ,u-\tau^*)| G_0 =x, R_0 = \rho, Y_0 = y].
	\end{flalign}
	Concerning the two cases for the time of the first jump, we can write:
	\begin{flalign} \label{eq}
		\Phi(x,\rho,y,u) =&  \mathbb{E}[\mathds{1}\big(\inf_{\substack{r \leq u}} G_r^{(x,\rho,y)} > 0 \big)\mathds{1}(\tau \leq u)] + \mathds{E}[\mathds{1}\big(\inf_{\substack{r \leq u}} G_r^{(x,\rho,y)} > 0 \big) \mathds{1}(\tau > u)] \notag \\ =& \mathds{E}[\mathds{1}\big(\inf_{\substack{r \leq u}} G_r^{(x,\rho,y)} > 0 \big) |\tau \leq u]\mathds{P}(\tau \leq u ) + e^{-\lambda u}Q(x,\rho,y,u).
	\end{flalign}
	Then, using (\ref{mark}), the law of total probability and the definition of the $s-$operator the first term can be written as: 
	\begin{flalign}\notag
		&\int_{0}^{u} \lambda e^{-\lambda s} \mathds{E}[\Phi(G_s, R_s, Y_s, u - s)| G_0 = x, R_0 = \rho, Y_0 =y] ds \notag\\ =& \int_{0}^{u} \lambda e^{-\lambda s} \int_{0}^{\infty} \int_{\mathbb{R}} \mathbb{E}[\Phi(x'+z,R_s,Y_s,u-s)|R_0 = \rho, Y_0 =y] p(x',x,s)dF(z)dx'ds \notag \\ =& \int_{0}^{u} \lambda e^{-\lambda s} \int_{0}^{\infty} \int_{\mathbb{R}} \mathcal{T}_s \Phi(x'+z, \rho, y, u) p(x',x,s)dF(z)dx'ds,  
	\end{flalign}  
	where we have conditioned on the pre-jump asset value and the subsequent jump size. Substituting back into (\ref{eq}), we get the required result.
\end{proof}

It is now straightforward to obtain the analogous integral equations for the PD functions under the regime switching and stochastic volatility models. We present both these results below, omitting the proofs for brevity, as they follow the proof of Proposition \ref{integral_eqn_gen}, almost identically.

\begin{corollary}\label{integral_eqn_2}\hfill
	\begin{itemize}
		\item[(i)] Consider the asset value process under the regime switching model (\ref{model-rs-def}) with jump rate and jump size distribution as in Proposition \ref{integral_eqn_gen}. Let $Q(x,\rho,u)$ and $p(x',x,s)$ be the survival probability and transition density, respectively, of the continuous regime switching OU process. Then, the survival probability $\Phi(x,\rho,u)$ satisfies the integral equation:
		\begin{flalign}\label{int-rs}
			\Phi(x,\rho, u)  = \int_{0}^{u} \lambda e^{-\lambda s}  \int_{0}^{\infty} \int_{\mathbb{R}} \mathcal{T}_s\Phi(x'+z, \rho , u) p(x',x,s) dF(z)dx'ds + e^{-\lambda u}Q(x,\rho, u).
		\end{flalign}
		
		\item[(ii)] Consider the asset value process under the stochastic volatility model (\ref{model-sv-def}) with jump rate and jump size distribution as in Proposition \ref{integral_eqn_gen}. Let $Q(x,y,u)$ and $p(x',x,s)$ be the survival probability and transition density, respectively, of the continuous stochastic volatility OU process. Then, the survival probability $\Phi(x,y,u)$ satisfies the integral equation:
		\begin{flalign}\label{int-sv}
			\Phi(x,y, u)=  \int_{0}^{u} \lambda e^{-\lambda s}  \int_{0}^{\infty} \int_{\mathbb{R}} \mathcal{T}_s\Phi(x'+z, y , u) p(x',x,s) dF(z)dx'ds + e^{-\lambda u}Q(x,y, u).
		\end{flalign}
	\end{itemize}	
\end{corollary}

\begin{remark}
	The proof above is based on the approach considered in \cite{mishura2016ruin} and constitutes an extension for processes where the diffusion term in non-zero. This results in having to consider appropriate stopping times, as well as the transition density of the OU process in the representation above.  Furthermore, it is worth mentioning the simple case of the non-jump OU process, the survival probability $\Phi(x,u)$ satisfies the integral equation:
	\begin{flalign}\label{simple_eqn}
		\Phi(x,u)  =  \int_{0}^{u} \lambda e^{-\lambda s}  \int_{0}^{\infty} \int_{\mathbb{R}} \Phi(x'+z,u-s) p(x',x,s) dF(z)dx'ds + e^{-\lambda u}Q(x,u),
	\end{flalign}
	which can be easily derived from the regime switching model by considering a single regime. This can be directly compared to the corresponding integral equation in \cite{mishura2016ruin}, where similar integral equations are used to study the ruin probability function for an analogous asset value process. On the other hand, in this work, the integral equations can lead to the existence and continuity results that suffice to obtain approximations of the PD functions, as we will see below. 
\end{remark}

\begin{remark}
	In the following sections we will often interchange between using the $s-$operator formulation and more analytical expressions in terms of an appropriate expected value. In particular, applying the law of total probability, the integral equation for the stochastic volatility model can be written in more detail as:
	\begin{flalign}\label{int-sv-exp}
		\Phi(x,y, u)  =  \int_{0}^{u} \lambda e^{-\lambda s}  \int_{0}^{\infty} \int_{0}^{\infty} \int_{\mathbb{R}} \Phi(x'+z, \nu , u-s) q(\nu , y,s ) p(x',x,s) dF(z)dx'd\nu ds + e^{-\lambda u}Q(x,y, u),
	\end{flalign}
	where $q(\nu , y, s)$ is the transition probability of the CIR volatility process and we have used that, by definition:
	$$ \mathcal{T}_s \Phi(x'+z, y, u) = \int_{0}^{\infty}  \Phi(x'+z, \nu , u-s) q(\nu , y,s ) d\nu.$$
	Similarly, for the generalized model, we can write:
	\begin{flalign}\label{int-gen}
		\Phi(x,&\rho, y, u) \notag\\ = &  \int_{0}^{u} \lambda e^{-\lambda s}  \int_{0}^{\infty} \int_{0}^{\infty} \int_{\mathbb{R}} \mathbb{E}[\Phi(x'+z, R_s^{\rho},\nu , u-s)] q(\nu , y, s) p(x',x,s) dF(z)dx'd\nu ds  + e^{-\lambda u}Q(x,\rho, y, u).
	\end{flalign}
	These analytical versions will be useful when using the integral equations to derive PIDEs for the PD processes, whereas the more concise form will be used when proving the required continuity results that follow.
\end{remark}

\begin{remark}
	Finally, it is also worth noting that, in obtaining the integral equations in this section, we have included the initial values of the regime and volatility processes as additional variables on which the PD function depends. An alternative approach is writing these as a system of integral equations. For example, considering the regime switching model with states in accordance to the IFRS 9 framework (i.e., Stage 1, 2 and 3) we can rewrite $\Phi_i(x,u) := \Phi(x, \rho^i,u)$ (and similarly $Q_i(x,u) := Q(x, \rho^i,u)$),  which results in the system of equations: 
	\begin{flalign} \notag
		\Phi_i(x, u) =  \int_{0}^{u} \lambda e^{-\lambda s} \int_{0}^{\infty} \int_{\mathbb{R}} \sum_{j=1}^{3} \Phi_j(x, u-s) \pi(R_j, \rho^i, s) p(x',x,s) dF(z)dx'ds + e^{-\lambda u}Q_1(x, u),
	\end{flalign}
	for $i= 1,2,3$, and where we have used that, in this case, $\mathcal{T}_s \Phi(x, \rho, u-s) = \sum_{j=1}^{3} \Phi(x,R_j, u-s) \pi(R_j, \rho, s)$, by definition. In our approach, we prefer to account for these externals parameters explicitly, and generalize this representation in the cases of the additional variables.	
\end{remark}

\subsection{Properties and existence of solutions}
\par Using the IEs derived above, we can now prove that the PD functions enjoy certain mathematical properties required to obtain the viscosity solutions to the corresponding PIDEs.
\par To consider such solutions, we require $\Psi$ (equivalently $\Phi$) to be a continuous function of $(x,u)$. To this end, we first note that $\Psi$ ($\Phi$) is a monotonically decreasing (increasing) function with respect to $x$ and a monotonically increasing (decreasing) function with respect to maturity $u$. Moreover, as it is bounded, we can conclude that it is an integrable function. We can take advantage of the integral equation forms to prove that solutions for the survival probabilities do exist, and moreover, that they are continuous with respect to $x$ and $u$. We prove this result for the generalized model, from which the other cases follow easily. 
\begin{remark}
	As mentioned, in what follows we will focus on the survival probability as a function of $(x,u)$, as the additional complexity we are interested in arises from the jump component of the OU process. It is straightforward to reproduce the proofs for the stochastic volatility variable, as well. Finally, the existence of the regime switching process simply creates a coupling that does not affect the results in this section.
\end{remark}

\begin{lemma}\label{cont-res}
	The probability of survival function under the generalized model $\Phi(x,\rho, y, u)$, as defined in (\ref{gen-surv}) is uniformly continuous as a function of $x$ and $u$.
\end{lemma}

\begin{proof}
	To prove this result we will use the integral formulation $(\ref{int-gen-prev})$. Consider a fixed $\epsilon > 0 $ and $||(x,u) - (x_0, u_0)|| < \delta$, for some $\delta > 0$ that will be specified. Then:
	$$
	\begin{aligned}
		\left|\Phi(x, \rho, y, u)-\Phi(x_0, \rho, y, u_{0})\right| 
		 \leq \left|\Phi(x, \rho, y, u)-\Phi(x_0, \rho, y, u)\right| + \left|\Phi(x_0, \rho, y, u)-\Phi(x_0, \rho, y, u_0)\right|
	\end{aligned} $$
	We will handle each of the terms above separately. We have: 
	\begin{itemize}
		\item[$x:$] Let $|x-x_0| < \delta_1$. Since the transition density of the non-jump OU process is uniformly continuous, we can select $\delta_1$ such that $\left|p(x', x, s)-p\left(x', x_{0}, s\right)\right|< \epsilon/2$. Then:
		$$\begin{aligned}
			&|\Phi(x,\rho, y, u) - \Phi(x_0,\rho, y, u)| \notag\\\leq&  \int_{0}^{u}\lambda e^{-\lambda s} \int_{ \mathbb{R}} \int_{ \mathbb{R}} \mathcal{T}_s \Phi(x'+z, \rho, y, u) |p(x',x,s)- p(x',x_0,s)|dF(z)dx'ds + e^{-\lambda u}|Q(x,u) - Q(x_0,u)|  \notag \\
			\leq& \int_{0}^{u} \lambda e^{-\lambda s} \int_{ \mathbb{R}} \int_{ \mathbb{R}} \mathcal{T}_s \Phi(x'+z, \rho, y, u) \frac{\epsilon}{2} dF(z)dx'ds + e^{-\lambda u} \frac{\epsilon}{2}
		\leq \int_{0}^{u} \lambda e^{-\lambda s} \frac{\epsilon}{2} ds + e^{-\lambda u}\frac{\epsilon}{2} = \frac{\epsilon}{2}.
		\end{aligned} $$	
		\item[$u:$] Similarly, let $|u-u_0| < \delta_2$. Then:
		$$
		\begin{aligned}
			&\left|\Phi(x_0, \rho, y, u)-\Phi\left(x_0, \rho, y, u_{0}\right)\right| \leq \\
			\leq & \int_{u_0}^{u} \lambda e^{-\lambda s} \int_{\mathbb{R}} \int_{\mathbb{R}}\left|\mathcal{T}_{s} \Phi\left(x^{\prime}+z, \rho, y, u\right)-\mathcal{T}_{s} \Phi\left(x^{\prime}+z, \rho, y, u_{0}\right)\right| 
			p(x', x, s) d F(z) d x^{\prime} ds.
		\end{aligned} $$
		By definition, we have $\left\|\mathcal{T}_{s}\right\| \leq 1$ and therefore:
		$\left|\Phi(x, \rho, y, u)-\Phi\left(x, \rho,  y, u_{0}\right)\right| \leq \int^{u}_{u_0} \lambda e^{-\lambda s} d s \leq \left(u-u_{0}\right) \max_{\substack{ s \in [0,u]}} \lambda e^{-\lambda s}=\lambda\left(u-u_{0}\right)<\lambda \delta_2$.  
	\end{itemize}
	
	By selecting $\delta_2 = \frac{\epsilon}{2\lambda}$ and $\delta = \delta_1 \wedge \delta_2$ we therefore obtain:
	$$ |\Phi(x, \rho, y, u)-\Phi(x_0, \rho, y, u_{0})| < \epsilon,$$ as required.

\end{proof}	 

We can now consider an appropriate fixed point result which will allow us to prove the existence of a solution to the IE for the survival probability. For this result we will refer to the Arzela-Ascoli and Schauder's fixed point theorems, as stated in \ref{arzela-ascoli} and \ref{schauder-state} of the Appendix, respectively.

\begin{proposition}\label{schauder}
	The integral equation $(\ref{int-gen-prev})$ admits a continuous solution.
\end{proposition}
\begin{proof}
	Consider the metric space of all continuous functions $\Phi(\cdot,\rho,y, \cdot)$ on $\mathcal{D} \times [0,T]$, denoted by $X := C(\mathcal{D} \times [0,T])$. Furthermore, define the functional operator $\mathcal{A}:X \rightarrow X$ by:
	\begin{eqnarray}
		\mathcal{A} \Phi(x,\rho,y,u) = \int_{ \mathbb{R}} \int_{ \mathbb{R}} \mathcal{T}_s \Phi(x'+z, \rho, y, u) p(x',x,s)dF(z)dx'.
	\end{eqnarray}
	We begin by proving that the operator $\mathcal{A}$ is: $(i)$ uniformly bounded, $(ii)$ equi-continuous and $(iii)$ compact. We separate these in the steps below:
	
	\begin{itemize}	
		\item[$(i)$] Uniform boundedness follows easily from the definition of $\Phi$ and the operator $\mathcal{T}_s$. Specifically:
		$$ |\mathcal{A}\Phi(x,\rho, y, u)| \leq 1.$$
		\item[$(ii)$] For equi-continuity we must show that, given $\epsilon > 0$, there exists $\delta >0$ such that if $||(x,u)-(x_0,u_0)|| < \delta$ then $|\mathcal{A}\Phi(x,\rho, y, u) - \mathcal{A}\Phi(x_0,\rho, y, u_0)| < \epsilon$, for all $\Phi \in X$. The proof follows very closely Lemma \ref{cont-res}, but we include the steps for completeness.
		
		To this end, we calculate:
		$$\begin{aligned}
			&|\mathcal{A}\Phi(x,\rho, y, u) - \mathcal{A}\Phi(x_0,\rho, y, u_0)| \notag \\& \leq  \int_{ \mathbb{R}} \int_{ \mathbb{R}} |\mathcal{T}_s \Phi(x'+z, \rho, y,u)p(x',x,s) - \mathcal{T}_s \Phi(x'+z, \rho, y,u_0)p(x',x_0,s)| dF(z)dx'.
		\end{aligned} $$
		Consider the integrand. We can write:
		\begin{flalign}
			 \mathcal{T}_s \Phi(x'+z, \rho, y,u) \Big(p(x',x,s) - p(x',x_0,s)\Big) + p(x',x,s) \Big( \mathcal{T}_s \Phi(x'+z, \rho, y,u) - \mathcal{T}_s \Phi(x'+z, \rho, y,u_0) \Big), \notag
		\end{flalign}
		and therefore:
		$$\begin{aligned}
			&|\mathcal{A}\Phi(x,\rho, y, u) - \mathcal{A}\Phi(x_0,\rho, y, u_0)| \leq 
			\int_{ \mathbb{R}}\int_{ \mathbb{R}} \mathcal{T}_s \Phi(x'+z, \rho, y, u) |p(x',x,s) - p(x',x_0,s)| \notag \\ &+ p(x',x_0,s) |\mathcal{T}_s \Phi(x'+z, \rho, y,u) - \mathcal{T}_s \Phi(x'+z, \rho, y,u_0)| F(z)dx'.
		\end{aligned}$$
		We know that $p(x',x,s)$ is uniformly continuous in $x$ and $\Phi$ is uniformly continuous in $u$. Therefore, given $\epsilon > 0$ we can select $\delta_1$ such that
		$$ |p(x',x,s) - p(x',x_0,s)| < \epsilon/2,$$ for all $x,x_0$ such that $|x-x_0| < \delta_1$. Furthermore, from Lemma \ref{cont-res} recall that, for $u,u_0$ such that $|u-u_0| < \delta_2$ we have: 
		$$|\Phi(x'+z, \rho, y, u) - \Phi(x'+z, \rho, y, u_0)| < \lambda \delta_2,$$ for all $\Phi \in X$. Select $\delta_2 = \frac{\epsilon}{2\lambda}$ and let $\delta = \delta_1 \wedge \delta_2$. Hence:
		\begin{flalign}
			|\mathcal{A}\Phi(x,\rho, y, u) &- \mathcal{A}\Phi(x_0,\rho, y, u_0)| \notag \\ \leq &\int_{ \mathbb{R}} \int_{ \mathbb{R}} \mathcal{T}_s \Phi(x'+z, \rho, y, u) \frac{\epsilon}{2} + p(x',x_0,s)\lambda \delta dF(z)dx' \leq \epsilon.
		\end{flalign}
		Notice that the choice of $\delta$ does not depend on $\Phi$, only on the given $\epsilon$. Hence, $|\mathcal{A}\Phi(x,\rho, y, u) - \mathcal{A}\Phi(x_0,\rho, y, u_0)| < \epsilon$, for all $\Phi \in X$ whenever $||(x,u) - (x_0,u_0)|| < \delta$, i.e., $\mathcal{A}$ is equi-continuous.
		\item[$(iii)$] We have that $\mathcal{A}\Phi$ is uniformly bounded and equi-continuous and therefore compact by the Arzela-Ascoli theorem.
	\end{itemize}
	With the above, we now turn to the main result. Consider now the Banach space $\mathcal{C} = \{\phi \in X, ||\phi|| \leq 1 \}$, where we have used the standard norm $$||\phi|| := \sup_{\substack{x,u}} |\phi|.$$
	With the operator $\mathcal{A}$ defined as above, we have: 
	$$ \Phi(x,\rho, y, u) = \int_{0}^{u} \lambda e^{-\lambda s} \mathcal{A}\Phi(x,\rho, y, u) ds + g(x,u),$$ with $g(x,u) := e^{-\lambda u} Q(x,\rho, y, u)$ and it is natural to define the operator $\mathcal{P}\Phi$, such that:
	\begin{eqnarray}
		\mathcal{P}\Phi(x, \rho, y, u) =\int_{0}^{u} \lambda e^{-\lambda s} \mathcal{A}\Phi(x,\rho,y, u)ds + g(x,u).
	\end{eqnarray}		
	Recall that $\mathcal{A}\Phi$ is bounded by 1 and, furthermore, by the definition of the operator, we also have that:
	$$ ||\mathcal{A}\Phi || \leq ||\Phi||.$$ Hence:
	\begin{flalign}
		||\mathcal{P}\Phi|| &\leq ||\int_{0}^{u} \lambda e^{-\lambda s} \mathcal{A}\Phi(x,\rho, y, u) ds|| + ||g||
		\leq \sup_{\substack{x,u}}  \int_{0}^{u} \lambda e^{-\lambda s} |\mathcal{A}\Phi(x,\rho, y, u)| ds + ||g||\notag \\
		& \leq \sup_{\substack{x,u}}  (1-e^{-\lambda u})|\Phi|+ \sup_{\substack{x,u}} e^{-\lambda u}|Q(x,\rho, y, u)|
		\leq \max(||\Phi||,  ||g||) \leq 1,
	\end{flalign}
	concluding that $\mathcal{P}$ maps function $\Phi \in X$ to $X$. 
	Notice that we can write $\mathcal{P} = \mathcal{J}\mathcal{A}$, with $\mathcal{J}\phi = \int_{0}^{u} \lambda e^{-\lambda s} \phi(s) ds$. It is straightforward to see that the linear operator $\mathcal{J}$ is compact, as is $\mathcal{A}$, as shown in Lemma \ref{cont-res}. Therefore, we can apply Schauder's fixed point theorem to conclude that $\mathcal{P}$ has a fixed point in $X$, which solves the integral equation $(\ref{int-gen-prev})$.
\end{proof}

\section{Partial Integro-Differential Equations for the PD function}\label{pides-visc}

%In the literature, it is standard to consider the case where the maturity is fixed and $\Psi,\Phi$ are functions of the starting time. This quantity can be considered for forecasting, as well as estimating future Lifetime provisions. 
Generally, when considering various credit modelling tasks, such as forecasting probability of default and expected losses, it is common throughout the literature to use path simulation techniques, particularly for practical purposes. Specific examples and applications can be seen in e.g., \cite{sak2012fast} and \cite{virolainen2004macro}. On the other hand, we will see in this section that the integral equation representations (\ref{int-rs}), (\ref{int-sv}) and (\ref{int-gen}) lead to PIDEs for the PD functions, which belong to families of well-studied equations. Hence, our approach relies solely on the equations derived in this and the previous section, which can be solved to retrieve the corresponding values, thereby eliminating the need for simulations and the larger errors which accompany such methods. 
\par Natural questions arise related to the regularity conditions that the survival function (and hence the corresponding PD function) must satisfy; for example, in the one dimensional L\'evy-driven OU case, classical solutions of PIDEs would require that $\Phi(x,u) \in \mathcal{C}^{2,1}\big(\mathcal{D} \times [0,T]\big)$, where $\mathcal{D}:= [0,\infty)$, i.e., the survival probability function would have to be twice and once continuously differentiable in $x$ and $u$, respectively, on the corresponding domains. In many cases, the required differentiability conditions are often assumed. However, we can avoid making such assumptions by considering viscosity solutions of the PIDEs, a notion introduced in \cite{crandall1983viscosity}. Viscosity solutions and their applications in finance have been studied in e.g., \cite{cont2005integro} and \cite{cont2005finite}. We begin by showing that the generalized PD function is a viscosity solution of a PIDE that will be derived. Then, we continue by showing that this and the PIDEs that result from the regime switching and stochastic volatility models can be derived directly from the corresponding integral equations, if the required regularity conditions hold. In our setting, it is understood that the survival functions are in fact viscosity solutions to these PIDEs, yet these calculations are important for two reasons: firstly, they establish the connection between the integral equations and the corresponding PIDEs, and secondly they constitute an efficient method of obtaining the form of the PIDEs, for which we can then show that the survival functions are indeed viscosity solutions.    
\par Finally, it is worth emphasizing the utility of the PIDEs obtained in this section. Specifically, the solutions to these equations are PD values across both initial positions, time and latent variables. Hence, we will see that, by considering numerical schemes for the PIDES, we obtain the complete evolution of the PD process that is required for applications in credit risk modelling. This is clearly preferable to common methods such as Monte Carlo estimations of the PD, where one must perform simulations of the underlying asset process for many different initial positions and time horizons, separately. This process is extremely computationally costly, especially when taking into account the order of convergence of many stochastic simulation schemes. For example, the Euler scheme for the simulation of the asset process we consider has a strong convergence of order 0.5, which is required since the PD value depends on the whole path of the asset process. 

\subsection{Viscosity solutions}
Viscosity solutions for non-local PDEs have also been studied in e.g., \cite{barles2008second}, \cite{cont2005integro} and \cite{hamadene2016viscosity} and references therein. We reiterate the importance of this approach: requiring only continuity of the underlying function, which has been proven for the survival function under the proposed models, we can define solutions of the equations in a weak sense and subsequently approximate them using numerical schemes. We begin by showing that the generalized survival probability is a viscosity solution of an appropriate PIDE (provided in equation (\ref{gen-visc}) below). For completeness, we include the definition of a viscosity solution below (altered to reflect the arguments of the survival function we are studying). 

\begin{definition}
	Consider an integro-differential operator for the function with arguments as above, $\mathcal{L}f(x,\rho,y,u)$ and a corresponding PIDE $\mathcal{L}f(x,\rho,y,u) =0$. Then, a function $\phi(x,\rho,y,u)$ is called a viscosity supersolution (subsolution) of the PIDE if, for any $\rho \in \mathcal{R}$, for every $(x, y, u) \in \mathcal{D} \times \mathcal{V} \times  [0, T]$, and every function $f(\cdot, \rho, \cdot, \cdot) \in \mathcal{C}^{2,1}\big(\mathcal{\tilde{D}} \times [0,T]\big)$, where $\tilde{D}:= \mathcal{D} \times \mathcal{V}$, such that $\phi(x,\rho,y,u)=f(x,\rho, y,u)$ and $\phi \geq f$ ($\phi \leq f$), the inequality $\mathcal{L} f(x,\rho,y,u) \leq 0$ ($\mathcal{L} f(x,\rho,y,u) \geq 0$) holds. A function $\phi(x,\rho,y, u)$ is a viscosity solution of the PIDE if $\phi$ is simultaneously a viscosity supersolution and subsolution.
\end{definition}

In what follows we will show that the survival function is a viscosity solution of an appropriate PIDE given below. We must first show that viscosity solutions do exist for the PIDEs in question. The proof that follows is based on the corresponding result in \cite{belkina2016viscosity}, extended to match the models studied in this work. 

\begin{proposition}
	The survival probability function $\Phi(x,\rho,y,u)$ is a viscosity solution of the PIDE $\mathcal{L}f(x,\rho,y,u) =0$, where:
	\begin{flalign} \label{gen-visc}
		\mathcal{L}f(x,&\rho,y,u) :=-\frac{\partial f}{\partial u}(x,\rho,y,u)+k_\rho(\theta_{\rho}-x) \frac{\partial f}{\partial x}(x,\rho,y,u) + \kappa (\mu -y)\frac{\partial f}{\partial y}(x,\rho, y,u) \notag \\& +\frac{1}{2} \sigma_\rho^2 y\frac{\partial^{2} f}{\partial x^{2}}(x,\rho,y,u) +\frac{1}{2} \xi^2 y \frac{\partial^{2} f}{\partial y^{2}}(x,\rho, y,u) +\sum_{j \neq \rho} q_{\rho j} \Big(f(x,j,y,u) - f(x,\rho,y,u)\Big) \notag \\ & + \int_{ \mathbb{R}}\Big(f(x+z,\rho,y, u) - f(x,\rho,y,u) \Big)\nu(dz). 
	\end{flalign}
\end{proposition}
\begin{proof}
	We must show that $\Phi(x,\rho,y,u)$ is simultaneously a viscosity supersolution and subsolution of the PIDE. 
	
	We begin with the supersolution case. For any $\rho \in \mathcal{R}$, consider fixed $(x,y,u) \in \mathcal{D} \times \mathcal{V} \times [0,T]$, with $\Phi(x,\rho, y,u) = 0 $ when $x\leq 0$, by definition. Furthermore, consider a function $f(\cdot,\rho ,\cdot, \cdot ) \in \mathcal{C}^{2,1}\big(\tilde{D} \times [0,T]\big)$ such that $\Phi(x,\rho,y, u)=f(x,\rho,y, u)$ and $\Phi((\cdot,\rho ,\cdot, \cdot )) \leq f(\cdot,\rho ,\cdot, \cdot )$ on $\mathcal{D} \times \mathcal{V} \times  [0,T]$. Now, let $h>0$ and let $\epsilon_x, \epsilon_y$ be small enough to ensure that $f \in \mathcal{C}^{2,1}\Big(  \tilde{d}_{\epsilon} \times [0,T]  \Big)$, where $\tilde{d}_{\epsilon} :=B_{\epsilon_x}(x) \times B_{\epsilon_y}(y)$, i.e., we are considering a neighborhood of the fixed $(x,u)$ where the functions will "touch". Finally, define the stopping time $\tau_h := \inf\{ t \geq 0: (G_t^x,Y_t^x) \notin (\overline{B_{\epsilon_x}(x)} \times \overline{B_{\epsilon_y}(y)}) \} \wedge h$, noting that we choose $h<u$ to ensure that $\tau_h < u$. Then, by the It\^o formula, and with the operator $\mathcal{L}$ is in (\ref{gen-visc}), we have: 
	\begin{flalign}
		f(G_{\tau_h}^x, R^{\rho}_{\tau_h}&,Y_{\tau_h}^y, u - \tau_h) - f(G_0^x,R^{\rho}_{0},Y_0^{y},u) = f(G_{\tau_h}^x, R_{\tau_h}^{\rho}, Y_{\tau_h}^{y},u - \tau_h) - f(x,\rho, y,u) \notag \\ &= \int_{0}^{\tau_h} \mathcal{A}f(G_{t}^x, R_{t}^{\rho}, Y_{t}^y,u - t) dt + \int_{0}^{\tau_h} \sigma \frac{\partial f}{\partial x} (G_{t}^x, R_{t}^{\rho}, Y_{t}^y,u - t) dB_t \notag \\ &+ \int_{0}^{\tau_h} \int_{\mathbb{R}} \Big( f(G_{t}^x+z, R_{t}^{\rho}, Y_{t}^y,u - t) - f(G_{t}^x, R_{t}^{\rho}, Y_{t}^y,u - t) \Big) \tilde{N}(dt,dz) \notag \\& + \sum_{j \neq \rho} p_{\rho j}(t_h) \Big(f(G_{\tau_h}^x,j,Y_{\tau_h}^{y},u-\tau_h) - f(G_{\tau_h}^x,\rho,Y_{\tau_h}^{y},u-\tau_h)\Big)\notag \\ &= \int_{0}^{\tau_h} \mathcal{A}f(G_{t}^x, R_{t}^{\rho}, Y_{t}^y,u - t) dt  + m_t,
	\end{flalign}
	where 
	\begin{flalign} 
		m_t :=& \int_{0}^{\tau_h} \sigma \frac{\partial f}{\partial x} (G_{t}^x, R_{t}^{\rho}, Y_{t}^y,u - t) dB_t \notag\\ &+ \int_{0}^{\tau_h} \int_{\mathbb{R}} \Big( f(G_{t}^x+z, R_{t}^{\rho}, Y_{t}^y,u - t) - f(G_{t}^x, R_{t}^{\rho}, Y_{t}^y,u - t) \Big) \tilde{N}(dt,dz)
	\end{flalign}
	is a martingale and therefore so is the stopped process $m_{t\wedge \tau_h}$.
	
	We have that $f(x,\rho, y,u) = \Phi(x,\rho, y,u)$ and recall that $\Phi(x,\rho,y,u) = \mathbb{E}[\Phi(G^{x}_{\tau_h}, R_{\tau_h}^{\rho}, Y_{\tau_h}^y ,u-\tau_h)]$ almost surely, from Proposition \ref{integral_eqn_gen}. Therefore: 
	\begin{flalign}
		&\Phi(G_{\tau_h}^x, R^{\rho}_{\tau_h},Y_{\tau_h}^y, u - \tau_h) \geq \notag \\ & f(G_{\tau_h}^x, R^{\rho}_{\tau_h},Y_{\tau_h}^y, u - \tau_h) = f(x,\rho, y, u)  + \int_{0}^{\tau_h} \mathcal{A}f(G_{t}^x, R^{\rho}_{t},Y_{t}^y, u - t) dt  + m_t, 
	\end{flalign}
	and so:
	\begin{flalign}
		\mathbb{E}[\Phi(G_{\tau_h}^x,& R^{\rho}_{\tau_h},Y_{\tau_h}^y, u - \tau_h )] \geq\mathbb{E}[\Phi(x,\rho, y, u)]  + \mathbb{E}\Big[\int_{0}^{\tau_h} \mathcal{A}f(G_{t}^x, R^{\rho}_{t},Y_{t}^y, u - t) dt \Big] + \mathbb{E}[m_t] \notag \\
		&\Rightarrow \Phi(x,\rho ,y, ,u) \geq   \Phi(x,\rho, y,u) + \mathbb{E}\Big[\int_{0}^{\tau_h} \mathcal{A}f(G_{t}^x, R^{\rho}_{t},Y_{t}^y, u - t) dt \Big],
	\end{flalign}
	and hence $\mathbb{E}\Big[\int_{0}^{\tau_h} \mathcal{A}f(G_{t}^x, R^{\rho}_{t},Y_{t}^y, u - t) dt \Big] \leq 0$. The final step is as in \cite{belkina2016viscosity}; when $h$ is sufficiently small we have $\tau_h = h$ and therefore, by the Lebesgue dominated convergence theorem we have: 
	\begin{eqnarray}
		\mathcal{A}f(x,\rho,y, u) = \lim_{h \downarrow 0} \frac{1}{h} \mathbb{E}\Big[\int_{0}^{\tau_h} \mathcal{A}f(G_{t}^x, R^{\rho}_{t},Y_{t}^y, u - t) dt \Big] \leq 0,
	\end{eqnarray}
	showing that $\Phi$ is indeed a viscosity supersolution. It follows directly, by switching the inequalities in the steps above, that $\Phi$ is a viscosity subsolution and therefore the result is proven.
\end{proof}

\begin{remark}
	Viscosity theory is not the only prism under which we can consider weak solutions. We can also study solutions in appropriate Sobolev spaces, for which we will need to define a notion of weak differentiability, and we can then use standard martingale approaches to obtain the corresponding PIDEs. We include details of this approach in Appendix \ref{sob-app}.
\end{remark}

\begin{remark}
	We furthermore note that if additional conditions hold, then the strong solutions that occur are equal to the viscosity solutions, as expected. The formulations via viscosity solutions are useful to obtain the form of the PIDEs the PD functions satisfy, and we will see that we can then consider additional conditions that lead to regular solutions through these equations.   
\end{remark}

\begin{comment}
\begin{lemma}
The survival probability functions $\Phi(x,u)$ and $\Phi(x,y,u)$ are viscosity solutions of the PIDEs $(\ref{pide2})$ and $(\ref{pide-sv})$, respectively.
\end{lemma}
\begin{proof}
The proofs are as above, with the generator processes given by: 
\begin{eqnarray}
\mathcal{L}f(x,u) = -\frac{\partial f}{\partial u}+\frac{1}{2} \sigma^{2} \frac{\partial^{2} f}{\partial x^{2}}+k(\theta-x) \frac{\partial f}{\partial x}   +\int_{z \in \mathbb{R}}\Big(f(x+z,u)-f(x,u)\Big) \nu(dz),
\end{eqnarray} and
\begin{eqnarray}
\mathcal{L}f(x,y,u)	-\frac{\partial f}{\partial u} +k(\theta - x) \frac{\partial f}{\partial x} + \kappa (\mu -y)\frac{\partial f}{\partial y} + \frac{1}{2} y \frac{\partial^{2} f}{\partial x^{2}} + \frac{1}{2} \xi^2 y \frac{\partial^{2} f}{\partial y^{2}}  \nonumber \\+ \lambda \int_{z \in \mathbb{R}} \Big(f(x+z,y, u) - f(x,y,u) \Big) \nu(dz),
\end{eqnarray}
respectively.
\end{proof}
\end{comment}

\subsection{The survival probability as a classical solution of PIDEs derived from the IE formulations}
To build a consistent framework we must ensure that that the PIDEs obtained above (which are satisfied by the PD functions in the viscosity sense) are derivable from the integral equations formulation in section \ref{integral-section}. In this section, we show that the integral equations indeed lead to the corresponding PIDEs.
We will begin with the calculations under the regime switching and stochastic volatility models, upon which the corresponding result under the generalized model will be built. In all the results below, we consider a fixed time horizon $T>0$. 

\begin{lemma} \label{pides-1} \hfill
	\begin{itemize}
		\item[$(i)$] Under the regime-switching model (\ref{model-rs-def}), the survival probability $\Phi(x,\rho,u)$ satisfies the PIDE:
		\begin{flalign} \label{int-markov}
			&\frac{\partial \Phi}{\partial u}(x,\rho,u)=k_\rho(\theta_{\rho}-x) \frac{\partial \Phi}{\partial x}(x,\rho,u) + \frac{1}{2} \sigma_\rho^2 \frac{\partial^{2} \Phi}{\partial x^{2}}(x,\rho,u)+ \sum_{j \neq \rho} q_{\rho j} \Big(\Phi(x,j,u) - \Phi(x,\rho,u)\Big)  \notag\\ 
			&+ \int_{ \mathbb{R}}\Big(\Phi(x+z,\rho,u) -  \Phi(x,\rho,u) \Big)\nu(dz), \,\,\ (x,\rho, u) \in \mathcal{D} \times \mathcal{R} \times [0,T],
		\end{flalign}
		with initial and boundary conditions: 
		\begin{align} \nonumber
			\Phi(x,\rho,0)= \mathbbm{1}_{\{x > 0\}},\,\,\ (x,\rho) \in \mathcal{D} \times \mathcal{R},\\ \nonumber
			\Phi(0,\rho,u) = 0, \,\,\ (\rho,u) \in  \mathcal{R} \times [0,T], \\ \nonumber
			\Phi(x,\rho, u) \rightarrow 1 \text{ as } x \rightarrow \infty, \,\,\ (\rho,u ) \in \mathcal{R} \times [0,T], 
		\end{align} 
		where $q_{ij}$ are the elements of the generator $Q$ of the switching pocess $R_t$, as defined in (\ref{markov-gen-matrix}). 
		
		\item[$(ii)$] 	Under the stochastic volatility model (\ref{model-sv-def}), the survival probability $\Phi(x,y,u)$ satisfies the PIDE:
		\begin{flalign} \label{pide-sv}
			\frac{\partial \Phi}{\partial u}(x,y,u)&  = k(\theta - x) \frac{\partial \Phi}{\partial x}(x,y,u) + \kappa (\mu -y)\frac{\partial \Phi}{\partial y}(x,y,u) + \frac{1}{2} y \frac{\partial^{2} \Phi}{\partial x^{2}}(x,y,u) + \frac{1}{2} \xi^2 y \frac{\partial^{2} \Phi}{\partial y^{2}}(x,y,u) \nonumber \\&+ \int_{\mathbb{R}} \Big(\Phi(x+z,y, u) - \Phi(x,y,u) \Big) \nu(dz), \,\,\ (x,y,u) \in \mathcal{D}\times \mathcal{V} \times [0,T],
		\end{flalign}
		subject to the initial and boundary conditions: 
		\begin{align}
			\Phi(x,y,0)= \mathbbm{1}_{\{x > 0\}},\,\,\ (x,y) \in \mathcal{D} \times \mathcal{V}, \notag\\
			\Phi(0,y,u) = 0, \,\,\ (y,u) \in \mathcal{V} \times [0,T], \notag \\
			\Phi(x,y, u) \rightarrow 1 \text{ as } x \rightarrow \infty, \,\,\ (y,u) \in \mathcal{V} \times [0,T], \notag \\
			\frac{\partial \Phi}{\partial y}(x,y,u) = 0 \text{ as } y \rightarrow \infty \,\,\ (x,u) \in \mathcal{D} \times [0,T].
		\end{align} 
		%Furthermore, on the boundaries of $y$ we solve $(\ref{pide2})$, with $\sigma^2 = y_{min}, y_{max}$, respectively.
	\end{itemize}
	
\end{lemma}
\begin{comment}	\begin{proof}
The proof follows directly from the proof of Proposition \ref{equiv-int}. Since the PD process remains time-homogeneous, the martingale argument still holds. The generator of $\Phi(x,u,\rho)$ is given by the expression below (see \cite{zhu2015feynman} for further details):

\begin{eqnarray}
\mathcal{A}\Phi(x,u,\rho)	=-\frac{\partial \Phi}{\partial u}+\frac{1}{2} \sigma_\rho^2 \frac{\partial^{2} \Phi}{\partial x^{2}}+k_\rho(\theta_\rho-x) \frac{\partial \Phi}{\partial x} + \sum_{j \neq \rho} q_{\rho j} \Big(\Phi(x,u,j) - \Phi(x,u,\rho)\Big)  \notag \\
+ \int_{z \in \mathbb{R}}\Big(\Phi(x+z,u,\rho)- \Phi(x,u,\rho)\Big)\nu(dz),
\end{eqnarray}
Employing the martingale argument and splitting the integral, as in the standard model, gives (\ref{int-markov}).
\end{proof}
\end{comment}
\begin{proof}
	\begin{itemize}
		\item[$(i)$] Our approach relies on taking advantage of the known fact that the transition densities satisfy the Kolmogorov equation. Specifically, we know that for $p(\cdot,x,u)$ we have:
		\begin{eqnarray}\label{kolm-3}
			\frac{\partial p}{\partial u}(\cdot, x,u) = k_{\rho}(\theta_{\rho} - x)\frac{\partial p}{\partial x} (\cdot,x, u) + \frac{1}{2}\sigma_{\rho}^2 \frac{\partial^2 p}{\partial x^2}(\cdot,x, u).
		\end{eqnarray}
		Recall that the same holds for the survival distribution of the continuous OU, $Q(x,\rho,u)$:
		\begin{eqnarray}\label{kolm-4}
			\frac{\partial Q}{\partial u}(x,\rho, u) = \mathcal{L}_1 Q(x,\rho,u),
		\end{eqnarray}
		with the generator operator under the regime switching model $\mathcal{L}_1$ given by:
		\begin{flalign}
			\mathcal{L}_1Q(x,\rho, t):= k_\rho(\theta_{\rho}-x) \frac{\partial f}{\partial x}(x,\rho,t) + \frac{1}{2} \sigma_\rho^2 \frac{\partial^{2} f}{\partial x^{2}}(x,\rho,t)+ \sum_{j \neq \rho} q_{\rho j} \Big(Q(x,j,t) - f(x,\rho,t)\Big). 
		\end{flalign}

		%where the operator $\mathcal{L}$ is defined as:
		%\begin{eqnarray}
		%\mathcal{L}f(x,t):= k(\theta - x)\frac{\partial f}{\partial x}(x,t) + \frac{1}{2} \sigma^2  \frac{\partial^2 f}{\partial x^2}(x,t),
		%\end{eqnarray}
		Furthermore, for the function $g(\rho,u): = \mathbb{E}[\Phi(\cdot,R_u, \cdot)|R_0=\rho]$, we have that: 
		\begin{eqnarray}\label{markov-gen}
			\frac{\partial g}{\partial u}(\rho, u) = \sum_{j \neq \rho} q_{\rho j} \big(g(j,u) - g(\rho,u)\big).
		\end{eqnarray}
		%where $q_{ij}$ is the $ij-$th entry of the generator matrix of $R_t$. 
		We now begin the calculations for the PIDE by making the change of variables $t:= u -s$:
		\begin{flalign}\nonumber
			\Phi(x,\rho, u) =  \int_{0}^{u} \lambda e^{-\lambda (u-t)} \int_{0}^{\infty} \int_{\mathbb{R}}  \mathcal{T}_{u-t}\Phi(x'+z, \rho, u) p(x',x,u-t) dF(z)dx'dt + e^{-\lambda u}Q(x,\rho,u).
		\end{flalign}
		%This allows us to differentiate $\Phi(x,\rho, u)$ without assuming the derivative of $\Phi$ in the integral.
		By the Leibniz rule, and substituting in the definition of the operator $\mathcal{T}_s$, we then have:
		\begin{eqnarray}
			\frac{\partial \Phi}{\partial u}(x,\rho, u)  = \lambda \int_{0}^{\infty}\int_{\mathbb{R}} \mathbb{E}[\Phi(x'+z, R_0^{\rho},u)]p(x',x,0)dF(z)dx' \nonumber\\+ \int_{0}^{u} -\lambda^2 e^{-\lambda(u-t)}  \int_{0}^{\infty}\int_{\mathbb{R}} \mathbb{E}[\Phi(x'+z,R_{u-t}^{\rho},t)] p(x',x,u-t)dF(z)dx'dt \nonumber \\  + \int_{0}^{u} \lambda e^{-\lambda(u-t)}  \int_{0}^{\infty}\int_{\mathbb{R}}\mathbb{E}[\Phi(x'+z,R_{u-t}^{\rho},t)] \frac{\partial p}{\partial u}(x',x,u-t)dF(z)dx'dt \nonumber\\
			+\int_{0}^{u} \lambda e^{-\lambda(u-t)} \int_{0}^{\infty}\int_{\mathbb{R}} \frac{\partial}{\partial u}\mathbb{E}[\Phi(x'+z,R_{u-t}^{\rho},t)] p(x',x,u-t)dF(z)dx'dt \nonumber\\ - \lambda e^{-\lambda u} Q(x,\rho, u) + e^{-\lambda u} \frac{\partial Q}{\partial u}(x,\rho,u).
		\end{eqnarray}
		The first term in the expression above can be written as:
		\begin{eqnarray}\notag
			\lambda \int_{0}^{\infty}\int_{\mathbb{R}}\mathbb{E}[\Phi(x'+z,\rho, u)]\delta_x(x')dF(z)dx' = \lambda \int_{\mathbb{R}} \Phi(x+z,\rho, u)dF(z). 
		\end{eqnarray}
		On the other hand, using the integral equation for $\Phi(x,\rho, u)$, the second term can be written as:
		\begin{eqnarray}\notag
			-\lambda \Phi(x,\rho,u) + \lambda e^{-\lambda u }Q(x,\rho, u).
		\end{eqnarray}
		Combining, we obtain: 
		\begin{flalign}\label{partial_u}
			\frac{\partial \Phi}{\partial u}(x&,\rho, u) = \int_{0}^{u} \lambda e^{-\lambda(u-t)}  \int_{0}^{\infty}\int_{\mathbb{R}} \mathbb{E}[\Phi(x'+z,R^{\rho}_{u-t},t)] \frac{\partial p}{\partial u}(x',x,u-t)dF(z)dx'dt \nonumber\\
			&+\int_{0}^{u} \lambda e^{-\lambda(u-t)} \int_{0}^{\infty}\int_{\mathbb{R}} \frac{\partial}{\partial u}\mathbb{E}[\Phi(x'+z,R^{\rho}_{u-t},t)] p(x',x,u-t)dF(z)dx'dt  \nonumber\\ &+\lambda \int_{\mathbb{R}} \Phi(x+z,\rho, u)dF(z)	-\lambda \Phi(x,\rho, u) +  e^{-\lambda u }\frac{\partial Q}{\partial u}(x,\rho, u).
		\end{flalign}
		We now consider the generator of $\Phi(x,\rho, u)$. It is straightforward to separate the components of $\mathcal{L}_1\Phi(x,\rho, u)$ since only the transition density $p(x',x,u-t)$ depends on $x$ and only $\mathbb{E}[\Phi(x'+z,R_{u-t}^{\rho},t)]$ depends on the regime. We then have: 
		\begin{flalign}
			&\mathcal{L}_1\Phi(x,\rho, u)  \equiv k_{\rho}(\theta_{\rho} - x)\frac{\partial \Phi}{\partial x} (x,\rho, u) + \frac{1}{2}\sigma_{\rho}^2 \frac{\partial^2 \Phi}{\partial x^2}(x,\rho, u) + \sum_{j \neq \rho} q_{\rho j} \Big(\Phi(x,j,u) - \Phi(x,\rho,u)\Big) &\nonumber\\  &=
			e^{-\lambda u} \mathcal{L}_1Q(x,\rho,u) + \int_{0}^{u} \lambda e^{-\lambda(u-t)}  \int_{0}^{\infty}\int_{\mathbb{R}} \mathbb{E}[\Phi(x'+z,R_{u-t}^{\rho},t)] \mathcal{L}p(x',x,u-t) dF(z)dx'dt & \nonumber\\
			&+\int_{0}^{u} \lambda e^{-\lambda(u-t)}  \int_{0}^{\infty}\int_{\mathbb{R}} \sum_{j \neq \rho} q_{\rho j}\Big(\mathbb{E}[\Phi(x'+z,R_{u-t}^{j},t)]  -  \mathbb{E}[\Phi(x'+z,R_{u-t}^{\rho},t)]\Big)p(x',x,u-t) dF(z)dx'dt & \nonumber \\&=
			e^{-\lambda u }\frac{\partial Q}{\partial u}(x,\rho, u)+ \int_{0}^{u} \lambda e^{-\lambda(u-t)} \int_{0}^{\infty}\int_{\mathbb{R}} \mathbb{E}[\Phi(x'+z,R_{u-t}^{\rho},t)] \frac{\partial p}{\partial u}(x',x,u-t)dF(z)dx'dt  &\nonumber\\ &+\int_{0}^{u} \lambda e^{-\lambda(u-t)}  \int_{0}^{\infty}\int_{\mathbb{R}} \frac{\partial}{\partial u}\mathbb{E}[\Phi(x'+z,R^{\rho}_{u-t},t)] p(x',x,u-t)      dF(z)dx'dt, &
		\end{flalign}
		where we have used the fact that $Q(x,\rho, u)$ and $p(x',x,u)$ satisfy (\ref{kolm-3}) and (\ref{kolm-4}), respectively. Using (\ref{partial_u}), we obtain:
		\begin{eqnarray}\notag
			\mathcal{L}_1\Phi(x,\rho,u) = \frac{\partial \Phi}{\partial u}(x,\rho, u) - \lambda \Big(\int_{\mathbb{R}} \Phi(x+z,\rho,u)dF(z)- \Phi(x,\rho, u) \Big),
		\end{eqnarray}
		which, upon rearranging, can be written as:
		\begin{flalign}
			\frac{\partial \Phi}{\partial u}(&x,\rho, u)= k_{\rho}(\theta_{\rho} - x)\frac{\partial \Phi}{\partial x}(x,\rho,u) + \frac{1}{2}\sigma_{\rho}^2 \frac{\partial^2 \Phi}{\partial x^2}(x,\rho,u)  \notag \\ & + \sum_{j \neq \rho} q_{\rho j} \Big(\Phi(x,j,u) - \Phi(x,\rho,u)\Big) + \int_{\mathbb{R}} \Big( \Phi(x+z,\rho,u)- \Phi(x,\rho,u) \Big)\nu(dz).
		\end{flalign}
		\item[$(ii)$] 
		The proof follows as above. Under model (\ref{model-sv-def}), for $p(\cdot,x,u)$ we now have:
		\begin{eqnarray}\label{eqn-sv-1}
			\frac{\partial p}{\partial t}(\cdot, x,u)= k(\theta - x)\frac{\partial p}{\partial x}(\cdot, x,u) + \frac{1}{2} y \frac{\partial^2 p}{\partial x^2}(\cdot, x,u).
		\end{eqnarray}
		In this case, we also have to account for the transition density of the underlying volatility process, $q(\cdot, y, u)$, which satisfies:
		\begin{eqnarray}\label{eqn-sv-2}
			\frac{\partial q}{\partial u}(\cdot, y, u)=\kappa(\mu - y)\frac{\partial q}{\partial y}(\cdot, y, u) + \frac{1}{2} \xi^2 y \frac{\partial^2 q}{\partial y^2}(\cdot, y, u),
		\end{eqnarray} 
		and recall that for $Q(x,y,u)$ we have:
		\begin{eqnarray}\label{kolm-6}
			\frac{\partial Q}{\partial u}(x,y, u) = \mathcal{L}_2 Q(x,y,u),
		\end{eqnarray}
		with the generator under the stochastic volatility model $\mathcal{L}_2$ given by:
		\begin{flalign}
			\mathcal{L}_2 f(x,y, t):= k(\theta - x) \frac{\partial f}{\partial x}(x,y,t) + \kappa (\mu -y)\frac{\partial f}{\partial y}(x,y,t) + \frac{1}{2} y \frac{\partial^{2} f}{\partial x^{2}}(x,y,t) + \frac{1}{2} \xi^2 y \frac{\partial^{2} Q}{\partial y^{2}}(x,y,t) 
		\end{flalign}

		Differentiating (\ref{int-sv-exp}) with respect to $x$ and $u$, and comparing $\mathcal{L}_2\Phi(x,y,u)$ with $\frac{\partial \Phi}{\partial u}(x,y,u)$ we obtain the required PIDE, using the same steps as in $(i)$.
	\end{itemize}
\end{proof}

We can now present the main result for the generalized model. Even though the steps are similar as the cases above, the dependence on both the regime and volatility processes creates additional terms and it is therefore worth outlining the proof in detail.

\begin{theorem}
	Under the generalized asset process model (\ref{model-gen-def}) the survival probability $\Phi(x,\rho,y,u)$ satisfies the PIDE:
	\begin{flalign} \label{pide-gen}
		&\frac{\partial \Phi}{\partial u}(x,\rho,y,u) \notag \\ &= k_\rho(\theta_{\rho}-x) \frac{\partial \Phi}{\partial x}(x,\rho,y,u) + \kappa (\mu -y)\frac{\partial \Phi}{\partial y}(x,\rho, y,u) +\frac{1}{2} \sigma_\rho^2 y\frac{\partial^{2} \Phi}{\partial x^{2}}(x,\rho,y,u)+\frac{1}{2} \xi^2 y \frac{\partial^{2} \Phi}{\partial y^{2}}(x,\rho, y,u)\notag \\ & +\sum_{j \neq \rho} q_{\rho j} \Big(\Phi(x,j,y,u) - \Phi(x,\rho,y,u)\Big)  
		+ \int_{ \mathbb{R}}\Big(\Phi(x+z,\rho,y, u) -  \Phi(x,\rho,y,u) \Big)\nu(dz), 
	\end{flalign}
	for $(x,\rho,y,u) \in  \mathcal{D} \times \mathcal{R} \times \mathcal{V}\times [0,T]$, with initial and boundary conditions: 
	\begin{align} \nonumber
		\Phi(x,\rho,y,0)= \mathbbm{1}_{\{x > 0\}},\,\,\ (x,\rho,y) \in \mathcal{D} \times \mathcal{R} \times \mathcal{V},\\ \nonumber
		\Phi(0,\rho,y, u) = 0, \,\,\ (\rho,y,u ) \in \mathcal{R} \times \mathcal{V} \times [0,T], \nonumber \\
		\Phi(x,\rho, y, u) \rightarrow 1, \text{ as } x \rightarrow \infty, \,\,\ (\rho,y, u) \in \mathcal{R} \times \mathcal{V} \times [0,T], \nonumber \\
		\frac{\partial \Phi}{\partial y}(x,y,u) = 0, \text{ as } y \rightarrow \infty \,\,\ (x,\rho, u) \in  \mathcal{D} \times \mathcal{R} \times [0,T].
	\end{align} 
\end{theorem}
\begin{proof}
	For the transition density $p(\cdot,x,u)$ we have:
	\begin{eqnarray}\label{gen-den-1}
		\frac{\partial p}{\partial t}(\cdot, x,u)= k_{\rho}(\theta - x)\frac{\partial p}{\partial x}(\cdot, x,u) + \frac{1}{2} \sigma_{\rho}^2 y \frac{\partial^2 p}{\partial x^2}(\cdot,x,u),
	\end{eqnarray}
	and for $q(\cdot, y, u)$ and $g(\rho,u):=\mathbb{E}[\Phi(\cdot,R_u^{\rho},\cdot, \cdot)]$ we know that (\ref{eqn-sv-2}) and (\ref{markov-gen}) hold, respectively.
	In this case, for $Q(x,\rho, y,u)$ we have:
	\begin{eqnarray}\label{kolm-7}
		\frac{\partial Q}{\partial u}(x,\rho, y, u) = \mathcal{L}_3 Q(x,\rho,y,u),
	\end{eqnarray}
	with the generator under the generalized model, $\mathcal{L}_3$, given by:
	\begin{flalign}
		\mathcal{L}_3 f(x,\rho,y, u)&:=k_\rho(\theta_{\rho}-x) \frac{\partial f}{\partial x}(x,\rho,y,t) + \kappa (\mu -y)\frac{\partial f}{\partial y}(x,\rho, y,t)  &&\notag \\ +\frac{1}{2} \sigma_\rho^2 &y\frac{\partial^{2} f}{\partial x^{2}}(x,\rho,y,t)+\frac{1}{2} \xi^2 y \frac{\partial^{2} f}{\partial y^{2}}(x,\rho, y,t) +\sum_{j \neq \rho} q_{\rho j} \Big(f(x,j,y,t) - f(x,\rho,y,t)\Big)  \label{l_3},&&
	\end{flalign}
	
	%where $q_{ij}$ is the $ij-$th entry of the generator matrix of $R_t$. 
	%\par To obtain the PIDE we begin by making the change of variables $t:= u -s$:
	%\begin{eqnarray}\nonumber
	%\Phi(x,\rho, u) =  \int_{0}^{u} \lambda e^{-\lambda (u-t)} \int_{0}^{\infty} \int_{-x'}^{\infty}  \mathbb{E}_R[\Phi(x'+z, R_{u-t}, t)] p(x',x,u-t) dF_{\rho}(z)dx'dt + e^{-\lambda u}Q(x,\rho,u),
	%\end{eqnarray}
	As in the results above it will be useful to work with the definition of $\mathcal{T}_s$, i.e., version (\ref{int-sv}) of the integral equation. With the change of variables $t=u-s$ and applying the Leibniz rule we now get:
	\begin{flalign}
		\frac{\partial \Phi}{\partial u}&(x,\rho, u)  = \lambda \int_{0}^{\infty}\int_{0}^{\infty}\int_{\mathbb{R}} \mathbb{E}[\Phi(x'+z, R_0^{\rho},\nu, u)]q(\nu,y,0)p(x',x,0)dF(z)dx'd\nu \nonumber\\& + \int_{0}^{u} -\lambda^2 e^{-\lambda(u-t)}  \int_{0}^{\infty}\int_{0}^{\infty}\int_{\mathbb{R}} \mathbb{E}[\Phi(x'+z,R_{u-t}^{\rho},\nu, t)] q(\nu,y,u-t)p(x',x,u-t)dF(z)dx'd\nu dt \nonumber \\ & +\int_{0}^{u} \lambda e^{-\lambda(u-t)} \int_{0}^{\infty}\int_{0}^{\infty}\int_{\mathbb{R}} \frac{\partial}{\partial u}\mathbb{E}[\Phi(x'+z,R_{u-t}^{\rho},\nu, t)] q(\nu,y,u-t) p(x',x,u-t) dF(z)dx'd\nu dt \nonumber\\
		&+ \int_{0}^{u} \lambda e^{-\lambda(u-t)}  \int_{0}^{\infty}\int_{0}^{\infty}\int_{\mathbb{R}}\mathbb{E}[\Phi(x'+z,R_{u-t}^{\rho},\nu, t)] \frac{\partial q}{\partial u}(\nu,y,u-t)p(x',x,u-t)dF(z)dx'd\nu dt \nonumber\\
		&+ \int_{0}^{u} \lambda e^{-\lambda(u-t)}  \int_{0}^{\infty}\int_{0}^{\infty}\int_{\mathbb{R}}\mathbb{E}[\Phi(x'+z,R_{u-t}^{\rho},\nu, t)] q(\nu,y,u-t)\frac{\partial p}{\partial u}(x',x,u-t)dF(z)dx'd\nu dt \nonumber\\
		&- \lambda e^{-\lambda u} Q(x,\rho,y, u) + e^{-\lambda u} \frac{\partial Q}{\partial u}(x,\rho,y,u).
	\end{flalign}
	From the first term we have: 
	\begin{eqnarray}\notag
		\lambda \int_{0}^{\infty} \int_{0}^{\infty}\int_{\mathbb{R}}\mathbb{E}[\Phi(x'+z,\rho,\nu, u)]\delta_y(\nu)\delta_x(x')dF(z)dx'd\nu = \lambda \int_{\mathbb{R}} \Phi(x+z,\rho, y,u)dF(z),  
	\end{eqnarray}
	whereas, from (\ref{int-gen}), the second term can be written as $-\lambda \Phi(x,\rho,y, u) + \lambda e^{-\lambda u }Q(x,\rho, y,u)$, and hence, we have:
	\begin{flalign}\label{partial_u-gen}
		&\frac{\partial \Phi}{\partial u}(x,\rho, y, u) = \lambda \int_{\mathbb{R}} \Phi(x+z,\rho, y,u)dF(z)-\lambda \Phi(x,\rho, y, u) +  e^{-\lambda u }\frac{\partial Q}{\partial u}(x,\rho, y, u) &\nonumber\\&+  \int_{0}^{u} \lambda e^{-\lambda(u-t)}  \int_{0}^{\infty}\int_{0}^{\infty}\int_{\mathbb{R}}\mathbb{E}[\Phi(x'+z,R_{u-t},\nu, t)] \frac{\partial q}{\partial u}(\nu,y,u-t)p(x',x,u-t)dF(z)dx'd\nu dt &\nonumber\\
		&+ \int_{0}^{u} \lambda e^{-\lambda(u-t)}  \int_{0}^{\infty}\int_{0}^{\infty}\int_{\mathbb{R}}\mathbb{E}[\Phi(x'+z,R_{u-t},\nu, t)] q(\nu,y,u-t)\frac{\partial p}{\partial u}(x',x,u-t)dF(z)dx'd\nu dt &\nonumber\\
		&+\int_{0}^{u} \lambda e^{-\lambda(u-t)} \int_{0}^{\infty}\int_{0}^{\infty}\int_{\mathbb{R}} \frac{\partial}{\partial u}\mathbb{E}[\Phi(x'+z,R_{u-t},\nu, t)] q(\nu,y,u-t) p(x',x,u-t) dF(z)dx'd\nu dt.&
	\end{flalign}
	Taking the derivatives of $\Phi(x,\rho, y,u)$ with respect to $x$ and $y$ is straightforward. Therefore, using (\ref{partial_u-gen}) with (\ref{gen-den-1}),
	(\ref{eqn-sv-2}) and (\ref{markov-gen}), we obtain:
	\begin{eqnarray}\notag
		\mathcal{L}_3\Phi(x,\rho, y, u)  =\frac{\partial \Phi}{\partial u}(x,\rho, y, u) - \lambda \Big(\int_{\mathbb{R}} \Phi(x+z,\rho, y,u)dF(z)	- \Phi(x,\rho, y, u) \Big),
	\end{eqnarray}
	and the expected PIDE follows.
\end{proof}

It is easy to see that the same steps can be used to obtain an equivalent PIDE for the simple PD function (corresponding to a regime switching model with one regime), whose integral equation representation is given by (\ref{simple_eqn}). The PIDE is shown below; we omit the proof for brevity, as it follows directly as a special case of the results above.
\begin{corollary} \label{equiv-int}
	Under model ($\ref{hom-gou}$) the survival probability $\Phi(x, u)$ satisfies the PIDE:
	\begin{flalign} \label{pide2}
		\frac{\partial \Phi}{\partial u}=k(\theta-x) \frac{\partial \Phi}{\partial x} +\frac{1}{2} \sigma^{2} \frac{\partial^{2} \Phi}{\partial x^{2}}  + \int_{\mathbb{R}}\Big(\Phi(x+z,u)-\Phi(x,u)\Big) \nu(dz)=0,\,\,\ (x,u) \in \mathcal{D} \times [0,T],
	\end{flalign}
	with initial and boundary conditions: 
	\begin{align} \notag
		\Phi(x,0)= \mathbbm{1}_{\{x > 0\}},\,\,\ x\in\mathcal{D}, \\
		\Phi(0,u) = 0, \,\,\ u \in [0,T], \notag \\
		\Phi(x, u) \rightarrow 1 \text{ as } x \rightarrow \infty, \,\,\ u \in [0,T] . 
	\end{align}
\end{corollary}

\begin{remark} \label{mis}
	It is worth noting that the approach we have developed results in PIDEs consistent with the ruin probability examined in \cite{mishura2016ruin}, where the asset process is given by:
	$$X_t(x)  = x + \int_{0}^{t} r(X_s(x) + c) ds + \sum_{i=1}^{N_t} (-Y_i),$$ with jump distribution $f(y)$ on $\mathbb{R_+}$ and jump intensity $\lambda$. It is shown that the survival probability $\phi(x,t) = \mathbb{P}\Big(\inf_{\substack{r \leq t}} X_r(x) > 0\Big)$ satisfies the PIDE:
	$$ \frac{\partial \phi}{\partial t} - (rx+c)\frac{\partial \phi}{\partial x} +  \lambda \Big(\phi(x,t) - \int_{0}^{x} \phi(x-y,t) dF(y) \Big)  = 0.$$
	Noting that $\phi(x,t) = 0$ for $x<0$, this can be rewritten as:
	$$ \frac{\partial \phi}{\partial t} - (rx+c)\frac{\partial \phi}{\partial x} -  \lambda \Big( \int_{ \mathbb{R}} \big(\phi\big(x+(-y),t\big) - \phi(x,t)\big) dF(y) \Big)  = 0,$$
	and therefore, from the definition of the L\'evy measure $\nu(\cdot)$, we obtain: 
	$$ \frac{\partial \phi}{\partial t} = (rx+c)\frac{\partial \phi}{\partial x} + \int_{\mathbb{R}} \Big(\phi\big(x+(-y),t\big) - \phi(x,t) \Big)\nu(dy)   = 0,$$
	We can see that this PIDE is equivalent to that in (\ref{pide2}), given that the diffusion term is zero and therefore the second derivative to $x$ is zero. This confirms that our results are consistent with existing definitions and models that have been studied in the literature, with the important difference that we take advantage of the integral equations and corresponding continuity results, which allow us to consider more complex models and to obtain appropriate solutions in all cases.
	%	\hfill $\triangleleft$
\end{remark}

\subsection{Regularity of the one-dimensional PD function}
Finally, we study the case of the one dimensional model to show that the resulting survival and PD function enjoys the properties required to consider equation (\ref{pide2}) in the strict sense. Specifically, this requires that $\Phi(x,u)$ is (at least) twice and once differentiable with respect to the spatial and temporal variables, respectively. These results are based on the corresponding calculations for the general family of second order parabolic PIDEs, as studied in \cite{garroni1992green}. A reminder of the main results that we will use are given in Appendix \ref{reg-par-pde}.
\par We first rewrite (\ref{pide2}) in a concise form for the calculations that follow. Note that we adopt the notation of the appropriate spaces used in \cite{garroni1992green}, which is also used in Appendix \ref{reg-par-pde}. It will be useful to set $\mathcal{D} := (0,\infty)$. This will allow us to consider a smooth initial condition, rather than the Heaviside function as in formulation of the PIDE in Corollary \ref{equiv-int}, since $\Phi(x,u) = 0$, for $x \leq 0$, by definition. Then, for a fixed time until maturity $T >0$, we write:  
\begin{eqnarray}\label{pde-reg-one-ou}
	\begin{cases}
		L\Phi(x,u) = I\Phi(x,u) & \text{ for } (x,u) \in Q_T:= \mathcal{D} \times [0,T] \\ 
		\Phi(x,0)= 1 & \text { for } x \in \mathcal{D} \\
		\Phi(x,t) =  \mathbbm{1}_{x > 0} & \text{ for } x \in \Sigma_T := \partial \mathcal{D} \times [0,T],
	\end{cases}
\end{eqnarray}
where we define the operator $L$ by $L\Phi(x,u) := \frac{\partial \Phi}{\partial u} - \mathcal{L}\Phi(x,u)$ and the integral operator $I$ by $I\Phi(x,u) :=\int_{\mathbb{R}}\Big(\Phi(x+z,u)-\Phi(x,u)\Big) \nu(dz)$, and $\partial \mathcal{D}$ is the standard notation for the boundary of domain $\mathcal{D}$. Our aim is to show that the above PIDE has a solution satisfying appropriate regularity conditions. To this end, we first define some relevant function spaces that will be required for the subsequent regularity results.

\begin{definition}
	Consider $\Omega \subset \mathbb{R}^n$ an open set, with closure $\bar{\Omega}$. Furthermore, consider a fixed time horizon $T >0$ and define $Q_T = \Omega \times [0,T]$, with closure $\bar{Q}_T$. We then define the following spaces, for $0<\alpha<1$: 
	\begin{itemize}
		\item $C^0(\bar{\Omega})$ is the Banach space of bounded continuous functions in $\bar{\Omega}$, with the natural supremum norm:
		$$
		\|\cdot\|_{C^0(\bar{\Omega})} \equiv\|\cdot\|_{0, \bar{\Omega}}=\sup _{\Omega}|\cdot|
		$$
		\item $C^{2,1}\left(\bar{Q}_T\right)$ is the Banach space of functions $\varphi(x, t)$ belonging to $C^0\left(\bar{Q}_T\right)$ together their derivatives $\frac{\partial f}{\partial x}, \frac{\partial^2 f}{\partial x^2}, \frac{\partial f}{\partial t}$ in $\bar{Q}_T$ with natural norm.
		
		\item $C^{\alpha, \frac{\alpha}{2}}\left(\bar{Q}_T\right)$ is the Banach space of function $\varphi$ in $C^0\left(\bar{Q}_T\right)$ which are Hölder continuous in $\bar{Q}_T$ with exponent $\alpha$ in $x$ and $\frac{\alpha}{2}$ in $t$ i.e. having a finite value for the seminorm
		$$
		\langle f \rangle_{\bar{Q}_T}^{(\alpha)} \equiv\langle f \rangle_{x, \bar{Q}_T}^{(\alpha)}+\langle f \rangle_{t, \bar{Q}_T}^{\left(\frac{\alpha}{2}\right)}
		$$
		where
		$$
		\begin{aligned}
			&\langle f \rangle_{x, \bar{Q}_T}^{(\alpha)}=\inf \left\{C \geq 0:\left|f(x, t)-f\left(x^{\prime}, t\right)\right| \leq C\left|x-x^{\prime}\right|^\alpha, \forall x, x^{\prime}, t\right\} \\
			&\langle f \rangle_{t, \bar{Q}_T}^{\left(\frac{\alpha}{2}\right)}=\inf \left\{C \geq 0:\left|f(x, t)-f \left(x, t^{\prime}\right)\right| \leq C\left|t-t^{\prime}\right|^{\frac{\alpha}{2}}, \forall x, t, t^{\prime}\right\}
		\end{aligned}
		$$
		The quantity
		$$
		\|f\|_{C^{\alpha, \frac{\alpha}{2}}\left(\bar{Q}_T\right)} \equiv\|f\|_{\alpha, \bar{Q}_T}=\|f\|_{0, \bar{Q}_T}+\langle f\rangle_{\bar{Q}_T}^{(\alpha)}
		$$
		defines a norm.
		\item $C^{2+\alpha, \frac{2+\alpha}{2}}\left(\bar{Q}_T\right)$ is the Banach space of functions $f(x, t)$ in $C^{2,1}\left(\bar{Q}_T\right)$ having a finite value for the seminorm:
		$$
		\langle f\rangle_{\bar{Q}_T}^{(2+\alpha)}=\left\langle\partial_t f\right\rangle_{\bar{Q}_T}^{(\alpha)}+\sum_{i, j=1}^d\left\langle\partial_{i j} f\right\rangle_{\bar{Q}_T}^{(\alpha)}+\sum_{i=1}^d\left\langle\partial_i f\right\rangle_{t, \bar{Q}_T}^{\frac{1+\alpha}{2}} .
		$$
		Then, the quantity
		$$
		\|f\|_{C^{2+\alpha, \frac{2+\alpha}{2}}\left(\bar{Q}_T\right)} \equiv\|f\|_{2+\alpha, \bar{Q}_T}=\sum_{2 r+s \leq 2}\left\|\partial_t^r \partial_x^s f\right\|_{0, \bar{Q}_T}+\langle f \rangle_{\bar{Q}_T}^{(2+\alpha)}
		$$
		defines a norm.
	\end{itemize}
\end{definition} 

\begin{proposition}\label{regularity-result}
	Consider a fixed time horizon $T>0$ and the space $Q_T:= \mathcal{D} \times [0,T]$, with closure $\bar{Q}_T$. Then, PDE (\ref{pde-reg-one-ou}) has a solution $\Phi(x,u)$, such that $\Phi(x,u) \in C^{2+\alpha, \frac{2+\alpha}{2}}\left(\bar{Q}_T\right)$.
\end{proposition}
\begin{proof}
	The proof relies on an appropriate fixed point argument. To this end, we first define the mapping $\mathcal{T}v = \Phi$, such that $v$ is a solution of $L\Phi(x,u) = Iv$. Note that from Theorem \ref{reg-theorem}, there exists a unique $\Phi \in  C^{2+\alpha, \frac{2+\alpha}{2}}\left(\bar{Q}_T\right)$ solving the local counterpart of (\ref{pde-reg-one-ou}), where the right hand side of the PDE is zero.
	\par Consider now a function $v \in  C^{2+\alpha, \frac{2+\alpha}{2}}\left(\bar{Q}_T\right)$. It follows that $ Iv \in C^{\alpha, \frac{\alpha}{2}}\left(\bar{Q}_T\right)$ and we also have that: 
	$$\|I v\|_{C^{\alpha, \frac{\alpha}{2}}\left(\bar{Q}_T\right)} \leq \varepsilon\|\nabla v\|_{C^{\alpha, \frac{\alpha}{2}}\left(\bar{Q}_T\right)}+C(\varepsilon)\|v\|_{C^{\alpha, \frac{\alpha}{2}}\left(\bar{Q}_T\right)},$$ from Theorem \ref{int-op}. Furthermore, by definition of the mapping $\mathcal{T}$ we have that if $v \in C^{2+\alpha, \frac{2+\alpha}{2}}\left(\bar{Q}_T\right)$  then $\Phi = \mathcal{T}v \in C^{2+\alpha, \frac{2+\alpha}{2}}\left(\bar{Q}_T\right)$. Hence, $\mathcal{T}$ is a map from $ C^{2+\alpha, \frac{2+\alpha}{2}}\left(\bar{Q}_T\right)$ onto itself and is also single-valued, by the uniqueness of the solution of the PDE.  
	\par We will now show that $\mathcal{T}$ is also a contraction in order to then apply Banach's fixed point argument. To this end, consider $v, v' \in  C^{2+\alpha, \frac{2+\alpha}{2}}\left(\bar{Q}_T\right)$, with the corresponding mappings $Tv, Tv' \in  C^{2+\alpha, \frac{2+\alpha}{2}}\left(\bar{Q}_T\right)$. By the definition of the mapping $T$ have that: 
	\begin{eqnarray}\label{pde-v-v}
		\begin{cases}
			L\Phi(x,u) = Iv & \text{ for } (x,u) \in Q_T:= \mathcal{D} \times [0,T] \\ 
			L\Phi'(x,u) = Iv' & \text{ for } (x,u) \in Q_T:= \mathcal{D} \times [0,T], \\ 
		\end{cases}
	\end{eqnarray}
	and therefore $L\hat{\Phi}(x,u) = I\hat{v}$, with $\hat{\Phi}:= \Phi - \Phi'$ and $\hat{v}$ is defined analogously. Hence:
	\begin{flalign}
		\|\hat{\Phi}\|_{C^{\alpha, \frac{\alpha}{2}}\left(\bar{Q}_T\right)} = &\|\mathcal{T}v - \mathcal{T}v'\|_{C^{\alpha, \frac{\alpha}{2}}\left(\bar{Q}_T\right)} \leq C \|I \hat{v} \|_{C^{\alpha, \frac{\alpha}{2}}\left(\bar{Q}_T\right)} \notag \\\leq  & \varepsilon\|\nabla \hat{v}\|_{C^{\alpha, \frac{\alpha}{2}}\left(\bar{Q}_T\right)}+C(\varepsilon)\|\hat{v}\|_{C^{\alpha, \frac{\alpha}{2}}\left(\bar{Q}_T\right)},
	\end{flalign}
	where the last inequality follows from Theorem \ref{int-op}. To show that $\mathcal{T}$ is indeed a contraction, we first note that all terms in the final expression above are bounded by $\|\hat{v}\|_{C^{\alpha, \frac{\alpha}{2}}\left(\bar{Q}_T\right)}$. We need $\|\hat{\Phi}\|_{C^{\alpha, \frac{\alpha}{2}}\left(\bar{Q}_T\right)} \leq k \|\hat{v}\|_{C^{\alpha, \frac{\alpha}{2}}\left(\bar{Q}_T\right)}$, with $k<1$. For this, notice that the first term in the final expression above can be made arbitrarily small, however the second depends on the $C$ value, which in turn depends on the time horizon $T$. We can therefore make $C(\epsilon)<1$ if we consider a small enough horizon, i.e., $T = \delta$, creating a solution $\Phi(x,u) \in C^{2+\alpha, \frac{2+\alpha}{2}}\left(\bar{Q}_{\delta}\right)$. We can apply this approach to find an appropriate solution in $\Omega \times [\delta, 2\delta]$, continuing until the entire interval $[0,T]$ is covered. 
\end{proof}
\begin{remark}
	Based on the results given by \cite{garroni1992green}, the result above can be extended for higher dimensions with $x \in \mathbb{R}^d$ and an analogous parabolic operator $L$. Therefore, Proposition \ref{regularity-result} holds for the stochastic volatility model, as well. Furthermore, in the case of the regime switching model we obtain a simple coupling of parabolic PIDEs (identical to that for the one dimensional model), and it is therefore expected that the regularity result holds for the PD function $\Phi(x,\cdot,u)$.  
\end{remark}
With this result we have shown that we can go beyond the notion of viscosity solutions in the case of the PD functions and obtain solutions that satisfy all required regularity properties. In the next paper we will consider numerical solutions to the PIDEs, some of which we can now interpret as strong solutions, under the conditions mentioned above.

	\section{Numerical estimation of PD functions} \label{num-schemes}
In this paper, we will develop numerical schemes to solve the PIDEs obtained in Section \ref{pides-visc} and use the resulting PD values in specific examples of the aforementioned IFRS 9 modelling tasks. We choose to focus on the numerical solutions of the PIDEs, rather than the corresponding integral equations, as we will be able to employ standard finite difference schemes to estimate the solutions, as detailed below. For clarity and illustrative purposes, we will first consider the one dimensional OU model given by (\ref{hom-gou}) (recall that we have shown that the PD function $\Phi(x,u) \in  C^{2+\alpha, \frac{2+\alpha}{2}}\left(\bar{Q}_T\right)$ is a strong solution to the corresponding PIDE). From the resulting finite difference scheme we will then be able to build the solutions for the regime switching and stochastic volatility models. We will focus on these models for our numerical solutions, as they cover the applications in credit risk that we consider in this paper, noting the case of the generalized model (\ref{model-gen-def}) can be developed by combining the methods that follow. However, due to the additional terms, the corresponding numerical scheme suffers from the well-known "curse of dimensionality" problem. 
\par Our approach in this Section follows the methodology developed in \cite{cont2005finite} and \cite{d2005robust}. We extend the numerical schemes by considering variable coefficients and further develop the corresponding methods for the regime switching and stochastic volatility models. Due to the additional variables, we will see that these require careful handling of the derivative discretizations to ensure the required stability and monotonicity properties hold. As mentioned, we will start with the one dimensional model, which produces the PIDE given in (\ref{pide2}), the finite difference scheme for which is similar to that developed in \cite{cont2005finite}. However, this first step will allow us to explicitly account for the variable drift term and detail its effect on the finite difference scheme and is therefore worth presenting the analytical calculations.

\subsection{One dimensional model}
\par Before implementing the numerical methods, it is important to discuss the spatial and temporal domains over which the schemes will be solved. We consider a spatial domain $x\in\mathcal{D} \subset \mathbb{R}$. Therefore, for the construction of the numerical scheme one can consider the interval $x\in [0,S]$ with non-trivial solutions $\Phi(x,u) \in (0,1)$ (in practice, the value of $S$ depends on the parameters of the underlying processes and its approximation may require Monte Carlo simulations). %For simplicity, we scale the values of $x$ and the remaining parameters of the OU process such that $x\in[0,1]$. 
However, given that the PIDEs contain the non-local integral terms, and the OU process is defined on $\mathbb{R}$, we will define $\mathcal{D} = [-B, B]$, for some constant $B>S$ and extend the boundary conditions $\Phi(x,\cdot)=0$ for $x\in[-B,0)$ and $\Phi(x,\cdot)=1$ for $x\in(S,B]$. This way, we will be able to calculate the integral term, as detailed below. For the temporal domain, we simply consider $t \in [0,1]$ (we rewrite $u$ as $t$ as there is not risk of confusion in what follows). Given the added complexity from the non-local term, we give a detailed explanations of each of the three aforementioned schemes in this section, along with examples of specific asset value processes and, subsequently, examples of the modelling tasks pertaining to credit risk under the IFRS 9 framework we previously discussed.
\par We write (\ref{pide2}) as follows:
\begin{eqnarray} \label{pide3}
	\frac{\partial \Phi}{\partial t}=k(\theta-x) \frac{\partial \Phi}{\partial x} + \frac{1}{2} \sigma^{2} \frac{\partial^{2} \Phi}{\partial x^{2}} +\int_{ \mathbb{R}}\Phi(x+z,t)\nu(dz) - \Phi(x,t)\int_{ \mathbb{R}} \nu(dz). 
\end{eqnarray}
%Noting that $\Phi(x+z,t)=0$ for $z \leq -x$ the above can be written as:
%\begin{eqnarray} \label{pide4}
%\frac{\partial \Phi}{\partial t}=\frac{1}{2} \sigma^{2} \frac{\partial^{2} \Phi}{\partial %x^{2}}+k(\theta-x) \frac{\partial \Phi}{\partial x}  - \lambda \Phi(x,t)+\lambda %\int_{-x}^{\infty}\Phi(x+z,t)f(z)dz.
%\end{eqnarray}
Note that in the discretized version of this PIDE we will also have to approximate the integral with respect to the L\'evy measure. We employ an implicit scheme leading to a backward time centered space (BTCS) method, and handle the non-local term explicitly, as in \cite{cont2005finite}. Consider space and time grids, with step sizes $\Delta x$ and $\Delta t$, and with $N$ and $T$ total points, respectively. Therefore, we have that $\Phi^q_{p}$ represents the survival probability at the grid point $t=t_0+ q\Delta t$, $x=x_0+p\Delta x$, i.e. $\Phi^q_p = \Phi(x_0+p \Delta x,t_0+ q\Delta t)$. Furthermore, let $L$, $D$ and $U$ be number of grid points in the intervals $[-B,0)$, $[0,S]$ and $(S,B]$, respectively, so that $N=L+D+U$. 
\par For the integral terms, we first must approximate the jump density by considering a ball around the $x$-value of the grid:
\begin{eqnarray}
	\bar{f_i} = \frac{1}{\Delta x}\int_{x_i-\frac{\Delta x}{2}}^{x_i+\frac{\Delta x}{2}} f(x)dx.	
\end{eqnarray}
Then, noting that $\nu(dz) = \lambda F(dz)$, we can approximate the first and second integral terms in (\ref{pide3}) using:
\begin{eqnarray} \label{integral_term}
	\mathcal{I}\Phi^q_p := \sum_{i=-J/2}^{J/2}\Phi_{p+i}^{q} \bar{f}_{i} \Delta z, \\
	\hat{I} =  \sum_{i=-J/2}^{J/2}\bar{f}_{i} \Delta z, \label{int_2}
\end{eqnarray}
for some $ J \in \mathbb{Z}_+$ large enough to ensure that $\hat{I}$ is sufficiently close to $1$. In the above, we have defined the operator $\mathcal{I}: \mathcal{C} \rightarrow \mathcal{C}$, where $\mathcal{C}$ is the Banach space as defined in Proposition \ref{schauder}. We will refer to this as the integral operator.
For simplicity, we will be taking $\Delta x = \Delta z$ in the calculations and numerical results below.
\par The resulting implicit scheme for PIDE (\ref{pide3}) is given by:
\begin{flalign}\label{BTCS}
	\frac{\Phi_{p}^{q+1}-\Phi_{p}^{q}}{\Delta t}=k(\theta - x_p) \frac{\Phi_{p+1}^{q+1}-\Phi_{p-1}^{q+1}}{2 \Delta x}+\frac{1}{2}\sigma^2 \frac{\Phi_{p+1}^{q+1}-2 \Phi_{p}^{q+1}+\Phi_{p-1}^{q+1}}{\Delta x^2} +\lambda \mathcal{I}\Phi^q_{p} - \lambda \hat{I} \Phi_{p}^{q},
\end{flalign} 
which, upon rearranging, can be written as: 
\begin{eqnarray}\label{ou-scheme}
	- \Phi_{p-1}^{q+1} c_{p} \Delta t   + \Phi_p^{q+1}\big(1+ a_{p}\Delta t\big) - \Phi_{p-1}^{q+1} b_{p} \Delta t  = (1- \lambda \Delta t \hat{I})\Phi_p^q + \lambda \Delta t \mathcal{I}\Phi_p^q,
\end{eqnarray}
for $ q = 1,2,\dots T-1$ and with coefficients $a_{p}, b_{p}$ and $c_{p},$ for $p = 0,1,\dots, N-1$, given by: 
\begin{eqnarray}
	c_{p} = \frac{\sigma^2 }{2\Delta x^2} - \frac{k(\theta-x_p)}{2\Delta x} \notag\\
	a_{p} =   \frac{\sigma^2}{\Delta x^2} \notag\\
	b_{p} = \frac{\sigma^2 }{2\Delta x^2} +\frac{k(\theta-x_p)}{2\Delta x} 
\end{eqnarray}
%Recall that we have to extend the boundary conditions outside of $x\in[0,1]$. Suppose that the discretization results in a grid for $x\in[0,1] $ of $\mathcal{N}$ points. Over this grid we   
Hence, system (\ref{BTCS}) can be written in the matrix form below:
$$ M \Phi^{q+1} = \Lambda \Phi^q + b, \text{ for } q = 0, 1, \dots, T-1,$$
where $\Phi^q, b \in \mathbb{R}^{N}$ and $M \in \mathbb{R}^{N\times N}$ are given by:
\begin{eqnarray} \label{matrices_1}
	\Phi^{q} =\left( \begin{array}{c}
		\Phi^q_0 \\
		\Phi^q_1 \\
		\vdots \\
		\Phi^q_{N-1}\\
	\end{array}
	\right),
	\,\,\,\,
	b=\left( \begin{array}{c}
		0 \\
		\vdots \\
		0 \\
		b_U \\
	\end{array}
	\right),
	\,\,\,\,
	M=\left[\  \begin{array}{cccccc}
		I_L & 0_{D} & 0_U  \\
		0_L & \mathcal{M} & 0_U\\
		0_L & 0_{D} & I_U \\
	\end{array}
	\right],
\end{eqnarray}	
with $b_U = (1,\cdots,1)^T \in \mathbb{R}^U$, $I_{n},0_n$ being the $n \times n$-dimensional identity and zero matrices, respectively, $\mathcal{M} \in \mathbb{R}^{D\times D}$ given by: 
\begin{flalign}\label{M2}
	\mathcal{M}=\left(  \begin{array}{cccccccccc}
		1& 0 & 0 & 0  &  \cdots & 0  & 0 &0 \\
		-c_{1}\Delta t & 1+a_{1}\Delta t & -b_{1}\Delta t  &  0 &  \cdots & 0  & 0\\
		0 & -c_{2}\Delta t & 1+a_{2}\Delta t & -b_{2}\Delta t  & \cdots & 0 &0  & 0 \\ 
		\ddots & \ddots & \ddots & \ddots &  \ddots & \ddots & \ddots & \ddots \\
		0 & 0 & 0 & 0 & \cdots&  -c_{{N-1}}\Delta t & 1+ a_{{N-1}}\Delta t & -b_{{N-1}}\Delta t\\
		0& 0 & 0 & 0  &  \cdots & 0  & 0 &1
	\end{array}
	\right),
\end{flalign}
and $\Lambda \in \mathbb{R}^{N \times N}$:
\begin{flalign} \label{matrices_2}
	\Lambda=\left(  \begin{array}{cccccccccc}
		\lambda \Delta t \bar{f}_{-1} & \lambda \Delta t \bar{f}_{0} + \hat{F} & \lambda \Delta t \bar{f}_{1}  & \cdots  & \lambda \Delta t \bar{f}_{J/2}  & 0  & \cdots &0\\
		\lambda \Delta t \bar{f}_{-2} &\lambda \Delta t \bar{f}_{-1}  & \lambda \Delta t \bar{f}_{0} + \hat{F}  & \cdots  & \lambda \Delta t \bar{f}_{J/2-1}  & \lambda \Delta t \bar{f}_{J/2}  & \cdots& 0\\
		\ddots & \ddots & \ddots & \ddots & \ddots & \ddots & \ddots & \ddots \\
		0 & 0 &0 & \cdots & \lambda \Delta t \bar{f}_{-J/2+1 } & \lambda \Delta t \bar{f}_{-J/2} & \cdots & \lambda \Delta t \bar{f}_{J/2} \\
	\end{array}
	\right),
\end{flalign}
where $\hat{F} := 1 - \lambda \Delta t \hat{I}$.
At each time step we can then calculate $\Phi^{q+1} = M^{-1} (\Lambda \Phi^q +b)$, to obtain the solution at time $t = q+1$. 
\par For the implementation of the numerical scheme we must analyze the necessary properties pertaining to its stability and monotonicity, for which we use the same definitions and approach as in \cite{cont2005finite}. Specifically, we define these conditions as follows. 
\begin{definition} \hfill
	\begin{itemize}
		\item[$(i)$] Scheme (\ref{BTCS}) is stable if and only if, for a bounded initial condition, there exist $\Delta x$ and $\Delta t$ such that the solution exists and is bounded, i.e. $|\Phi^q_p| \leq C,$ for all $p,q$ and some $C>0$.
		\item[$(ii)$] Scheme (\ref{BTCS}) is monotone, i.e. for two initial conditions $\Phi^0$ and $\tilde{\Phi}^0$: 
		$$ \Phi^0 > \tilde{\Phi}^0 \Rightarrow  \Phi^q > \tilde{\Phi}^q,$$ for all $q$. Note that the comparison of the two vectors is to be understood in the element-by-element sense. This condition is often referred to as the discrete comparison principle.
	\end{itemize}
\end{definition}
These conditions must hold in order to avoid spurious oscillations in the numerical solutions and nonsensical values. We will show that the numerical scheme for the PIDE we have obtained is conditionally stable and monotone. As we will see from the result below, the condition is not restrictive and can easily be satisfied when selecting the parameters of the scheme, without significant computational cost. 

\begin{proposition}\label{stab_mon}
	Scheme (\ref{BTCS}) is stable and monotone if
	\begin{eqnarray}
		\Delta x \leq \frac{\sigma ^2}{k\theta }\text{ and } \Delta t \leq \frac{1}{\lambda \hat{I}}.
	\end{eqnarray}	
\end{proposition}

\begin{proof}
	We prove the two results separately, starting with stability. Let $\Phi^0$ be a bounded initial condition, i.e. $||\Phi^0||_{\infty} < \infty$. Following \cite{cont2005finite}, we will proceed by induction and contradiction. Let $||\Phi^q||_{\infty} \leq ||\Phi^0||_{\infty}$ and suppose $ ||\Phi^{q+1}||_{\infty} > ||\Phi^0||_{\infty}$. Therefore, there exists $p_0 \in \{1,\dots, N-1\}$ such that $|\Phi^{q+1}_{p_0}|=||\Phi^{q+1}||_{\infty}$ and $|\Phi^{q+1}_{p_0}| \geq |\Phi^{q+1}_{p}|$, for all $p \in \{1,\dots, N-1\}$. We will prove that this leads to a contradiction. Observing that $a_p = b_p + c_p$, we can write: 
	\begin{eqnarray}\label{res_1}
		||\Phi^{q+1}||_{\infty} = |\Phi^{q+1}_{p_0}| = \big[-c_{{p_0}}\Delta t + (1+a_{{p_0}}\Delta t) - b_{{p_0}}\Delta t \big]|\Phi^{q+1}_{p_0}|. 
	\end{eqnarray}
	To proceed, we will need coefficients $a_{{p}}, b_{{p}},c_{{p}}$ be non-negative. This is true for $a_{p}$, for all $p$. From the remaining coefficients, we obtain the condition: 
	\begin{eqnarray}\notag
		\frac{\sigma^2 }{2\Delta x^2} \geq  \frac{|k(\theta-x_p)|}{2\Delta x},
	\end{eqnarray} and hence by requiring: 
	$\sigma ^2 \geq k||\theta -x||_{\infty}  \Delta x$,
	we can ensure that the condition is satisfied at all points on the $x$-grid. Given that $\Phi$ is identically zero for $x <0$, this can be succinctly written as:
	\begin{eqnarray}\label{cond}
		\Delta x \leq \frac{\sigma ^2}{k\theta}.
	\end{eqnarray} 
	Continuing from (\ref{res_1}), and noting that all coefficients are non-negative, we now have: 
	\begin{align}
		||\Phi^{q+1}||_{\infty}  \leq -c_{{p_0}}\Delta t |\Phi^{q+1}_{p_0-1}| + (1+a_{{p_0}}\Delta t)|\Phi^{q+1}_{p_0}| - b_{{p_0}}\Delta t |\Phi^{q+1}_{p_0+1}| \notag \\ \leq |-c_{{p_0}}\Delta t \Phi^{q+1}_{p_0-1} + (1+a_{{p_0}}\Delta t)\Phi^{q+1}_{p_0} - b_{{p_0}}\Delta t \Phi^{q+1}_{p_0+1}|,
	\end{align}
	and from (\ref{ou-scheme}), in combination with the induction hypothesis, it follows that: 
	\begin{align}
		||\Phi^{q+1}||_{\infty}  \leq  |(1 - \lambda \Delta t \hat{I})\Phi^q_{p_0} + \lambda \Delta t \mathcal{I}\Phi^q_{p_0}| \leq ||\Phi^{0}||_{\infty},
	\end{align}
	where the last steps hold if $1-\lambda \Delta t\hat{I} \geq 0$, leading to the second condition. Then, since $\Phi^q_{p_0} \leq ||\Phi^q||_{\infty}$, and $\mathcal{I}\Phi^q_{p_0}\leq \hat{I}||\Phi^q||_{\infty}$, the above contradicts the assumption that $ ||\Phi^{q+1}||_{\infty} > ||\Phi^0||_{\infty}$. Hence, the scheme is stable, provided (\ref{cond}) is satisfied.
	\par Monotonicity is proven in a similar way. Specifically, let $\Phi^q,\Tilde{\Phi}^q$ be two solutions corresponding to initial conditions $\Phi^0,\Tilde{\Phi}^0$, respectively, with $\Phi^0 \geq \Tilde{\Phi}^0$ and define $d^q := \Phi^q-\Tilde{\Phi}^q$. We will proceed by induction and contradiction, as above. We have that $d^0 \geq 0$ and also assume $d^q \geq 0 $. We suppose that $d^{p+1} <0$, i.e., there exists $p_0$ such that $\inf_{\substack{p}} d^{q+1}_p = d^{q+1}_{p_0} <0$. Then:
	\begin{flalign}\label{res_2}
		\inf_{\substack{p}} d^{q+1}_p = d^{q+1}_{p_0} = -c_{{p_0}}\Delta t d^{q+1}_{p_0} + (1+a_{{p_0}}\Delta t) d^{q+1}_{p_0} - b_{{p_0}}\Delta t d^{q+1}_{p_0} \notag\\ \geq 
		-c_{{p_0}}\Delta t d^{q+1}_{p_0-1} + (1+a_{{p_0}}\Delta t)d^{q+1}_{p_0} - b_{{p_0}}\Delta t d^{q+1}_{p_0+1} 
		= (1 - \lambda \Delta t \hat{I}) d^{q}_{p_0} + \lambda \Delta t \mathcal{I}d^q_{p_0} \geq 0,
	\end{flalign} 
	where the last step follows from the induction hypothesis and we have supposed that condition (\ref{cond}) is satisfied. By contradiction, we therefore conclude that $d^{p+1} >0$, as required. 
\end{proof}

\subsection{Regime switching}

We now turn to the regime switching model, with regimes $r \in \mathcal{R}$ and a total of $R$ regimes. In the BTCS discretized version of $(\ref{int-markov})$ we let $\Phi^q_{p,r}$ represent the survival probability at the grid point $t_q=t_0+ q\Delta t$, $x_p=x_0+p\Delta x$, when the underlying Markov process is originally in state $r \in \mathcal{R}$, i.e., $\Phi^q_{p,r}=\Phi(x_p,r,t_q)$, with $p=0,1,\dots N-1, q=0,1,\dots T$, $r \in \mathcal{R}$. The discretized PIDE can be written as:
\begin{flalign}\label{ou-rs-scheme}
	-\Phi_{p-1,r}^{q+1} c_{p,r} \Delta t + \Phi_{p,r}^{q+1}\big(1+ a_{p,r}\Delta t  \big) - \Phi_{p-1,r}^{q+1} b_{p,r} &\Delta t - \sum_{j\neq r}q_{r j}\Phi^{q+1}_{p,j}  \Delta t \notag \\ &= (1- \lambda \Delta t \hat{I})\Phi_{p,r}^q + \lambda \Delta t \mathcal{I}\Phi_{p,r}^q %\Phi_{p,r}^q + \lambda \Delta t I_{p,r}^q,
\end{flalign}
where $\mathcal{I}\Phi^q_{p,r}$ and $\hat{I}$ are as in (\ref{integral_term}) and (\ref{int_2}), respectively. The coefficients of this scheme are then given by:
\begin{eqnarray}
	c_{p,r} = \frac{\sigma_r^2}{2\Delta x^2} - \frac{k_r(\theta_r-x_p)}{2\Delta x} \notag\\
	a_{p,r} = \frac{\sigma_r^2 }{\Delta x^2} + \sum_{j\neq r} q_{\rho j} \notag\\
	b_{p,r} = \frac{\sigma_r^2}{2\Delta x^2} + \frac{k_r(\theta_r-x_p)}{2\Delta x}.
\end{eqnarray}
In matrix notation, the regime switching PIDE can be written as:
\begin{eqnarray} \label{matrix-form}
	M^{RS}\Phi^{q+1} = \Lambda^{RS} \Phi^q+b^{RS}, 
\end{eqnarray}
where the block-form matrices  $\Phi, M^{RS}, \Lambda^{RS} \in \mathbb{R}^{NR \times NR}$ and $b^{RS} \in \mathbb{R}^{NR}$ are given by: 
\begin{eqnarray} \label{matrices}\notag
	\Phi^{q} =\left[ \begin{array}{c}
		\Phi^q_{\cdot,1} \\
		\Phi^q_{\cdot,2} \\
		\vdots\\
		\Phi^q_{\cdot,R}\\
	\end{array}
	\right], \,\,\
	b^{RS}=\left[\begin{array}{c}
		b \\
		b\\
		\vdots \\
		b \\
	\end{array}
	\right], \,\
	\Lambda^{RS}=\left[  \begin{array}{cccc}
		\Lambda  &0_{N} & \cdots & 0_{N} \\
		0_{N} &\Lambda& \cdots & 0_{N}  \\
		\vdots & \vdots & \vdots & \vdots \\
		0_{N} & 0_{N} &  \cdots  & \Lambda \\
	\end{array}
	\right], 
\end{eqnarray}
\begin{eqnarray}
	M^{RS}=\left[  \begin{array}{cccc}
		M_{r_1} & -\Delta t q_{r_1r_2} I_{N} & \cdots  & -\Delta tq_{r_1r_R} I_{N} \\
		-\Delta t q_{r_22r_1} I_{N}  & M_{r_2} & \cdots &  -\Delta t q_{r_2r_R} I_{N}  \\
		\vdots & \vdots & \vdots & \vdots \\
		-\Delta t q_{r_Rr_1} I_{N}  & -\Delta t q_{r_Rr_2} I_{N} & \cdots &  M_{r_R}\\
	\end{array}
	\right],
\end{eqnarray}
with $\Phi_{\cdot, r} ^ q = (\Phi_{0,r}^q, \dots \Phi_{N-1,r}^q)^T$, $b, \Lambda,$ as in (\ref{matrices_1}) and (\ref{matrices_2}) and the regime-specific matrices $M_{r_i}\in \mathbb{R}^{N \times N}$ for $r_i \in \mathcal{R}, i=1,\dots, R$, are as in (\ref{matrices_1}) by replacing $a_p, b_p,c_p$ with $a_{p,r},b_{p,r},c_{p,r}$.
\begin{comment}	\begin{eqnarray}
M_\rho=\left(  \begin{array}{cccccccccc}
1 & 0 & 0 & 0 & \cdots & 0 & 0  & 0\\
-c_{1,\rho}\Delta t & 1+a_{1,\rho}\Delta t & -b_{1,\rho}\Delta t  &  0 &  \cdots & 0  & 0 &0 \\
\ddots & \ddots & \ddots & \ddots &  \ddots & \vdots & \vdots & \vdots \\
0 & 0 & 0 & 0 & 0 &  -c_{{N-1},\rho}\Delta t & 1+ a_{{N-1},\rho}\Delta t & -b_{{N-1},\rho}\Delta t\\
0 & 0 & 0 & 0 & 0 & 0 & 0&1\\
\end{array}
\right).
\end{eqnarray} \end{comment}
As in the discretization of the PIDE for the one dimensional OU model, we will have to prove the appropriate stability and monotonicity results for the regime switching model. %Also note that when $R_t=3$, corresponding to the absorbing default state, the transition probability to any other state and the survival probability is zero.
\begin{lemma}\label{num-rs}
	Scheme (\ref{ou-rs-scheme}) is stable and monotone if
	\begin{eqnarray}\label{cond_2}
		\Delta x \leq \frac{\sigma_r^2} {k_r \theta_r }\text{ and } \Delta t \leq \frac{1}{\lambda \hat{I}}
	\end{eqnarray}
	for all $r \in \mathcal{R}$.
\end{lemma}
\begin{proof}
	Let $\Phi^0$ be a bounded initial condition for the survival probability. Note that this initial condition accounts for all regimes. As above, we will proceed by induction and contradiction. Let $||\Phi^q||_{\infty} \leq ||\Phi^0||_{\infty}$ and suppose $ ||\Phi^{q+1}||_{\infty} > ||\Phi^0||_{\infty}$. In this case, this means that there exists $(p_0, r_0) \in \{0,1,\dots, N-1 \} \times  \mathcal{R}$ such that $|\Phi^{q+1}_{p_0,r_0}|=||\Phi^{q+1}||_{\infty}$, with $|\Phi^{q+1}_{p_0,r_0}| \geq |\Phi^{q+1}_{p,r}|$, for all $(p,r) \in \{0,1,\dots, N-1\} \times \mathcal{R}$. Hence:
	\begin{flalign}
		&||\Phi^{q+1}||_{\infty} =|\Phi_{p_0, r_0}^{q+1}|=\big[-c_{p_0,{r_0}}\Delta t + (1+a_{p_0,{r_0}}\Delta t) - b_{{p_0},{r_0}}\Delta t - \sum_{j\neq r_0} q_{r_0j} \Delta t \big] |\Phi^{q+1}_{p_0, r_0}| \notag \\& \leq -c_{{p_0},{r_0}} |\Phi^{q+1}_{p_0-1,{r_0}}|\Delta t + (1+a_{{p_0},{r_0}}\Delta t)|\Phi^{q+1}_{p_0,{r_0}}| - b_{{p_0},{r_0}} |\Phi^{q+1}_{p_0+1,{r_0}}| \Delta t - \sum_{j\neq r_0} q_{r_0j} |\Phi^{q+1}_{p_0,j}|\Delta t \notag \\& \leq |-c_{{p_0},{r_0}} \Phi^{q+1}_{p_0-1,{r_0}}\Delta t + (1+a_{{p_0},{r_0}}\Delta t)\Phi^{q+1}_{p_0,{r_0}} - b_{{p_0},{r_0}} \Phi^{q+1}_{p_0+1,{r_0}}\Delta t - \sum_{j\neq r}q_{r_0 j}\Phi^{q+1}_{p_0,j} \Delta t| \notag \\& \leq |(1 - \lambda \Delta t \hat{I})\Phi^q_{p_0,r_0} + \lambda \Delta t \mathcal{I}\Phi^q_{p_0,r_0}| \leq ||\Phi^{0}||_{\infty}, \notag
	\end{flalign}
	where the last inequality follows from the same calculations as in Proposition \ref{stab_mon}. In the above, we must have $a_{p,r}, b_{p,r}, c_{p,r}>0$ for each regime $r \in \mathcal{R}$, leading to the first condition in (\ref{cond_2}).
	\par For monotonicity, again let $\Phi^q,\Tilde{\Phi}^q$ be two solutions corresponding to $\Phi^0,\Tilde{\Phi}^0$, respectively, with $\Phi^0 \geq \Tilde{\Phi}^0$. Assume $d^q := \Phi^q-\Tilde{\Phi}^q>0$ and suppose $d^{q+1} \leq 0 $. Hence, there exists $p_0,r_0$ such that $\inf_{\substack{p,r}} d^{q+1}_{p,r} = d^{q+1}_{p_0,r_0} <0$. Proceeding as in Proposition \ref{stab_mon}:
	\begin{flalign}
		&\inf_{p,r}d^{q+1}_{p,r} = d^{q+1}_{p_0,r_0} =\big[-c_{{p_0},{r_0}}\Delta t + (1+a_{{p_0},{r_0}}\Delta t)  - \sum_{j\neq r_0} q_{r_0 j} \Delta t- b_{{p_0},{r_0}}\Delta t \big] d^{q+1}_{p_0, r_0} \notag \\& \geq  -c_{{p_0},{r_0}}d^{q+1}_{p_0-1,{r_0}}\Delta t + (1+a_{{p_0},{r_0}})d^{q+1}_{p_0,{r_0}}\Delta t - \sum_{j\neq r_0} q_{r_0 j} d^{q+1}_{p_0,j} \Delta t - b_{{p_0},{r_0}} d^{q+1}_{p_0+1,{r_0}}\Delta t \notag \\& = (1 - \lambda \Delta t \hat{I}) d^{q}_{p_0,r_0} + \lambda \Delta t \mathcal{I}d^q_{p_0,r_0}  \geq 0. \notag
	\end{flalign}
	Stability and monotonicity for scheme (\ref{ou-rs-scheme}) thus follow by contradiction. 
\end{proof}

%We complete this section with an example in the context of asset processes in an IFRS-9 compliant regime.
\subsection{Stochastic volatility model}
Finally, we present the numerical scheme for model (\ref{pide-sv}). For this case, we must consider a discretization of the volatility process $y \in \mathcal{V}$ of size $V$. For the numerical implementation we use $y \in [0, Y_{\max}]$ for some appropriate value $Y_{\max}$. As above, we adopt the notation $\Phi^q_{p,j}$ for the survival probability at the grid point $t_q=t_0+ q\Delta t$, $x_p=x_0+p\Delta x$, and $y_j = y_0 +j\Delta y$, i.e. $\Phi^q_{p,j}=\Phi(x_p,y_j,t_q)$, with $p=0,1,\dots N-1, q=0,1,\dots T$ and $j=0,1, \dots, V-1$. 
\par In this case the discretization scheme requires an alternative approach. Specifically, when the coefficient of the diffusion term becomes 0 or $y \rightarrow 0$, the analogous to the previous stability and monotonicity conditions fails. The solution to this is to consider an Alternating Direction Implicit (ADI) approximation to the first derivative terms corresponding to both the asset and CIR volatility processes, as shown below:
\begin{eqnarray}
	\frac{\partial \Phi}{\partial x}\approx
	\begin{cases}
		\frac{\Phi_{p+1,j}^{q+1}-\Phi_{p,j}^{q+1}}{\Delta x},\,\,\ \text{ if } k(\theta - x_p) \geq 0 \\
		\frac{\Phi_{p,j}^{q+1}-\Phi_{p-1,j}^{q+1}}{\Delta x} ,\,\,\ \text{ if } k(\theta - x_p) < 0		
	\end{cases}
\end{eqnarray}
\begin{eqnarray}
	\frac{\partial \Phi}{\partial y}\approx
	\begin{cases}
		\frac{\Phi_{p,j+1}^{q+1}-\Phi_{p,j}^{q+1}}{\Delta y},\,\,\ \text{ if } \kappa(\mu - y_j) \geq 0 \\
		\frac{\Phi_{p,j}^{q+1}-\Phi_{p,j-1}^{q+1}}{\Delta y} ,\,\,\ \text{ if } \kappa(\mu - y_j) < 0		
	\end{cases}
\end{eqnarray}
Furthermore recall that at the boundary $y = Y_{\max}$ we have:
\begin{eqnarray}
	\frac{\Phi_{p,V}^{q+1}-\Phi_{p,V-2}^{q+1}}{2\Delta y} \approx \frac{\partial \Phi}{\partial y}(x,Y_{\max},u) = 0,
\end{eqnarray}
and therefore $\Phi_{p,V} \approx \Phi_{p,V-2}$. This allows us to approximate the second derivative at the boundary by:
\begin{eqnarray}
	\frac{\Phi_{p,V}^{q+1}-2\Phi_{p,V-1}^{q+1} + \Phi_{p,V-1}^{q+1}}{\Delta y^2} = \frac{2\big(\Phi_{p,V-2}^{q+1}-2\Phi_{p,V-1}^{q+1}\big)}{\Delta y^2} 
\end{eqnarray}
We can now write the implicit scheme for the PIDE corresponding to the stochastic volatility model:
\begin{flalign}\label{ou-sv-scheme}
	- \Phi_{p-1,j}^{q+1} c_{p,j} \Delta t   + \Phi_{p,j}^{q+1}\big(1+ a_{p,j}\Delta t   \big) - \Phi_{p-1,j}^{q+1}& b_{p,j} \Delta t - \Phi^{q+1}_{p,j-1} e_{p,j}  \Delta t - \Phi^{q+1}_{p,j+1} f_{p,j}  \Delta t \notag \\&
	= (1- \lambda \Delta t \hat{I})\Phi_{p,j}^q + \lambda \Delta t \mathcal{I}\Phi_{p,j}^q,
\end{flalign}
where the coefficients are given by:
\begin{flalign}
	c_{p,j} &= \frac{y_j}{2 \Delta x^2} - \frac{k(\theta-x_p)}{\Delta x}\mathds{1}_{\{k(\theta-x_p)<0\}}, \notag\\
	a_{p,j} &= \frac{y_j}{\Delta x^2} + \frac{\xi^2 y_j}{\Delta y^2} + \abs*{\frac{k(\theta-x_p)}{\Delta x}} + \abs*{\frac{\kappa(\mu-x_p)}{\Delta y}},\notag\\
	b_{p,j} &= \frac{y_j}{2 \Delta x^2} +  \frac{k(\theta-x_p)}{\Delta x}\mathds{1}_{\{k(\theta-x_p)>0\}}, \notag\\
	e_{p,j} &= \frac{\xi^2 y_j}{2\Delta y^2}\mathds{1}_{\{y \neq Y_{\max}\}} + \frac{\xi^2 y_j}{\Delta y^2}\mathds{1}_{\{y = Y_{\max}\}}- \frac{\kappa(\mu-x_p)}{\Delta y}\mathds{1}_{\{\kappa(\mu-y_j)<0 \,\ \cap \,\ y\neq 0 \,\ \cap \,\ y\neq Y_{\max}\}} , \notag\\
	f_{p,j}& = \frac{\xi^2 y_j}{2\Delta y^2}\mathds{1}_{\{y \neq Y_{\max}\}} + \frac{\kappa(\mu-x_p)}{\Delta y}\mathds{1}_{\{\kappa(\mu-y_j)>0 \,\ \cap \,\ y\neq Y_{\max}\}} .
\end{flalign}
The solution of the resulting scheme:
$$ M^{SV}\Phi^{q+1} = \Lambda^{SV} \Phi^q+b^{SV}$$ will result in estimations of the survival probability at each state of the underlying stochastic volatility process. Therefore, as in the regime switching model, we obtain the vectors $\Phi^q = [\Phi^q_{\cdot, 0}\,\ \Phi^q_{\cdot,1} \cdots \Phi^q_{\cdot,V-1}]^T \in \mathbb{R}^{NV \times NV}$, $b^{SV} = [b \,\ b \cdots b]^T \in \mathbb{R}^{NV}$ and matrix  $M^{SV}\in \mathbb{R}^{NV \times NV}$ as given below, in block form:
\begin{eqnarray} \notag
	M^{SV}=\left[  \begin{array}{ccccc}
		M_0 & -\Delta t f_{p,0}I_{N} &  0_{N} & \cdots &0_{N} \\
		-\Delta t e_{p,1} I_{N}  & M_1 & 	-\Delta t f_{p,1} I_{N}  &  \cdots & 0_{N}  \\
		\vdots & \vdots & \vdots & \vdots \\
		0_{N} & \cdots & -\Delta t e_{p,V-2} I_{N}  &M_{V-2} & -\Delta t f_{p,V-2} I_{N}   \\
		0_{N}  & \cdots & 0_N & -\Delta t e_{p,V-1} I_{N} &  M_{V-1} 
	\end{array}
	\right], \notag
\end{eqnarray}
where $M_j\in \mathbb{R}^{N \times N}$
\begin{comment}	\begin{eqnarray}
\Tilde{M}_j=\left(  \begin{array}{cccccccccc}
1 & 0 & 0 & 0 & \cdots & 0 & 0  & 0\\
-c_{1,j}\Delta t & 1+a_{1,j}\Delta t & -b_{1,j}\Delta t  &  0 &  \cdots & 0  & 0 &0 \\
\ddots & \ddots & \ddots & \ddots &  \ddots & \vdots & \vdots & \vdots \\
0 & 0 & 0 & 0 & 0 &  -c_{{N-1},j}\Delta t & 1+ a_{{N-1},j}\Delta t & -b_{{N-1},j}\Delta t\\
0 & 0 & 0 & 0 & 0 & 0 & 0&1\\
\end{array}
\right),
\end{eqnarray}
\end{comment}
for $j=0,1,\dots, V-1$ is given by (\ref{matrices_1}) by replacing $a_p, b_p,c_p$ with $a_{p,j},b_{p,j},c_{p,j}$ and $\Lambda^{SV} \in \mathbb{R}^{NV \times NV}$ is in the same form as $\Lambda^{RS}$ in (\ref{matrices}).
We now prove the required stability and monotonicity results for the stochastic volatility case.
\begin{lemma}\label{num-sv}
	Scheme (\ref{ou-sv-scheme}) is unconditionally stable and monotone.
\end{lemma}
\begin{proof}
	The proof follows almost identically to the regime switching case. Again, let $\Phi^0$ be a bounded initial condition for the survival probability, $||\Phi^q||_{\infty} \leq ||\Phi^0||_{\infty}$ and suppose $ ||\Phi^{q+1}||_{\infty} > ||\Phi^0||_{\infty}$. Hence, there exists $(p_0,j_0) \in \{0,1,\dots, N-1\} \times \{0,1,\dots, V-1\}$ such that $|\Phi^{q+1}_{p_0,j_0}|=||\Phi^{q+1}||_{\infty}$, with $|\Phi^{q+1}_{p_0,j_0}| \geq |\Phi^{q+1}_{p,j}|$, for all $p \in \{0,1,\dots, N-1\} \times  \{0,1,\cdots,V-1\}$. Hence, noting that no conditions on $\Delta x, \Delta t$ need be imposed since all the coefficients are positive by construction, we have:
	\begin{flalign} 
		&||\Phi^{q+1}||_{\infty} =|\Phi_{p_0, j_0}^{q+1}|=\big[-c_{{p_0},{j_0}}\Delta t + (1+a_{{p_0},{j_0}}\Delta t)  - b_{{p_0},{j_0}}\Delta t -e_{p_0,{j_0}} \Delta t -f_{{p_0},{j_0}} \Delta t \big] |\Phi^{q+1}_{p_0, j_0}| \notag \\ &\leq
		-  c_{{p_0},{j_0}}  |\Phi_{p_0-1,{j_0}}^{q+1}|\Delta t  + \big(1+ a_{{p_0},{j_0}}\Delta t  \big)|\Phi_{p_0,j_0}^{q+1}| - b_{{p_0},{j_0}} |\Phi_{p_0+1,j_0}^{q+1}|\Delta t \notag \\ &- e_{{p_0},{j_0}} |\Phi^{q+1}_{p_0,j_0-1}|  \Delta t - f_{{p_0},{j_0}}|\Phi^{q+1}_{p_0,j_0+1}|\Delta t
		\leq |(1 - \lambda \Delta t \hat{I})\Phi^q_{p_0,j_0} + \lambda \Delta t \mathcal{I}\Phi^q_{p_0,j_0}| \leq ||\Phi^0||_{\infty}, \notag
	\end{flalign}
	\par To conclude, we prove the scheme is monotone. Consider initial conditions $\Phi^0,\Tilde{\Phi}^0$, respectively, with $\Phi^0 \geq \Tilde{\Phi}^0$. With $d^q := \Phi^q-\Tilde{\Phi}^q >0$, we also suppose $d^{q+1} < 0$, i.e., there exists a pair $(p_0,j_0)$ such that $\inf_{\substack{p,j}} d^{q+1}_{p,j} = d^{q+1}_{p_0,j_0} <0$, and calculate:
	\begin{flalign} 
		\inf_{p,j}d^{q+1}_{p,j} &= d^{q+1}_{p_0,j_0} =\big[-c_{{p_0},{j_0}}\Delta t + (1+a_{{p_0},{j_0}}\Delta t)  - b_{{p_0},{j_0}}\Delta t -e_{p_0,{j_0}} \Delta t -f_{p_0,{j_0}} \Delta t \big] d^{q+1}_{p_0, j_0} \notag \\ &\geq  -c_{{p_0},{j_0}} d^{q+1}_{p_0-1,{r_0}}\Delta t + (1+a_{{p_0},{j_0}}\Delta t  )d^{q+1}_{p_0,{j_0}} - b_{{p_0},{j_0}} d^{q+1}_{p_0+1,{j_0}}\Delta t \nonumber \\ &- e_{p_0,{j_0}}d^{q+1}_{p_0,j_0-1}\Delta t - f_{p_0,j_0}d^{q+1}_{p_0,j_0+1}\Delta t  =(1 - \lambda \Delta t \hat{I}) d^{q}_{p_0,j_0} + \lambda \Delta t \mathcal{I}d^q_{p_0,j_0}  \geq 0. \notag
	\end{flalign}
\end{proof}

\begin{remark}
	It is worth noting that the resulting system for the PD function is dense due to the jump integral term, adding to the computational complexity of the scheme. 
	Hence, additional methods such as implicit handling of the jump term and/or Crank-Nicolson schemes can be useful. We omit these methods from the present work, as they are not the main focus, however we refer the interested reader to relevant research, such as \cite{d2005robust}, \cite{jwo2020investigation} and \cite{carr2007numerical}. 
\end{remark}

\begin{remark}
	As previously mentioned, for the credit risk modelling tasks we will consider, using either the regime switching or the stochastic volatility model suffices. We will see multiple such examples in Section \ref{ifrs-section}. Similar calculations as those considered above can be applied to PIDE (\ref{pide-gen}), for the estimation of the survival probability under the generalized model. Stability and monotonicity follow from combining Lemmata \ref{num-rs} and \ref{num-sv}. However, a mentioned, the combination of the regime switching and stochastic volatility variables lead to an intractable numerical scheme, plagued with the "curse of dimensionality".   
\end{remark}

\section{Applications in credit risk}\label{ifrs-section}
\subsection{IFRS 9 provision calculations}

As discussed, the IFRS 9 framework requires practitioners to take into consideration multiple risk factors and their evolution for provision calculations and other modelling tasks. Naturally, the evolution of the PD is of paramount importance in these credit risk problems. Specifically for provisioning, using the PD function we can now estimate provisions for Stage 1 and Stage 2 exposures. Recall that financial institutions must account for additional provisions for exposures which display a significant increase in credit risk. These forward-looking lifetime provisions must be calculated per exposure, with some minor differences depending on the type of portfolio (e.g., for corporate loan portfolios many consider contamination effects). In this section, we display how the framework outlined above can be used to calculate the provisions under IFRS 9. The main contribution is the calculation of Expected Lifetime provisions for Stage 2 exposures which, as previously stated, is a novel requirement introduced by the these regulatory standards. We provide specific examples of provision calculations for each case below.  
\par These calculations depend on multiple risk parameters corresponding to the credit exposure; the PD and Loss Given Default (LGD), as well as the amortization schedule, which affects the Exposure at Default (EAD), i.e., the remaining value of the loan which is not repaid in the case of default. Naturally, these risk parameters may vary according to each application and case. For example, many consider the LGD to evolve according to some stochastic process with a correlation to the PD (see e.g., \cite{miu2006basel}, \cite{witzany2011two}). As our main focus is the PD function, we will consider a constant LGD and typical amortization schedule under the assumption of a zero interest rate in the examples that follow. The methodology, however, can be generalized to also consider a an appropriate LGD function (or stochastic process ) and any type of amortization. 
\subsubsection{Stage 1 provisions}
For Stage 1 loans standard regulations apply and we need only consider provisions as the Expected Losses (EL) that can be incurred on the current exposure. This calculation is given by the simple formula:
\begin{eqnarray}\label{stage_1_losses}
	EL:=\mathbb{E}[L_t] = EAD_t LGD_t PD_t.
\end{eqnarray}
Using the implicit numerical schemes we can calculate the PD value representing the probability the loan defaults within some fixed time $t$, represented straightforwardly by $$PD_t = \Psi(x,t).$$ For example, the probability of a default event occurring within the current unit of time (typically year) is $PD_1 = \Psi(x,1).$ An example of the provision calculation for varying initial positions is given below.

%\begin{example}\label{ex}
%\end{example} \hfill $\triangleleft$
\begin{example}\label{stage1}
	Consider a Stage 1 loan, with asset process given by:
	\begin{eqnarray}
		dG_t=k(\theta - G_t)dt + \sigma dB_t +  \int_{\mathbb{R}} z N(dt,dz), \,\ G_0 = x,
	\end{eqnarray}
	where $(k, \theta, \sigma ) = (0.5, 3.5,2.0)$, the Compound Poisson Process has normally distributed jumps, with size $Z \sim N(0.0,0.2)$ and rate $\lambda = 1.0$. We select $\mathcal{D} = [-10,10]$, with $S = 8.0$, as estimated by Monte Carlo experiments, $N = 1001$ and $T=101$. The resulting survival probability is graphed in Fig. $\ref{fig:btcs_fig_1}$.
	\begin{figure}[h ]
		\begin{center}
			\includegraphics[height= 63mm, width=75mm,scale=0.5]{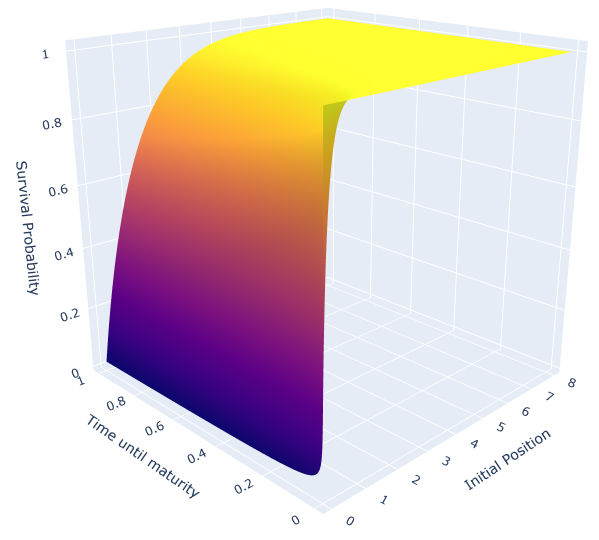}
			\includegraphics[height=63mm, width=72mm, scale=0.5]{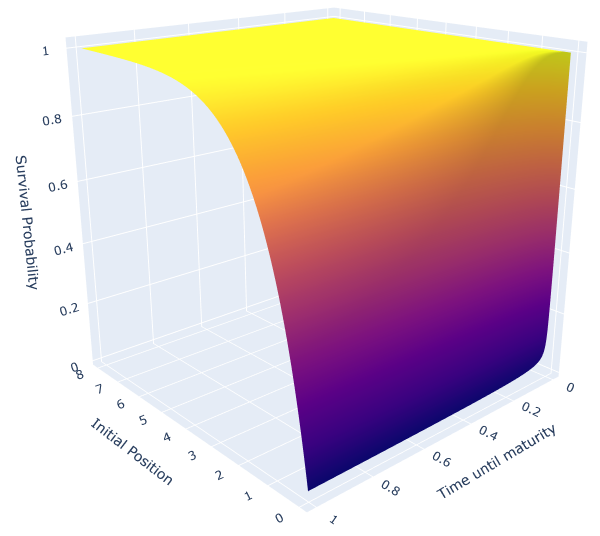}
		\end{center}
		\caption{Survival probability for asset process with $(k, \theta, \sigma ) = (0.5, 3.5,2.0)$, using the BTCS scheme with $J=150$.}
		\label{fig:btcs_fig_1}
	\end{figure}
	Then, depending on the initial position of the asset at the time of loan origination, the provisions are calculated as
	$EL = 100\cdot PD \cdot 75\%$. In Table \ref{Tab:prov_s1} we present the results for some initial positions $x\in[0,1]$.
	\begin{center}
		\begin{tabular}{@{} *{7}{c} @{}}
			\headercell{Initial Position} & $PD_1 (\%)$ & $LGD (\%)$ & Provisions (\%) \\
			%\cmidrule(l){2-7}
			\midrule
			$0.0$  & 100.00 &  75 &  75.00  \\
			$0.1$  & 83.24 &  75 &  62.43  \\
			$0.2$  & 68.22 &  75 & 51.17  \\
			$0.3$  & 55.01 &  75 &  41.26  \\
			$0.4$  & 43.64	&  75 &  32.73 \\
			$0.5$  & 34.06 &  75 &  25.55  \\
			$0.6$  & 26.16 &  75 & 19.62 \\
			$0.7$  & 19.78 &  75 &  14.84  \\
			$0.8$& 14.73 &  75 & 11.05   \\
			$0.9$  & 10.82 &  75 &  8.11  \\
			$1.0$ & 7.83  &  75 & 5.87   \\
		\end{tabular}
		\captionof{table}{Provision calculations for Stage 1 loan, given initial position of the asset process.}
		\label{Tab:prov_s1}
	\end{center}
	Hence, if, at the time of calculation, the asset process is estimated to start at $x=0.6$, the provisions are $19.62\%$ of the current exposure.
	\noindent
	\hfill $\triangleleft$
\end{example}
As expected, the provisions are a decreasing function of the initial position. A similar table can be produced at any point during the lifetime of the loan, by estimating the corresponding PD values. 

\subsubsection{Stage 2 Provisions}
We now turn to loan provision calculations for Stage 2 loans. Under IFRS 9, if and when loans transition to Stage 2, the lender is obligated to consider all future losses for provisioning purposes. Hence, at any time $t$, and assuming discrete amortization payments, the following formula for the expected losses occuring at some time $i > t$ applies:
\begin{eqnarray}
	\mathbb{E}[L_i] = \mathbb{E} \Big[ \frac{1}{(1+r_i)^{i-t}} LGD_i PD^{PiT}_i EAD_i | \mathcal{F}_t\Big].
\end{eqnarray}
The corresponding formula for the Lifetime Expected Credit Losses (ECL - also referred to as Expected Lifetime Provisions) at time $t$ is given by:
\begin{eqnarray}\label{ECL}
	ECL_t =  \mathbb{E} \Big[\sum_{i=t+1}^{T} \frac{1}{(1+r_i)^{i-t}} LGD_i PD^{PiT}_i EAD_i | \mathcal{F}_t\Big],
\end{eqnarray}
where $r_i$ is the interest rate at time $i$ and $T$ 
the maturity. In the above, $PD^{PiT}_i$ represents the conditional Point-in-Time PD, which is the probability of default occurring at a given future time period. Specifically, we define:
\begin{eqnarray} \label{pitpd}
	PD^{PiT}_u =  \mathbb{P}\Big(\inf_{\substack{r \leq u}} G_r^x \leq 0, \inf_{\substack{r \leq u-1}} G_r^x > 0\Big).
\end{eqnarray} 
%for $t \geq s $. Without loss of generality, we will consider a starting time $t=0$, which also allows us to adopt the generalized notation.  
In order to calculate the $PD^{PiT}$ in terms of the PD function resulting from the solution of the PIDE, we note that:
$$  \mathbb{P}\Big(\inf_{\substack{r \leq u}} G_r^x \leq 0\Big) = \mathbb{P}\Big(\inf_{\substack{r \leq u-1}} G_r^x \leq 0\Big) + \mathbb{P}\Big(\inf_{\substack{ r \leq u}} G_r^x \leq 0, \inf_{\substack{ r \leq u-1}} G_r^x > 0\Big),$$ and hence:
$$PD^{PiT}_u  = \Psi(x,u)  -\Psi(x,u-1) = \Phi(x,u-1)  -\Phi(x,u).$$ %where we have written $\Psi, \Phi$ rather than $\tilde{\Psi},\tilde{\Phi}$, for brevity. 
Therefore, (\ref{ECL}) can now be written as:
\begin{eqnarray} \label{ecl2}
	ECL _t= \sum_{i=1}^{T}\big(\Phi(x,i-1)  -\Phi(x,i)\big)\mathbb{E} \Big[\frac{LGD_i  EAD_i}{(1+r_i)^{i-t}} | \mathcal{F}_t \Big].
\end{eqnarray}
\begin{example}\label{stp}
	Consider a credit exposure with asset process as in Example \ref{stage1}. However, we now suppose the exposure has been transferred to Stage 2, with remaining maturity $T = 10$ years. We also consider that the asset process of the borrower is currently estimated at $x=1.80$. To estimate the Stage 2 provisions we require the PD function and use ($\ref{ecl2}$) (as mentioned, we consider $r=0$ for simplicity):\\
	
	\begin{comment}
	\begin{center}
	\begin{tabular}{@{} *{8}{c} @{}}
	\headercell{Time} &$EAD_t$ & $\Psi(x,t)$ & $PD_t^{PiT}$ & $LGD$ & $ECL_t$\\
	%\cmidrule(l){2-7}
	\midrule
	$1.0$  & 100 &  0.2464  &  0.2464 &75   &  18.48   \\
	$2.0 $  & 90 &  0.4346 &  0.1882  &75  & 12.70   \\
	$3.0$ & 80 &0.5497 &  0.1151 & 75  & 6.91  \\
	$4.0$& 70  & 0.6287	  &  0.0790 & 75  &  4.15 \\
	$5.0$  & 60 &  0.6869  &  0.0582 &75   &  2.62  \\
	$6.0 $  & 50 &   0.7313&  0.0444  &75  & 1.67   \\
	$7.0$ & 40 &0.7661 &  0.0348 & 75  &  1.04 \\
	$8.0$& 30  & 0.7937	  &  0.0276 & 75  &  0.62  \\
	$9.0$ & 20 &0.8159&  0.0222 & 75  & 0.33  \\
	$10.0$& 10  & 0.8338&  0.0179 & 75  &  0.13  \\ \hline 
	ECL & & & & & 48.65  
	\end{tabular}
	\captionof{table}{Provision calculations for Stage 2 loan, given an asset process with initial position $x=0.2$.}
	\end{center}
	\end{comment}
	
	\begin{center}
		\begin{tabular}{@{} *{8}{c} @{}}
			\headercell{\thead{Time until \\  maturity ($u$)}} &$EAD_u$ & $PD_u (\%)$ & $PD_u^{PiT}(\%)$ & $LGD(\%)$ & $EL_t$\\
			%\cmidrule(l){2-7}
			\midrule
			$10.0$& 100  & 21.69&  1.59 & 75  &  1.19  \\  
			$9.0$ & 90 & 20.10&  1.76 & 75  & 1.19  \\
			$8.0$& 80  & 18.34	  &  1.97 & 75  &  1.18  \\
			$7.0$ & 70 & 16.37 &  2.20 & 75  &  1.16 \\
			$6.0 $  & 60 & 14.17&  2.49  & 75  & 1.12   \\
			$5.0$  & 50 &  11.68  &  2.81 &75   &  1.05  \\
			$4.0$& 40  & 8.87	  &  3.10 & 75  &  0.93 \\
			$3.0$ & 30 & 5.77 &  3.14 & 75  & 0.71  \\
			$2.0 $  & 20 &  2.63 &  2.24  &75  & 0.34  \\
			$1.0$  & 10 &  0.39  &  0.39 &75   &  0.03  \\ \hline
			
			$ECL$ & & & & & 8.89  
		\end{tabular}
		\captionof{table}{Provision calculations for Stage 2 loan, given an asset process with initial position $x=1.80$.}
	\end{center}

	The exposure and expected losses are in percentages of the remaining exposure. As shown, the current Lifetime provisions are given by the sum of the final column: $ECL = 8.89\%$. \hfill
	$\triangleleft$
\end{example}
In practice, we expect that when an exposure is classified as Stage 2, the parameters of the underlying asset process may differ, compared to the Stage 1 counterpart. In the example above, we purposely consider the same asset process so as to highlight the differences in the final provision estimations. 

\subsubsection{Provisions under the regime switching model}
\par As discussed, the new regulatory framework aims to ensure that all financial institutions have accounted for future losses and abrupt changes in credit risk parameters, which can create severe losses and subsequent liquidity and solvency issues, both for institutions and their customers. As risk classification is widely considered a Markov process both in theory and by practitioners, considering transiton probabilities for loans can allow us to forececast PD values and estimate worst-case scenario provisions for loan exposures. We note that, in practice, estimating the parameters of the asset prices under each regime may be difficult. However, many financial institutions consider such models and its mathematical framework is well established, see e.g., \cite{chatterjee2015centre} and \cite{bruche2005estimating}. Another approach is to use historical parameters from Stage 1 and Stage 2 loans to estimate the changes that occur when a loan transitions between Stages. For this example, we will be estimating the provisions under the regime switching model developed above. To this end, we consider an IFRS 9 compliant transition matrix:
$$
P=\left(  \begin{array}{c|ccc}
	\text{IFRS 9 Rating} & \text{Stage 1} & \text{Stage 2} & \text{Stage 3}\\
	\hline
	\text{Stage 1} & p_{11} & p_{12} & p_{12}\\
	\text{Stage 2} & p_{12} & p_{22} & p_{13}  \\
	\text{Stage 3} & 0 & 0 &1
\end{array}
\right).
$$
For a loan originating in Stage 1, we can now forecast credit losses by taking into account the probability of a SICR (significant increase in credit risk) event. Under the regime switching model, in the case of a transition to another Stage, we will need to estimate the PD values for the asset process governed by the new parameters. Using the straightforward notation $\Phi^i$ or $PD^i$ to emphasize the Stage (regime) under which the specific PD value is estimated, we can then define the "\emph{Stage-weighted provisions}", given by:
\begin{flalign} \label{weightedprov}
	WP_t := p_{11} EAD_t LGD_t PD^1 + p_{12} \sum_{i=t}^{T}\big(\Phi^2(x,i-1) -\Phi^2(x,i)\big)\mathbb{E} \Big[\frac{LGD_i  EAD_i}{(1+r_i)^{i-t}} | \mathcal{F}_t \Big]   + p_{13}EAD_tLGD_t,
\end{flalign}
where the third term occurs in the case of default (i.e. transition to Stege 3), we have that $PD=100\%$. This calculation holds for the case where we consider that the transition to Stage 2 occurs one period (e.g., year) after. However, we can also consider the cases where the deterioration occurs at any point $k>t$. For this calculation we require the $k$-step transition matrix of the underlying rating process, which is known to be $P^k$, whose elements will be symbolized as below:
$$
P^k=\left(  \begin{array}{c|ccc}
	\text{IFRS 9 Rating} & \text{Stage 1} & \text{Stage 2} & \text{Stage 3}\\
	\hline
	\text{Stage 1} & p^k_{11} & p^k_{12} & p^k_{12}\\
	\text{Stage 2} & p^k_{12} & p^k_{22} & p^k_{13}  \\
	\text{Stage 3} & 0 & 0 &1
\end{array}
\right).
$$
with the understanding that $p^k_{ij}$ represents the $k$-th step transition probability. We then have:
\begin{flalign} \label{weightedprovk}
	\mathbb{E}[WP_k|\mathcal{F}_t] = p^k_{11} EAD_k LGD_k PD^1_k + p^k_{12} \sum_{i=k}^{T}\big(\Phi^2(x,i-1)  &-\Phi^2(x,i)\big)\mathbb{E} \Big[\frac{LGD_i  EAD_i}{(1+r_i)^{i-t}} | \mathcal{F}_t \Big] \notag\\
	&+ p^k_{13}EAD_tLGD_t.
\end{flalign}
At any point, with the dynamics of the underlying Markov process, we can obtain the corresponding $WP_t$ values and calculate the above \emph{Expected Stage-weighted provisions}. This calculation takes the future evolution of the loan, as well as the regime into consideration to provide an estimation that incorporates all scenarios.
For illustrative purposes, we consider the example below.
\begin{example}\label{rs-example}
	Consider an asset process governed by the regime-switching model below: 
	\begin{eqnarray} \label{proc}
		dG_t =  
		\begin{cases} 
			k_1(\theta_1 -G_t)dt +\sigma_1 dB_t + \int_{\mathbb{R}} z N(dt,dz), \,\ G_0 = x,  \text{ if $R_t=$Stage 1, }   \\
			k_2(\theta_2 -G_t)dt +\sigma_2 dB_t + \int_{ \mathbb{R}} z N(dt,dz), \,\ G_0 = x,  \text{ if $R_t=$Stage 2, } \\
			k_3(\theta_3 -G_t)dt +\sigma_3 dB_t + \int_{ \mathbb{R}} z N(dt,dz), \,\ G_0 = x,  \text{ if $R_t=$Stage 3, } \\
		\end{cases}
	\end{eqnarray}
	% \begin{eqnarray} 
	% 	dG_t =  
	% 	\begin{cases} 
	% 		  0.3(0.7 -x)dt +0.5dB_t + \int_{z\in \mathbb{R}} z N(dt,dz) %\text{ if $S_t=$Stage 1,}   \\
	% 		  0.1(0.4 -x)dt +0.9dB_t + \int_{z\in \mathbb{R}} z N(dt,dz) %\text{ if $S_t=$Stage 2,} \\
	%	\end{cases}
	%\end{eqnarray}
	with regime-specific parameters given in Fig. \ref{graphs_rs} and $\mathcal{D} = [-6.0,6.0]$, with $S=4.0$ (we can also consider a different limit value $S$ for each regime. However, in this example the Monte Carlo estimates indicate that the same value suffices). We consider 
	normally distributed jumps with $Z \sim N(0.0,0.5)$, rate $\lambda = 1.0$. We have set $N=1001, T= 1001$ for the space and time grids, respectively. Furthermore, the generator matrix of the underlying Markov process given by:
	$$Q=\left(  \begin{array}{ccc}
		-0.5 & 0.3  & 0.2 \\
		0.3 & -0.6 & 0.3  \\
		0.0 & 0.0 &  0.0 \\
	\end{array}
	\right).$$
	The graphs in Fig. \ref{graphs_rs} display the estimated survival probability in each regime (Stage), resulting from solving scheme $(\ref{matrix-form})$. 
	
	\begin{figure}[h]
		\begin{center}
			\includegraphics[height=58mm, width=58mm,scale=0.3]{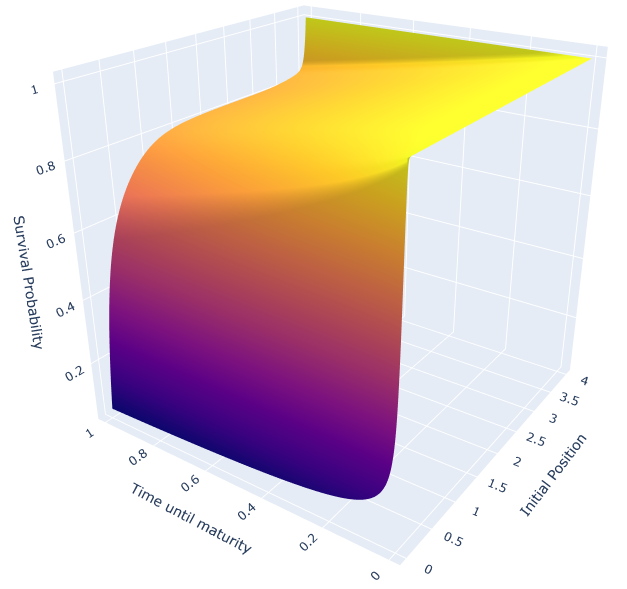}
			\includegraphics[height=58mm, width=57mm, scale=0.3]{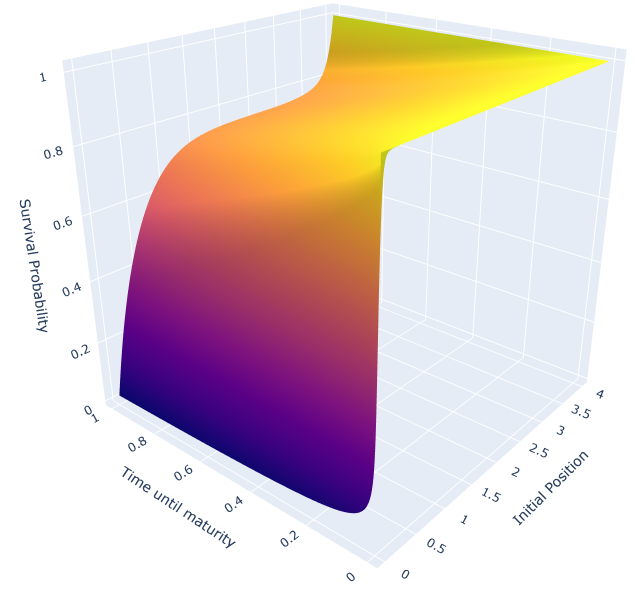}
			\includegraphics[height=58mm, width=58mm, scale=0.3]{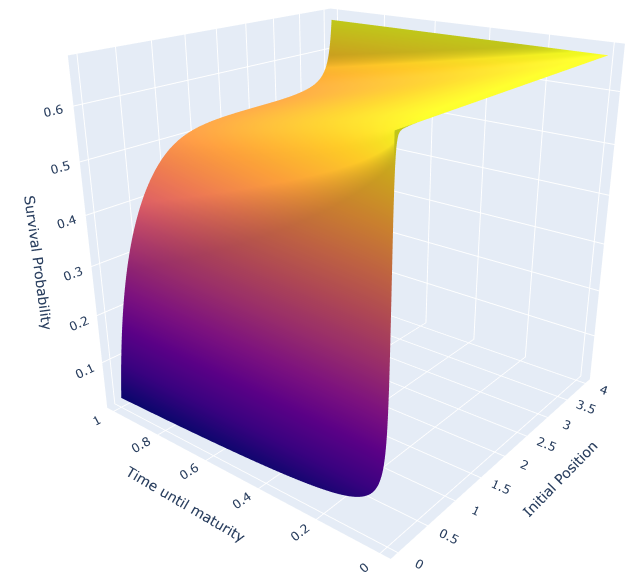}
		\end{center}
		\caption{Survival probability for (\ref{proc}) with $(k_1,k_2, k_3)=(0.3,0.2,0.0), (\theta_1,\theta_2, \theta_3)=(0.8,0.5,0.0), (\sigma_1, \sigma_2,\sigma_3)=(0.3,0.5,0.0)$, with $R_0 =$ Stage 1 (top left), $R_0 =$ Stage 2 (top right). The average survival probability across all regimes is also shown (bottom), which also account for the survival probability when $R_0 =$ Stage 3, for which we have $\Phi(x,u) \equiv 0$. }
		\label{graphs_rs}
	\end{figure}
	%In this example, we consider a loan exposure driven by the same asset process as in Example \ref{ex_rs}, which is currently in Stage 1. 
	We consider an initial position of $x=0.3$ and maturity $T=10$. The transition matrix of the underlying Markov process is obtained by calculating $P = \exp(Q)$:  
	\begin{eqnarray}\label{transmatrix}
		P=\left(  \begin{array}{c|ccc}
			\text{IFRS 9 Rating} & \text{Stage 1} & \text{Stage 2}  & \text{Stage 3}\\
			\hline
			\text{Stage 1} & 0.63 & 0.18 & 0.19\\
			\text{Stage 2} & 0.18 & 0.57 & 0.25 \\
			\text{Stage 3} & 0.00& 0.00 & 1.00 \\
		\end{array}
		\right).
	\end{eqnarray}
	We will perform the provisioning scenario analysis by forecasting the Stage-weighted provisions, given by ($\ref{weightedprovk}$) for the next four years. We first calculate the $k-$step transition matrices:
	\begin{eqnarray} \notag
		P^2=\left( \begin{array}{ccc}
			0.43 & 0.21 & 0.36 \\
			0.21 & 0.36 & 0.43 \\
			0.00 & 0.00 & 1.00\\
		\end{array}
		\right),
		\,\,\,\,
		P^3=\left( \begin{array}{ccc}
			0.31 & 0.20 &0.50 \\
			0.20 & 0.24 & 0.56\\
			0.00 & 0.00 & 1.00\\
		\end{array}
		\right),
		\,\,\,\,
		P^4=\left( \begin{array}{ccc}
			0.23 & 0.17 & 0.60 \\
			0.17 & 0.18 & 0.66 \\
			0.00 & 0.00 & 1.00\\
		\end{array}
		\right).
	\end{eqnarray}
	At time $t=0$ we consider the forward looking scenarios and can calculate the Stage 1 and Stage 2 provisions. Recall that Stage 1 provisions are given by (\ref{stage_1_losses}). Stage 2 (expected lifetime) provisions are given in column $ECL_t$ in Table \ref{weighted_table} below, which also contains the Point-in-Time Stage 1 and Stage 2 PD required to calculate the provisions.
	
	\begin{comment}
	\begin{center}
	\begin{tabular}{@{} *{8}{c} @{}}
	\headercell{Time} &$EAD_t$ & $\Psi^1(x,t)$ & $\Psi^2(x,t)$ & Stage 1 $PD_t^{PiT}$ &Stage 2 $PD_t^{PiT}$ & $LGD$ & $EL_t$\\
	%\cmidrule(l){2-7}
	\midrule
	$1.0$  & 100 & 0.1550 & 0.3576  &0.1550  &0.3576 &75\%  &  26.82  \\
	$2.0 $  & 90 &  0.3279 &  0.5562 & 0.1729 & 0.1986 &75\% & 13.41  \\
	$3.0$ & 80 &  0.4592 &  0.6577 & 0.1313 & 0.1015  & 75\% & 6.09 \\
	$4.0$& 70  &  0.5564&  0.7162 & 0.0972 & 0.0585 & 75\% &  3.07 \\
	$5.0$  & 60 &  0.6294  &  0.7516 & 0.0730 &0.0354 &75\%  &  1.59 \\
	$6.0 $  & 50 &   0.6851&  0.7738 & 0.0557 & 0.0222 &75\% & 0.83\\
	$7.0$ & 40 &  0.7279 &  0.7879 & 0.0428 & 0.0141 & 75\% & 0.42 \\
	$8.0$& 30  &  0.7610 &  0.7970& 0.0331 & 0.0091 & 75\% &  0.21 \\
	$9.0$ & 20 & 0.7867  &0.8031  &0.0257 &0.0061 & 75\% & 0.09 \\
	$10.0$& 10  & 0.8068 &  0.8071& 0.0201 &0.0040 & 75\% &  0.03 \\
	\end{tabular}
	\captionof{table}{Stage 1 and 2 PDs and provisions.}
	\label{weighted_table}
	\end{center}
	\end{comment}
	
	\begin{center}
		\begin{tabular}{@{} *{8}{c} @{}}
			\headercell{\thead{Time until \\ maturity ($u$)}} &$EAD_u$ & $\Psi^1(x,u) (\%)$ & $\Psi^2(x,u)(\%)$ & {\thead{Stage 1  \\ $PD_u^{PiT}(\%)$}} &{\thead{Stage 2  \\ $PD_u^{PiT}(\%)$}} & $LGD(\%)$ & $EL_u$\\
			%\cmidrule(l){2-7}
			\midrule
			$10.0$& 100  & 4.76 &  10.55 & 0.52 &1.53 & 75\% &  1.147 \\
			$9.0$ & 90 & 4.24  & 9.02  &0.50 & 1.50 & 75\% & 1.012 \\
			$8.0$& 80  &  3.74 &  7.52 & 0.48 & 1.44 & 75\% &  0.864 \\
			$7.0$ & 70 &  3.26 &  6.08 & 0.48 & 1.33 & 75\% & 0.698 \\
			$6.0 $  & 60 & 2.78 &  4.75 & 0.47 & 1.17 &75\% & 0.527 \\
			$5.0$  & 50 &  2.31  &  3.58 & 0.47 &0.99 &75\%  & 0.371 \\
			$4.0$& 40  &  1.84 &  2.59 & 0.46 & 0.80 & 75\% &  0.240 \\
			$3.0$ & 30 &  1.38 &  1.79 & 0.46 & 0.65  & 75\% & 0.146 \\
			$2.0 $  & 20 &  0.92 &  1.14 & 0.47 & 0.58 &75\% & 0.087  \\
			$1.0$  & 10 & 0.45 & 0.56  &0.45  & 0.56 &75\%  &  0.042  \\
		\end{tabular}
		\captionof{table}{Stage 1 and 2 PDs and expected losses.}
		\label{weighted_table}
	\end{center}
	
	As shown in Example \ref{stp}, the Lifetime Provisions can be calculated by the sum of the expected losses column, $EL_u$,. We can now calculate the Stage 1 and Stage 2 provisions at each subsequent time period, which we will use to calculate the Stage-weighted provisions for the next four years, calculated by (\ref{weightedprovk}). The results are shown in Table \ref{final-wp}, where the final column below contains the Stage-weighted provisions. 
	
	\begin{comment}
	\begin{center}
	\begin{tabular}{@{} *{8}{c} @{}}
	\headercell{Time}  & $EAD_t$ &  Stage 1 Provisions & Lifetime Provisions &Stage 3 Provisions &$WP_t$\\
	%\cmidrule(l){2-7}
	\midrule
	$1.0$  & 100 &11.63  & 52.56& 75.00 & 31.04 \\
	$2.0 $  & 90 &11.67 &25.74 &67.50 &34.72   \\
	$3.0$ & 80 &7.88 &12.33 &60.00 &34.91 \\
	$4.0$&70 &5.10 &6.25 & 52.50&33.74 \\
	\end{tabular}
	\captionof{table}{Stage-weighted provision calculations for the next 4 years.}
	\label{final-wp}
	\end{center}
	\end{comment}
	
	\begin{center}
		\begin{tabular}{@{} *{8}{c} @{}}
			\headercell{\thead{Time until \\ maturity ($u$)}}  & $EAD_u$ &  Stage 1 Provisions & $ECL$ &Stage 3 Provisions &$WP_u$\\
			%\cmidrule(l){2-7}
			\midrule
			$10.0$  & 100 & 0.390  & 5.13 & 75.00 & 15.42 \\
			$9.0 $  & 90 & 0.338 & 3.99 &67.50 & 25.28   \\
			$8.0$ & 80 & 0.288 &2.98 &60.00 &30.69 \\
			$7.0$&70 & 0.247 &2.11 & 52.50& 31.92 \\
		\end{tabular}
		\captionof{table}{Stage-weighted provision calculations for the next 4 years.}
		\label{final-wp}
	\end{center}
	The large difference observed even between the Lifetime and Stage-weighted provisions is evidence of the importance of such scenario analysis in provision calculations. Particularly in cases similar to that examined in this example, where the probability of transitioning to a default state is quite high the results can have an extremely large effect, which risk managers must account for in risk and provisioning policies. 
	\noindent
	\hfill $\triangleleft$
\end{example}

\subsection{Further Applications in credit risk modelling}\label{credit-sect}
\subsubsection{Pricing of Credit Default Swaps}
Another financial field in which the PD function plays a paramount role is credit derivatives pricing. In particular, we consider the fair price of Credit Default Swap (CDS). A default swap is a contract that protects the holder of an underlying swap from the losses caused by the default to the obligation’s issuer. Therefore, the evolution of the PD values can be used for the pricing, hedging and managing of such options. Extensive work has been done on modeling and pricing CDSs, such as in \cite{cariboni2007pricing} and \cite{houweling2005pricing}. Specifically, it can be shown that the price of the CDS is given by:
\begin{eqnarray} \notag
	C D S=(1-R)\left(-\int_{0}^{T} e^{-r s} d\Phi(x,s)\right)-c \int_{0}^{T} e^{-r s} \Phi(x,s) ds,
\end{eqnarray}
and the corresponding par spread:
\begin{eqnarray}\notag
	c^*=\frac{(1-R)\Big(-\int_{0}^{T} e^{-r s} d\Phi(x,s) \Big)} {\int_{0}^{T} e^{-r s} \Phi(x,s) ds}, 
\end{eqnarray}
where $R$ is the specific recovery rate and $r$ is the risk-free rate. The above expression can be discetized as follows:
\begin{eqnarray}\label{par}
	c^*=\frac{(1-R) \sum_{i=1}^{n} e^{-rt_i} (\Phi(x,t_{i-1})- \Phi(x,t_i))} {\frac{1}{2} \sum_{i=1}^{n}e^{-rt_i} (\Phi(x,t_{i-1}) + \Phi(x,t_i)) \Delta t_i}, 
\end{eqnarray}
where the Trapezoidal rule has been used for the discretization of the denominator. Estimating the price and par rate of CDS therefore requires the term structure of the underlying risk-free and survival probability processes. We present a simplified example, whereby the interest rate is again considered to be zero. 
\begin{example}\label{ex-sv}
	Consider a CDS with maturity $T=10$ years and recovery rate $R=0.5$, where the asset process evolves according to the following stochastic volatility model:
	\begin{eqnarray} \notag
		\begin{cases} 
			dG_t = k(\theta -G_t)dt +\sigma(Y_t)dB_t + \int_{\mathbb{R}} z N(dt,dz), \,\,\,\ G_0 = x\\
			dY_t = \kappa(\mu - Y_t)dt  + \xi \sqrt{Y_t}dW_t, \,\,\,\ Y_0=y. \\
		\end{cases}
	\end{eqnarray}
	The parameters of the stochastic model are given in Fig. \ref{graphs_sv}, we let $\mathcal{D}=[-5.0,75.0]$ and consider a spatial and temporal discretization with 200 and 1000 steps respectively. Jumps are again normally distributed, with size $Z \sim N(0.3,0.5)$, rate $\lambda = 1.0$ and $J=90$. Furthermore, we set a grid with 200 steps for the volatility $\mathcal{V} = [0.0,200.0]$. The parameters and resulting graph of the PD function can be seen in Figure \ref{graphs_sv}.
	\begin{figure}[h]
		\begin{center}
			\includegraphics[height=62mm, width=80mm, scale=0.5]{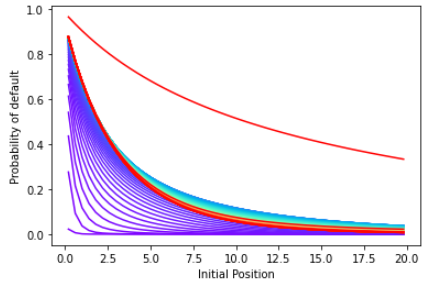}
			\includegraphics[height=62mm, width=80mm, scale=0.5]{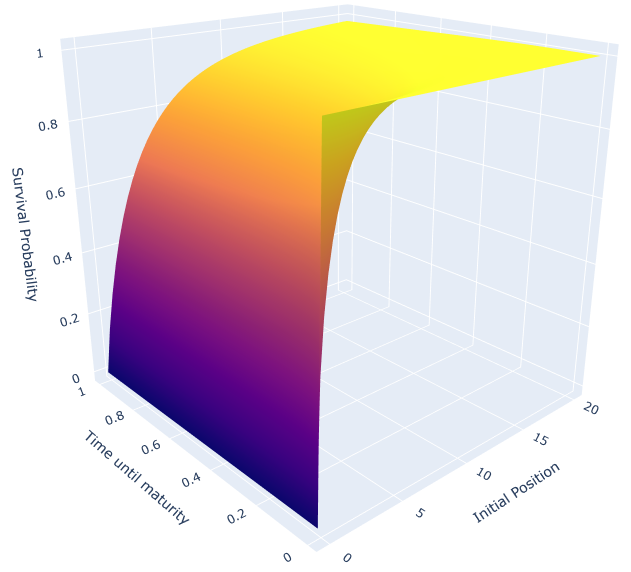}
		\end{center}
		%	ksi_choice = 1.0
		%	omega_choice = 0.5
		%		theta_vol_choice = 0.3
		%		kappa_choice = 0.5
		%		theta_choice = 0.3
		%		jump_rate_choice = 1.0
		%		N_choice = 4
		%		mu_choice = 0.0
		%		sigma_jump_choice = 0.5
		
		\caption{(Left) Evolution of the PD under the stochastic volatility model with $(k,\theta, \kappa,\mu,\xi)=(2.0,2.0,0.05,0.1,0.07)$, for various values of the starting volatility $Y_0$. (Right) The average survival probability across all volatility values.}
		\label{graphs_sv}
	\end{figure}
	
	We assume an initial position of $x= 3.0$ and plot the evolution of the average survival probability in Fig. \ref{graphs_phi_cd}. 
	\begin{comment}
	\begin{center}
	\begin{tabular}{@{} *{2}{c} @{}}
	\headercell{Time} &$\Phi(0.8,t)$ \\
	%\cmidrule(l){2-7}
	\midrule
	$0$  & 100.0\% \\
	$1 $  & 91.93\%   \\
	$2$  & 84.41\% \\
	$3 $  & 77.52\%   \\
	$4$ & 71.14\%  \\
	$5$  & 65.22\% \\
	$6 $  & 59.72\%   \\
	$7$ & 54.63\%   \\
	$8$ &  49.93\%   \\
	$9$  & 45.62\% \\
	$10$ & 41.68\%   \\
	\end{tabular}
	\end{center}
	\noindent
	\end{comment}
	\begin{figure}[h]
		\begin{center}
			\includegraphics[height=50mm, width=70mm,scale=0.3]{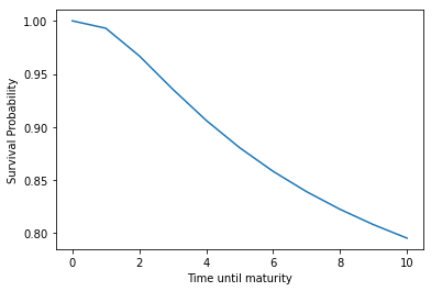}
		\end{center}
		\caption{Evolution of the average survival probability under the stochastic volatility model described in Example \ref{ex-sv}, with starting point $x=0.3$.}
		\label{graphs_phi_cd}
	\end{figure}
	The resulting par spread is calculated using (\ref{par}) to obtain $c^* = 0.33$.
	
	\hfill $\triangleleft$
\end{example}

\subsubsection{Credit Portfolio Optimization}
For many financial institutions, one of the most important tasks is securitization of credit exposures. Ultimately, this can be formulated as an optimization problem. The PD function affects the risk of each exposure and, by extension, the corresponding return as well. To this end, we present a simple example to show how such an optimization problem can be solved under the stochastic volatility PD model. 
\begin{example}
	Suppose a securitization agency creates a portfolio consisting of loans (or credit derivatives), each with different underlying asset process. The resulting PD functions will differ depending on the loan's (or derivative's) characteristics and asset value processes. Slightly abusing notation, we suppose that, for a portfolio of three loans, the corresponding PD functions are given by $PD_i, i=1,2,3$, estimated using the methodology developed above. 
	\par The agency aims to select the investment allocated to each of the credit exposures.
	Specifically, it poses the following portfolio optimization problem: suppose $w_i, i=1,2,3$ and $r_i, i=1,2,3$ represent the weight of total investment allocated to each institution's set of loans and  their average return, respectively. Consider, furthermore, that the required portfolio rate of return is set to be $R^*$. For the credit exposure $i$, at time $t$, the expected loss is given by $EL^i_t = EAD^i_tLGD^i_tPD^i_t$, and we can then define the total loss function for the agency by $L(t)= \sum_{1}^{3} w_iEL^i_t$. In order to rebalance the portfolio at each period, the securitization agency is then interested in solving a portfolio optimization problem (we present a very simple such problem, which can be solved analytically to illustrate the use of the method). The optimization we consider is the following: \\
	\begin{eqnarray*}
		\min_{\substack{\bf w}} \mathbb{E}[U(L(t))],  \text{ subject to }\\
		w_1 r_1 + w_2 r_2 + w_3 r_3 = R, \\
		w_1 +w_2 +w_3=1,
	\end{eqnarray*} 
	for an appropriate loss function $U$. While the following analysis can be extended to any convex loss function, for the sake of simplicity, we illustrate the calculation selecting the quadratic loss function $U(L)=bL^2-L$ (in the sense of a negative utility function). To standardize the optimization problem, we consider that $EAD_i$ is given as a percentage of the original loan value and for simplicity we consider a constant $LGD^i_t = 1$ for $i=1,2,3$ and all $t$. At any point in time $t$, the expected loss utility is then:
	$$ \sum_{i=1}^{3} U(w_iEAD_i)PD^{PiT}_i,$$ where $PD^{PiT}_i$ is the current Point-in-Time default probability corresponding to exposure $i$. The agency must optimize the portfolio by solving: 
	\begin{eqnarray*}
		\text{minimize } f(w_1,w_2,w_3):= \sum_{i=1}^{3} b(w_iEAD_i)^2-w_iEAD_i, \text{ subject to} \\
		w_1 r_1 + w_2 r_2 + w_3 r_3 = R, \\
		w_1 +w_2 +w_3=1,
	\end{eqnarray*} 
	This simple quadratic optimization problem can now be solved either analytically or numerically. It is straightforward to calculate:
	$$w_3^* =\frac{PD^{PiT}_1(2bEAD_1^2\delta\epsilon - \epsilon)+PD^{PiT}_2(\gamma-2baEAD_2^2\beta\gamma)-PD^{PiT}_3}{2(bPD^{PiT}_1EAD_1^2\epsilon^2+bPD^{PiT}_2EAD_2^2\gamma^2+bPD^{PiT}_3EAD_3^2)}, $$
	where $\beta=\frac{R-r_1}{r_2-r_1}, \gamma=\frac{r_3-r_1}{r_2-r_1}, \delta = \frac{r_2-R}{r_2-r_1}$ and $\epsilon=\frac{r_3-r_2}{r_2-r_1}$. A straightforward substitution using the two conditions will result in the corresponding values $w_1^*$ and $w_2^*$. We consider the above setting with average returns from each institution's instruments $r=(r_1,r_2,r_3)^T$ and current exposures $EAD=(EAD_1,EAD_2,EAD_3)^T$ given by:
	\begin{eqnarray}\notag
		r=(0.1 \,\ 0.3 \,\ 0.1)^T,
		\,\,\,\,
		EAD=(0.9 \,\ 0.8 \,\ 0.7)^T.
	\end{eqnarray}
	In order to obtain the vector containing the PD values, we consider the three asset classes described by the processes below:
	
	\begin{eqnarray}
		dG^1_t=k_1(\theta_1 - G^1_t)dt + \sigma_1 dB_t +  \int_{\mathbb{R}} z N(dt,dz), \,\,\ x_0 = 1.00 \notag\\
		dG^2_t=k_2(\theta_2 - G^2_t)dt + \sigma_2 dB_t +  \int_{\mathbb{R}} z N(dt,dz), \,\,\ x_0 = 0.20 \notag\\
		dG^3_t=k_3(\theta_3 - G^3_t)dt + \sigma_3 dB_t +  \int_{\mathbb{R}} z N(dt,dz),\,\,\ x_0 = 0.50, \notag
	\end{eqnarray}
	with $(k_1, k_2, k_3) = (0.5,0.8,0.5), (\theta_1, \theta_2, \theta_3)= (3.5, 3.0, 2.5), (\sigma_1, \sigma_2, \sigma_3) = (2.0, 1.5, 2.5)$, $J = 150, \lambda =1.0$ and jump distributions $Z \sim N(0.0,0.2)$ for all three.
	Solving PIDE $(\ref{pide2})$, we obtain the $PD^{PiT}$ values:
	$$PD^{PiT}=(0.0783 \,\ 0.1447 \,\ 0.0447)^T,$$
	and fixing the expected total return to be $R=25.00\%$, the resulting optimal weights $w^*=(w_1^* \,\ w_2^* \,\ w_3^*)^T$ are:
	$$w^*=(0.163 \,\ 0.750 \,\ 0.087)^T.$$
	%\hfill$\triangleleft$
\end{example}  \hfill $\triangleleft$\\
\par For extensive work on portfolio optimization problems with defaultable assets, we refer the interested reader to e.g., \cite{asanga2014portfolio}. Furthermore, empirical studies of the applicability of standard ruin probabilities in practice can be found in \cite{braun2015solvency}. In the example above, we focus on a case where an agency must assess and optimize a portfolio of loan exposures with varying characteristics. Such cases could be loans originating in different sectors; e.g., in \cite{pasricha2020portfolio}, the authors consider a portfolio of risky bonds originating from the Industry and Service sector.

	\section{Conclusion}
	In this paper we have focused on a generalized approach of estimating PD values, accounting for both the cases of variable starting times and maturities. We show that under certain conditions imposed on the models representing the asset processes, these two cases can be dealt with equivalently and lead to important novel representations of the PD function. Specifically, with the integral equation approach, we can construct a robust mathematical framework that allows us to develop both theoretical and numerical tools for the calculation of the PD values. This methodology has important advantages over standard Monte Carlo methods, as well as over existing approaches using PIDEs, as we are able to consider sophisticated models that incorporate multiple latent variables, without sacrificing mathematical rigor for required regularity assumptions.
	\par In terms of practicality and applications, the proposed framework covers many of the difficulties financial institutions face due to the new regulatory requirements for provision calculations, as well as continuous credit risk monitoring for SICR events. We hypothesize that this approach could be useful for practitioners, given that it constitutes a complete and efficient modelling framework, with which one can calculate Point-in-Time and Lifetime PD values, each of which are used extensively in credit risk management. Specifically, this framework is motivated by the needs created by the IFRS 9 regulations, under which forecasting credit losses accurately and efficiently is of paramount importance. We show how the PD estimations can be used to calculate Stage 2 provisions, as well as more advanced, scenario-based provisions, and of extensive further applications in credit risk modelling. 
	\par Finally, we note that this approach most likely is best fit for corporate and small business loans, where the estimation of asset processes has been documented in well-established work. Of course, it is possible that with new developments in payment services and Open-Banking solutions (in accordance to the Payment Services Directive 2), such methods could be applied to individual consumers, given sufficient historical data. An example of recent work done in this direction is \cite{tobback2019retail}. To conclude, it is important to mention that the LGD parameter is also of great importance for provision calculations; in this work we considered a constant LGD, however, in practice, LGD values require separate model development, often related to current macroeconomic variables, as shown in \cite{bellotti2012loss}. Future research could focus on considering appropriate models for the evolution of the LGD, in combination with the PD function.

	\appendix
	\section{Stochastic processes}
	
	\subsection{The continuous Ornstein Uhlenbeck process}\label{ou-app}
	In its simplest form, the OU process $X_t$ is defined as the stochastic process satisfying the SDE: 
	\begin{eqnarray}\label{ou}
		dX_t=k(\theta - X_t)dt + \sigma dB_t, \,\, X_s=x, 
	\end{eqnarray}
	for some known $x$, where, as above, $B_t$ represents the standard Brownian motion and $k, \theta$ and $\sigma$ are positive real constants. The OU process is a mean-reverting, Gaussian and Markov process, which is also temporally homogeneous. We can therefore equivalently write (\ref{ou}) as:
	\begin{eqnarray}\label{ou-cont}
		dX_u=k(\theta -X_u)du+ \sigma dB_u, \,\, X_0=x,
	\end{eqnarray}
	where $u = t-s$. For simplicity, we write $X_t^x$ to indicate the OU process with $X_0=x$. We adopt this convention for all stochastic processes in the remainder of this work. Employing It\^o's formula we can obtain the solution to the above SDE:
	$$ X_t = x e ^{-kt}+ \theta(1-e^{-kt}) + \sigma \int_{0}^{t} e^ {-k(t-u)}dB_u,$$
	from which is it easy to see that $X_t \sim N(\theta + (x - \theta) e^{-kt}, \sigma^2(1-e^{-2kt})/2k).$
	These properties are what make this particular family of processes widely used in many applications. We will also need the following regarding the transition density and hitting time for the OU process.
	
	\begin{theorem}\label{td}
		The transition density of the OU process, with initial condition $X_0=x$ is given by: 
		\begin{eqnarray}
			p(y,x,t) \equiv \mathbb{P}\big(X_t = y | X_0 =x\big)=  \sqrt{\frac{k}{\pi \sigma^2 (1-e^{-2kt})}} \exp\Big(-\frac{k(y-\theta - (x-\theta) e^{-kt})^2}{\sigma^2 (1-e^{-2kt})}\Big).
		\end{eqnarray}
	\end{theorem}
	
	Furthermore, for the OU process as defined in (\ref{ou}), we define the corresponding survival probability distribution, given by $Q(x, t) :=\mathbb{P}\Big(\inf_{\substack{r \leq t}} G_r^{x} > 0 \Big)$. The distribution can be obtained via appropriate Volterra equations, for which we refer the reader to \cite{lipton2018first}. We will not employ this representation explicitly, however we will use the fact that both $p(y,x,t)$ and $Q(x,t)$ are continuous in $(x,t)$. More details can be found in Appendix \ref{gens-cou}.

	\subsection{Continuous Time Markov Chain}\label{Markov}
	In this section we outline, for the reader's convenience, the background and important results pertaining to Continuous Time Markov Chains (CTMC), which are used for the regime-switching models.
	
	\begin{definition}
		A continuous time Markov chain is a continuous stochastic process $X_t, t\geq 0$, with a discrete state space $\mathcal{R}$, of cardinality $|\mathcal{R}| < \infty$, satisfying the Markov property, and such that:
	\begin{eqnarray}\label{markov-gen-matrix}
		\mathbb{P}(X_{t+\delta}=j | X_t=i) = \begin{cases}
			q_{ij}\delta + o(\delta), \,\ i \neq j\\
			1 + q_{ii} \delta + o(\delta),
		\end{cases}
	\end{eqnarray}
	as $\delta \downarrow 0$. In the above $q_{ij}$ are known as the transition rates, for which we have $\sum_{i \in S} q_{ij}=0$, $q_{ij} \geq 0$ for $i\neq j$.  
	\end{definition}
	The matrix $Q$ with entries $(Q)_{ij}=q_{ij}$, for $i,j\in \mathcal{R}$ is known as the generator matrix of the Markov process (also referred to as the transition rate matrix). Similar to the discrete time Markov chains, we can define the transition matrix for a CTMS, $P(t), t\geq 0$, with entries: 
	\begin{eqnarray}
		p_{ij} = 	\mathbb{P}(X_{t}=j | X_0=i),
	\end{eqnarray}
 	for $i,j\in \mathcal{R}$. The following result holds for the transition matrix, from which we are also able to obtain a connection between the transition and generator matrices.
 	\begin{theorem}
 		The transition matrix $P(t)$ satisfies the Kolmogorov forward equation:
 		\begin{eqnarray}\notag
 			P'(t) = P(t)Q, 
 		\end{eqnarray}
 	and hence:
 	\begin{eqnarray}
 		P(t) = e^{tQ}.
 	\end{eqnarray}
 	\end{theorem}
 	\par Finally, we remind the reader that we say that a state $i$ is transient if, given that the chain starts at $i$, it is possible, but not certain, that the chain will return to $i$. Equivalently, there exists a non-zero probability that the chain will never return to $i$. On the other hand, a state $i$ is defined as absorbing if $\mathbb{P}(X_t=i|X_0=i)=1$, for all $t\geq 0$, i.e. the probability of transitioning from $i$ to any other state is zero. 
\begin{comment}	
	\begin{definition}
		For the discrete space Markov process $M=(\mathcal{R},P)$, we can define the generator matrix $Q =(Q_{ij})_{i,j=1,\dots n}$ by: 
		\begin{eqnarray}
		Q = \lim_{t \downarrow 0 } \frac{ P_t - I_{R }}{t},
		\end{eqnarray}
		where $I_{R }$ is the $R \times R$ identity matrix.
		Alternatively, we can write as in \cite{zhu2015feynman}:
		\begin{eqnarray}
		\mathbb{P}(X(t+\delta)=j | X(t)=i) = \begin{cases}
		q_{ij}\delta + o(\delta), \,\ i \neq j\\
		1 + q_{ii} \delta + o(\delta),
		\end{cases}
		\end{eqnarray}
		as $\delta \downarrow 0$, and where $\sum_{i \in S} q_{ij}=0$, with $q_{ij} \geq 0$ for $i\neq j$. 
	\end{definition}
	Finally, the relationship between the transition matrix $P$ and generator $Q$ of a Markov chain is given by: 
	\begin{eqnarray}
	P = \exp{(Q)} = I + Q + \frac{Q^2}{2!} + \dots .
	\end{eqnarray}
	The above can also be generalized to the $k-$step transition matrix, i.e. $P(k) = \exp{(kQ)}$. 
\end{comment}
	
	\subsection{L\'evy processes} \label{levy}
	Throughout this paper, we have abopted the notation used in \cite{oksendal2007applied}. We start by defining a L\'evy process:
	
	\begin{definition}[L\'evy process]
		A L\'evy process $\{L_t\}_{t \geq 0} $ is a stochastic process for which the following conditions hold:
		\begin{itemize}
			\item $L_0=0$.
			\item $L$ has independent and stationary increments, i.e., if $ t >s$ then $L_t-L_s$ is independent from $ L_s$ and $L_t-L_s \stackrel{D}{=} L_{t-s}$.
			\item  $L$ is stochastically continuous, i.e for all $\epsilon >0 $ and all $ s >0$ we have $$\lim_{t \rightarrow s} \mathbb{P}(|X(t)-X(s)|>\epsilon)=0.$$
		\end{itemize}
	\end{definition}
	
	A consequence of the above definition is the celebrated It\^o - L\'evy decomposition. First, we define the following required quantities:
	
	\begin{definition}
		Let $L_t$ be a L\'evy process, whose jump is defined as $\Delta L_t = L_t - L_{t_-}$. Furthermore, let $\textbf{B}_0$ be the family of Borel sets $U \subset \mathbb{R} $, whose closure does not contain 0. Then, for $U \in \textbf{B}_0$, define the Poisson random measure of the L\'evy process $L_t$ by: 
		$$ N(t,U) = \sum_{0 < s \leq t } \mathds{1}(\Delta L_s).$$ 
	\end{definition}
	The Poisson random measure represents the number of jumps of size $\Delta L_s \in U$, which occur up to time $t$. We can therefore define the intensity of the jumps as follows:
	\begin{definition}
		The intensity of a L\'evy jump process $L_t$, known as the L\'evy measure of $L_t$ is defined as:
		$$\nu(U) = \mathbb{E}[N(1,U)],$$
		where, as above, $U \in \textbf{B}_0$.
	\end{definition}
	A useful consequence of the above definitions is that if $\nu$ is the L\'evy measure of a simple Compound Poisson Process with rate $\lambda$ and jump size density $f(z)$, then we have that $$ \nu(U) = \lambda f(U).$$ To this end, we will employ the following result in the subsequent sections, due to \cite{kyprianou2006introductory}:
	
	\begin{theorem}\label{KYPRIANOU-APP}
		Consider the Poisson random measure  $N(t,U)$, with $U \in \textbf{B}_0$, and corresponding L\'evy measure $\nu(U)$. Then the process:
		$$ X_t = \int_{0}^t \int_{B} z N(ds,dz), $$
		where $B \in \mathcal{B}(\mathbb{R})$, is a Compound Poisson Process with rate $\nu(B)$ and jump distribution $\frac{\nu(dx)|_B}{\nu(B)}$.
	\end{theorem}
	In this paper, we will focus on jump terms in the form above. We can now present the following theorem:
	\begin{theorem}[It\^o - L\'evy decomposition]
		Let $\{L_t\}_{t \geq 0}$ be a L\'evy process. Then, we have 
		\begin{eqnarray}
		L_t = bt + \sigma B_t + \int_{|z|<1} z \tilde{N}(t,dz) + \int_{|z| \geq1} z N(t,dz), 
		\end{eqnarray}
		for $ t \geq 0 $, where $b, \sigma \in \mathbb{R}$, $B_t$ is a Brownian motion and $\tilde{N}(t,dz):=N(t,dz)- \nu(dz)t$ is the compensated Poisson measure. 	
	\end{theorem}
	More generally, we can define the stochastic process $ X(t)$, as: 
	\begin{eqnarray}
	dX_t=a(t)dt + \sigma(t)dB(t) + \int_{|z|<1} H(t,z) \tilde{N}(dt,dz) + \int_{|z| \geq1} H(t,z) N(dt,dz),
	\end{eqnarray}
	known as L\'evy - It\^o processes. Moreover, by combining the compensator with the drift term the above can be written as: 
	\begin{eqnarray}\label{levy-ito}
	dX_t=a(t)dt + \sigma(t)dB(t) + \int_{z \in \mathbb{R}} H(t,z) N(dt,dz).
	\end{eqnarray}
	We will adopt this formulation throughout the remainder of this work. For such processes, we have the following results, which are extension of the standard It\^o and Generator formulas. 
	
	\begin{theorem}[It\^o formula]
		Let $X_t \in \mathbb{R}$ be an It\^o- L\'evy process and consider a function $f(x,t)$, with $f \in C^{2}(\mathbb{R} \times [0,T])$. Then, the dynamics of the process $f(X_t,t)$ are given by the following version of the It\^o formula:
		\begin{align*}
		& df(X_t,t) = \frac{\partial f}{\partial t}(X_t,t)dt + \frac{\partial f}{\partial x} (X_t,t) \big(a(t) dt  +  \sigma(t) dB_t \big) + \frac{1}{2} \frac{\partial^2 f}{\partial x^2} (X_t,t) \sigma^2(t) dt \\
		& + \int_{\mathbb{R}}\big( f\big(X_{t-}+ H(t,z),t\big) - f(X_{t-},t)\big) N(dt, dz)
		\end{align*}
		\end {theorem}

		\begin{definition}[Generator]
			For a L\'evy-It\^o process, given by (\ref{levy-ito}), and function $f:\mathbb{R} \times [0,T] \rightarrow \mathbb{R}$  we define the generator $\mathcal{L}$ by:
			\begin{eqnarray}
			\mathcal{L}f(x,t) = \lim_{t \downarrow 0 }\frac{\mathbb{E}^x[f(X_t,t)]-f(x,t)}{t},
			\end{eqnarray}
			where $\mathbb{E}^x[f(X_t,t)] = \mathbb{E}[f(X_t,t)|X_0=x].$	
		\end{definition}
		It particular, it can be shown that the generator admits the following form: 
		\begin{eqnarray}
		\mathcal{L}f(x,t) = \frac{\partial f}{\partial t} +a(t)\frac{\partial f}{\partial x} + \frac{1}{2} \sigma(t)^{2} \frac{\partial^{2} f}{\partial x^{2}} +\int_{\mathbb{R}}\big(f(x+z,t)-f(x,t)\big) \nu(dz), 
		\end{eqnarray}

		\section{Infinitesimal generators and PDEs for the continuous Ornstein - Uhlenbeck models}\label{gens-cou}
		
		In this Appendix we detail the generators for the continuous counterparts of the processes considered in this paper, built upon the continuous OU process given by (\ref{ou-cont}) in Appendix \ref{ou-app}. Firstly, we recall that the survival probability $Q(x,t):= \mathbb{P}\Big(\inf_{\substack{r \leq t}} X_r^{x} > 0 \Big)$, as well as the corresponding transition density $p(\cdot,x,t)$ satisfy the equation: 
		
		\begin{eqnarray}\label{kolm}
			\frac{\partial f}{\partial t}(x,t) = \mathcal{L}f(x,t),
		\end{eqnarray}
		where $\mathcal{L}$ represents the operator:
		\begin{eqnarray}\label{ou-continuous}
			\mathcal{L}f(x,t) = k(\theta-x) \frac{\partial f}{\partial t} (x,t) + \frac{1}{2}\sigma^2\frac{\partial f^2}{\partial x^2} (x,t).
		\end{eqnarray}
		This is known as the Kolmogorov backward equation. 
		In general, when considering the survival probabilities and transition densities under the continuous versions of the regime switching, stochastic volatility and generalized models, analogous equations to (\ref{kolm}) are produced. These depend on the generators of the processes, which now include terms to capture the evolution of the regime and/or volatility processes. For more details on the generators of regime switching and stochastic volatility models see e.g., \cite{hainaut2011financial,zhu2015feynman}. Below we display the Kolmogorov backward equations under the continuous regime switching, stochastic volatility and generalized models, respectively:
		\begin{flalign}
			\frac{\partial f}{\partial t}(x,\rho, t)&=\mathcal{L}_1Q(x,\rho, t) &&\notag\\&:=k_\rho(\theta_{\rho}-x) \frac{\partial f}{\partial x}(x,\rho,t) + \frac{1}{2} \sigma_\rho^2 \frac{\partial^{2} f}{\partial x^{2}}(x,\rho,t)+ \sum_{j \neq \rho} q_{\rho j} \Big(Q(x,j,t) - f(x,\rho,t)\Big) \label{l_1}&&\\
			\frac{\partial f}{\partial t}(x,y,t)&=\mathcal{L}_2 f(x,y, t)&&\notag\\&:=k(\theta - x) \frac{\partial f}{\partial x}(x,y,t) + \kappa (\mu -y)\frac{\partial f}{\partial y}(x,y,t) + \frac{1}{2} y \frac{\partial^{2} f}{\partial x^{2}}(x,y,t) + \frac{1}{2} \xi^2 y \frac{\partial^{2} Q}{\partial y^{2}}(x,y,t) \label{l_2}&&\\
			\frac{\partial f}{\partial t}(x,\rho, y,t)&=\mathcal{L}_3 f(x,\rho,y, u):=k_\rho(\theta_{\rho}-x) \frac{\partial f}{\partial x}(x,\rho,y,t) + \kappa (\mu -y)\frac{\partial f}{\partial y}(x,\rho, y,t)  &&\notag \\ &+\frac{1}{2} \sigma_\rho^2 y\frac{\partial^{2} f}{\partial x^{2}}(x,\rho,y,t)+\frac{1}{2} \xi^2 y \frac{\partial^{2} f}{\partial y^{2}}(x,\rho, y,t) +\sum_{j \neq \rho} q_{\rho j} \Big(f(x,j,y,t) - f(x,\rho,y,t)\Big)  \label{l_3},&&
		\end{flalign}
		where we define separately the operators $\mathcal{L}_1, \mathcal{L}_2$ and $\mathcal{L}_3$ for notational convenience. Note that we write the dependency on the regime as a subscript on the right hand side of the equations above.
		
		\subsection{Regularity of solutions to parabolic PDEs}\label{reg-par-pde}
		We will require the following results pertaining to the regularity of solutions of the second order parabolic PDE (\ref{kolm}). The following are due to \cite{garroni1992green}. We first define some relevant function spaces that will be required for the subsequent regularity results.
		
		\begin{definition}
			Consider $\Omega \subset \mathbb{R}^n$ an open set, with closure $\bar{\Omega}$. Furthermore, consider a fixed time horizon $T >0$ and define $Q_T = \Omega \times [0,T]$, with closure $\bar{Q}_T$. We then define the following spaces, for $0<\alpha<1$: 
			\begin{itemize}
				\item $C^0(\bar{\Omega})$ is the Banach space of bounded continuous functions in $\bar{\Omega}$, with the natural supremum norm:
				$$
				\|\cdot\|_{C^0(\bar{\Omega})} \equiv\|\cdot\|_{0, \bar{\Omega}}=\sup _{\Omega}|\cdot|
				$$
				\item $C^{2,1}\left(\bar{Q}_T\right)$ is the Banach space of functions $\varphi(x, t)$ belonging to $C^0\left(\bar{Q}_T\right)$ together their derivatives $\frac{\partial f}{\partial x}, \frac{\partial^2 f}{\partial x^2}, \frac{\partial f}{\partial t}$ in $\bar{Q}_T$ with natural norm.
				
				\item $C^{\alpha, \frac{\alpha}{2}}\left(\bar{Q}_T\right)$ is the Banach space of function $\varphi$ in $C^0\left(\bar{Q}_T\right)$ which are Hölder continuous in $\bar{Q}_T$ with exponent $\alpha$ in $x$ and $\frac{\alpha}{2}$ in $t$ i.e. having a finite value for the seminorm
				$$
				\langle f \rangle_{\bar{Q}_T}^{(\alpha)} \equiv\langle f \rangle_{x, \bar{Q}_T}^{(\alpha)}+\langle f \rangle_{t, \bar{Q}_T}^{\left(\frac{\alpha}{2}\right)}
				$$
				where
				$$
				\begin{aligned}
					&\langle f \rangle_{x, \bar{Q}_T}^{(\alpha)}=\inf \left\{C \geq 0:\left|f(x, t)-f\left(x^{\prime}, t\right)\right| \leq C\left|x-x^{\prime}\right|^\alpha, \forall x, x^{\prime}, t\right\} \\
					&\langle f \rangle_{t, \bar{Q}_T}^{\left(\frac{\alpha}{2}\right)}=\inf \left\{C \geq 0:\left|f(x, t)-f \left(x, t^{\prime}\right)\right| \leq C\left|t-t^{\prime}\right|^{\frac{\alpha}{2}}, \forall x, t, t^{\prime}\right\}
				\end{aligned}
				$$
				The quantity
				$$
				\|f\|_{C^{\alpha, \frac{\alpha}{2}}\left(\bar{Q}_T\right)} \equiv\|f\|_{\alpha, \bar{Q}_T}=\|f\|_{0, \bar{Q}_T}+\langle f\rangle_{\bar{Q}_T}^{(\alpha)}
				$$
				defines a norm.
				\item $C^{2+\alpha, \frac{2+\alpha}{2}}\left(\bar{Q}_T\right)$ is the Banach space of functions $f(x, t)$ in $C^{2,1}\left(\bar{Q}_T\right)$ having a finite value for the seminorm:
				$$
				\langle f\rangle_{\bar{Q}_T}^{(2+\alpha)}=\left\langle\partial_t f\right\rangle_{\bar{Q}_T}^{(\alpha)}+\sum_{i, j=1}^d\left\langle\partial_{i j} f\right\rangle_{\bar{Q}_T}^{(\alpha)}+\sum_{i=1}^d\left\langle\partial_i f\right\rangle_{t, \bar{Q}_T}^{\frac{1+\alpha}{2}} .
				$$
				Then, the quantity
				$$
				\|f\|_{C^{2+\alpha, \frac{2+\alpha}{2}}\left(\bar{Q}_T\right)} \equiv\|f\|_{2+\alpha, \bar{Q}_T}=\sum_{2 r+s \leq 2}\left\|\partial_t^r \partial_x^s f\right\|_{0, \bar{Q}_T}+\langle f \rangle_{\bar{Q}_T}^{(2+\alpha)}
				$$
				defines a norm.
			\end{itemize}
		\end{definition}
		
		\begin{theorem}
		Consider a bounded domain $\Omega$, the operator $L:= \frac{\partial f}{\partial t}(x,t) - \mathcal{L}f(x,t)$ and the PDE:
		\begin{eqnarray}\label{pde-reg}
			\begin{cases}
				Lf = g(x,t) & \text{ for } (x,t) \in Q_T \\ 
				f(x,0)= \varphi(x) & \text { for } x \in \Omega \\
				 f(x,t) = \psi(x,t) & \text{ for } x \in \Sigma_T:=\partial\Omega \times [0,T].
			\end{cases}
		\end{eqnarray}
		Then, for any $g \in C^{\alpha, \frac{\alpha}{2}}\left(\bar{Q}_T\right), \varphi \in C^{2+\alpha}(\bar{\Omega}), \psi \in C^{2+\alpha, \frac{2+\alpha}{2}}\left(\Sigma_T\right)$, with $0<a<1$, (\ref{pde-reg}) has a unique solution from the class $C^{2+\alpha, \frac{2+\alpha}{2}}\left(\bar{Q}_T\right)$ and satisfies the inequality:
		$$
		\|f\|_{2+\alpha, \bar{Q}_T} \leq C\left(\|g\|_{\alpha, \bar{Q}_T}+\|\varphi\|_{2+\alpha, \bar{\Omega}}+\|\psi\|_{2+\alpha, \Sigma_T}\right),
		$$
		with the constant $C$ independent of $f, \varphi$ and $\psi$.	
		\end{theorem}
		
		When extending to L\'evy models and the corresponding integro-differential equations, we will need the following result. 
		\begin{lemma}\label{int-op}
			Consider $f \in C^{\alpha+2, \frac{2+\alpha}{2}}\left(\bar{Q}_T\right)$ and the differential operator:
			$$If(x, t)=\int_{\Omega}[f(x+z, t)-f(x, t)]\nu(dz).$$
			Then, for $0<a<1$, we have that:
			$$\|I f\|_{C^{\alpha, \frac{\alpha}{2}}\left(\bar{Q}_T\right)} \leq \varepsilon\|\nabla f\|_{C^{\alpha, \frac{\alpha}{2}}\left(\bar{Q}_T\right)}+C(\varepsilon)\|f\|_{C^{\alpha, \frac{\alpha}{2}}\left(\bar{Q}_T\right)}.$$
		\end{lemma}
	
		Note that Lemma \ref{int-op} is a simplified version of the corresponding results in \cite{garroni1992green}, where the authors consider additional integral operators of higher orders.

\section{Kolmogorov equations for regime switching and stochastic volatility models}\label{kolm-models}

Below we recall the Kolmogorov backward equations under the continuous regime switching, stochastic volatility and generalized models, respectively. These depend on the generators of the processes, which now include terms to capture the evolution of the regime and/or volatility processes. 
\begin{flalign}
	\frac{\partial f}{\partial t}(x,\rho, t)&=\mathcal{L}_1Q(x,\rho, t):= &&\notag\\&k_\rho(\theta_{\rho}-x) \frac{\partial f}{\partial x}(x,\rho,t) + \frac{1}{2} \sigma_\rho^2 \frac{\partial^{2} f}{\partial x^{2}}(x,\rho,t)+ \sum_{j \neq \rho} q_{\rho j} \Big(Q(x,j,t) - f(x,\rho,t)\Big) \label{l_1}&&\\
	\frac{\partial f}{\partial t}(x,y,t)&=\mathcal{L}_2 f(x,y, t):= &&\notag\\k(\theta -& x) \frac{\partial f}{\partial x}(x,y,t) + \kappa (\mu -y)\frac{\partial f}{\partial y}(x,y,t) + \frac{1}{2} y \frac{\partial^{2} f}{\partial x^{2}}(x,y,t) + \frac{1}{2} \xi^2 y \frac{\partial^{2} Q}{\partial y^{2}}(x,y,t) \label{l_2}\\
	\frac{\partial f}{\partial t}(x,\rho, y,t)&=\mathcal{L}_3 f(x,\rho,y, u):=k_\rho(\theta_{\rho}-x) \frac{\partial f}{\partial x}(x,\rho,y,t) + \kappa (\mu -y)\frac{\partial f}{\partial y}(x,\rho, y,t)  &&\notag \\ +\frac{1}{2} \sigma_\rho^2 &y\frac{\partial^{2} f}{\partial x^{2}}(x,\rho,y,t)+\frac{1}{2} \xi^2 y \frac{\partial^{2} f}{\partial y^{2}}(x,\rho, y,t) +\sum_{j \neq \rho} q_{\rho j} \Big(f(x,j,y,t) - f(x,\rho,y,t)\Big)  \label{l_3},&&
\end{flalign}
where we define separately the generator operators $\mathcal{L}_1, \mathcal{L}_2$ and $\mathcal{L}_3$, for notational convenience. For more details on the generators of regime switching and stochastic volatility models see e.g., \cite{hainaut2011financial,zhu2015feynman}.  

\section{PIDEs for the PD functions in Sobolev spaces}\label{sob-app}
For completeness, we first recall some basic definitions pertaining to the theory of weak derivatives in Sobolev spaces. %Note that the definitions are generally given for any open subset $\Omega \subset \mathbb{R}^n$, however we will present them directly for the space we are interested in, $Q$.

\begin{definition} ({\bf Weak derivative})\label{weak-deriv}
	Consider an open subset $\Omega \subset \mathbb{R}^n$ and the space of continuous functions which are $k$ times continuously differentiable, for $k=1,2,\dots,$ denoted by $C^k( \Omega )$ and $L^1_{loc}(\Omega)$ the space of locally integrable functions. Furthermore, let $\alpha =(\alpha_1, \dots, \alpha_n)$ be a multi-index, with order $|\alpha|:= \sum_i \alpha_i$, and denote $D^a u$ by:
	$$D^\alpha u=\frac{\partial^{|\alpha|} u}{\partial x_1^{\alpha_1} \ldots \partial x_n^{\alpha_n}}=\frac{\partial^{\alpha_1}}{\partial x_1^{\alpha_1}} \ldots \frac{\partial^{\alpha_n}}{\partial x_n^{\alpha_n}} u.$$ Then, for $f \in L^1_{\text{loc}}(\Omega)$ we define $u \in L^1_{\text{loc}}(\Omega)$ to be the $\alpha$th weak derivative of $f$, $D^{\alpha}f = u$, if:
	$$  \int_{\Omega} f D^\alpha \varphi d x=(-1)^{|\alpha|} \int_{\Omega} u \varphi d x, $$ for every smooth test functon with compact support, $\varphi$. 
\end{definition}
With $\mathcal{Q}:= \mathcal{D}\times \mathcal{V} \times [0,T]$, we are interested in functions which are twice weakly differentiable with respect to the initial condition and once with respect to the time until maturity. Hence, we can therefore work in the Sobolev space containing all such functions $W^{2,1}(\mathcal{Q})=\left\{f \in L^p(\mathcal{Q}): D^\alpha f \in L^1(\mathcal{Q}),|\alpha| \leqslant 2\right\}$. 
\par To work in this space, we also need an appropriate weak version of the It\^o formula, pertaining to cases where the underlying function may not enjoy the required regularity properties; these results are given by Theorems \ref{jump1} and \ref{jump2}, %\subsection{A weak version of the It\^o formula for stochastic processes}\label{weak}
due to \cite{krylov2008controlled} and \cite{okhrati2015ito}, respectively. We include the results below, for completeness:

\begin{theorem} \label{jump1}  %check
	Consider the stochastic process $$dX_t=a(x,t) dt + \sigma(x,t) d B_{t}$$ and a region $\mathcal{Q}$, where $B_t$ is a standard Brownian motion, with function $f$ such that function $f \in W^{2,1}(\mathcal{Q})$. Moreover, let $\tau$ be some Markov time such that $\tau < \tau_\mathcal{Q}$, where $t_Q$ is the exit time of the process from the region $\mathcal{Q}$. Then, if there exists some constant $K$ such that $|\sigma(x,t)|+|a(x,t)| \leq K$, for some fixed time $s$ we have that:
	%	\begin{flalign}
	%			&e^{-\varphi_{\tau}} f(s+\tau, X_\tau)-e^{-\varphi_{t}} f(s+t, X_t) = \int_{t}^{\tau} e^{-\varphi_{u}}\frac{\partial f}{\partial u}(s+u, X_u) du \notag\\
	%			&+\int_{t}^{\tau} e^{-\varphi_u} \frac{\partial f}{\partial x}(s+u, X_u) \sigma(u,X_u) dB(t)
	%			+\frac{1}{2}\int_{t}^{\tau} e^{-\varphi_{u}} \frac{\partial^2 f(s+u, X_u)}{\partial x^2} \sigma^2(u,X_u) du,
	%	\end{flalign}
	\begin{flalign}
		f(X_\tau, s+\tau)&- f(X_t,s+t) = \int_{t}^{\tau} \frac{\partial f}{\partial u}(X_u, s+u) du \notag\\
		&+\int_{t}^{\tau} \frac{\partial f}{\partial x}(X_u, s+u) \sigma(X_u,u) dB(t)
		+\frac{1}{2}\int_{t}^{\tau}  \frac{\partial^2 f(X_u, s+u)}{\partial x^2} \sigma^2(X_u,u) du,
	\end{flalign}
	almost surely.
\end{theorem}

\begin{theorem} \label{jump2}
	Consider the stochastic process with representation
	$$
	X_t=\gamma t +\int_{0}^{t}\int_{\mathbb{R}} z N(dz, du),$$
	where $\gamma \in \mathbb{R}$. 
	Assume $f: \mathcal{Q} \rightarrow \mathbb{R}$ is a continuous function on $U\mathcal{Q}$ such that $f \in L_{loc}^{1}(\mathcal{Q})$, i.e., $f$ is locally integrable. Furthermore, assume the existence of locally bounded weak first order derivatives, as defined in \ref{weak-deriv}. Then:
	\begin{eqnarray}
		f(X_t,t)=f(X_0,0)+\int_{0}^{t} \frac{\partial f}{\partial s}( X_u,u) d u+\gamma \int_{0}^{t} \frac{\partial f}{\partial x}(X_u,u) du+\nonumber \\ 
		\quad+\int_{0}^{t} \int_{\mathbb{R}}f( X_{u-}+z,u)-f(X_{u-},u) N(d u, d z),
	\end{eqnarray}
	where all derivatives are understood in the weak sense.
\end{theorem}
In the Sobolev setting, we can now derive a PIDE using the common approach of appropriate martingale arguments (similar analysis has been given in e.g., \cite{moller1995stochastic}).
\begin{lemma} \label{gen1}
	The survival probability with a variable starting time and fixed maturity $T$, $\Phi(x,\rho, y, s;T)$ satisfies the partial integro-differential equation, almost surely:
	\begin{flalign} \label{integrodiff}
		&\frac{\partial \Phi}{\partial s}(x,\rho, y, s; T)+ \mathcal{L}_3\Phi(x,\rho, y, s; T) + \int_{ \mathbb{R}}\Big(\Phi(x+z,\rho, y,s;T)-\Phi(x,\rho, y, s; T)\Big)\nu(dz)=0,
	\end{flalign}
	for $(x,\rho,y,s) \in  \mathcal{D} \times \mathcal{R} \times \mathcal{V}\times [0,T]$, with initial and boundary conditions: 
	\begin{align} \nonumber
		\Phi(x,\rho,y,T;T)= \mathds{1}_{\{x > 0\}} ,\,\,\ (x,\rho,y) \in \mathcal{D} \times \mathcal{R} \times \mathcal{V},\\ \nonumber
		\Phi(0,\rho,y, s;T) = 0, \,\,\ (\rho, y, s ) \in  \mathcal{R} \times \mathcal{V} \times [0,T], \nonumber \\
		\Phi(x,\rho, y, s;T) \rightarrow 1, \text{ as } x \rightarrow \infty, \,\,\ (\rho,y, s) \in \mathcal{R} \times \mathcal{V} \times  [0,T] , \nonumber \\
		\frac{\partial \Phi}{\partial y}(x, \rho, y,s;T) = 0, \text{ as } y \rightarrow \infty \,\,\ (x,\rho,s) \in  \mathcal{D} \times \mathcal{R} \times [0,T],
	\end{align} 
	with the generator operator $\mathcal{L}_3$ as given in (\ref{l_3}). 
\end{lemma} 
\begin{proof}
	We begin by considering the dynamics of the survival probability. As $\Phi$ is differentiable in $W^{2,1}(Q)$, we will employ Theorems \ref{jump1} and \ref{jump2} above. We then obtain:
	\begin{flalign}
		\Phi(G_w,R_w,& Y_w, w) - \Phi(G_s,\rho, y, s) = \int_{s}^{w} \Big(\frac{\partial \Phi}{\partial r}(G_r,R_r, Y_r, r) +  \mathcal{L}_3\Phi(x,\rho, y, s) \Big) dr \nonumber\\&+  
		\int_{s}^{w} \sigma(R_r)\sqrt{Y_r} \frac{\partial \Phi}{\partial x}(G_r,R_r, Y_r, r) dB_{r} + \int_{s}^{w} \xi \sqrt{Y_r} \frac{\partial \Phi}{\partial x}(G_r,R_r, Y_r, r) dW_{r}&\nonumber\\&+\int_{s}^{w} \int_{ \mathbb{R}} \Big( \Phi(G_r+z,R_r, Y_r, r)-\Phi(G_r,R_r, Y_r, r) \Big) N(dr, dz), \quad \quad
	\end{flalign}
	where the derivatives are understood in the weak sense in accordance to definition \ref{weak-deriv}. Also note that we omit the dependence on the $t$ parameter, for brevity. %In the expression above we constain the jump sizes $z$ so that the returns process $G_{r_-} + z $ remains positive, i.e. $ z > -G_r$.
	We are now able to formulate the following result regarding the survival probability.
	%We have shown that the dynamics of $\Phi$ are given by:
	%\begin{eqnarray}
	%\Phi(u, G_{u})-\Phi(s, x)=\int_{u}^{s}(\frac{\partial {\Phi}}{\partial r}+\frac{1}{2} %\sigma^{2} \frac{\partial^{2} \Phi}{\partial x^{2}}+k(\theta-G_{r}) \frac{\partial %\Phi}{\partial x}) dr  \nonumber\\
	%+\int_{s}^{u} \frac{\partial \Phi}{\partial x} \sigma dW_{r}+\int_{s}^{u} \int_{z %\geq-G_{r-}} \Big(\Phi(r, G_{r-}+z)-\Phi(r, G_{r-}) \Big) N(dr, dz)
	%\end{eqnarray}
	We write the dynamics above in terms of the compensated Poisson measure $\tilde{N}(dt,dz)=N(dt,dz)-\nu(dz)dt$.
	The last term then becomes:
	$$\int_{s}^{w} \int_{\mathbb{R}}\Big(\Phi(G_r+z,R_r, Y_r, r)-\Phi(G_r,R_r, Y_r, r) \Big)(\tilde{N}(dr, dz)+\nu(dz) dr).$$
	Combining with the dynamics of $\Phi$ above and using the fact that the sum of the non-martingale quantities must be identically zero, we obtain PIDE (\ref{integrodiff}), as required. The boundary conditions follow by definition of the survival probability.
\end{proof}

\section{Existence and continuity of the PD function}

\begin{theorem} \label{arzela-ascoli} {\bf Arzelà-Ascoli}. 
	Let $(X, d)$ be a compact metric space and $C(X)$ the space of continuous functions on $X$. Then, if a sequence of continuous functions $\{f\}_{n=1}^{\infty}$ in $C(X)$ is bounded and equicontinuous it has a uniformly convergent subsequence.
\end{theorem}

\begin{theorem} {\bf Schauder Fixed Point}\label{schauder-state}
	Let $(X,\|\cdot\|)$ be a Banach space and $S \subset X$ is compact, convex, and nonempty. Any continuous operator $A: S \rightarrow S$ has at least one fixed point.
\end{theorem}

\section{Regularity of solutions to parabolic PDEs}\label{reg-par-pde}
We will require the following results pertaining to the regularity of solutions of the second order parabolic PDE (\ref{kolm}). The following are due to \cite{garroni1992green}. 

\begin{theorem}\label{reg-theorem}
	Consider a bounded domain $\Omega$, the operator $L:= \frac{\partial f}{\partial t}(x,t) - \mathcal{L}f(x,t)$, where $\mathcal{L}$ is the generator operator, and the PDE:
	\begin{eqnarray}\label{pde-reg}
		\begin{cases}
			Lf = g(x,t) & \text{ for } (x,t) \in Q_T \\ 
			f(x,0)= \varphi(x) & \text { for } x \in \Omega \\
			f(x,t) = \psi(x,t) & \text{ for } x \in \Sigma_T:=\partial\Omega \times [0,T].
		\end{cases}
	\end{eqnarray}
	Then, for any $g \in C^{\alpha, \frac{\alpha}{2}}\left(\bar{Q}_T\right), \varphi \in C^{2+\alpha}(\bar{\Omega}), \psi \in C^{2+\alpha, \frac{2+\alpha}{2}}\left(\Sigma_T\right)$, with $0<a<1$, (\ref{pde-reg}) has a unique solution from the class $C^{2+\alpha, \frac{2+\alpha}{2}}\left(\bar{Q}_T\right)$ and satisfies the inequality:
	$$
	\|f\|_{2+\alpha, \bar{Q}_T} \leq C\left(\|g\|_{\alpha, \bar{Q}_T}+\|\varphi\|_{2+\alpha, \bar{\Omega}}+\|\psi\|_{2+\alpha, \Sigma_T}\right),
	$$
	with the constant $C$ independent of $f, \varphi$ and $\psi$.	
\end{theorem}

When extending to L\'evy models and the corresponding integro-differential equations, we will need the following result. 
\begin{lemma}\label{int-op}
	Consider $f \in C^{\alpha+2, \frac{2+\alpha}{2}}\left(\bar{Q}_T\right)$ and the differential operator:
	$$If(x, t)=\int_{\Omega}[f(x+z, t)-f(x, t)]\nu(dz).$$
	Then, for $0<a<1$, we have that:
	$$\|I f\|_{C^{\alpha, \frac{\alpha}{2}}\left(\bar{Q}_T\right)} \leq \varepsilon\|\nabla f\|_{C^{\alpha, \frac{\alpha}{2}}\left(\bar{Q}_T\right)}+C(\varepsilon)\|f\|_{C^{\alpha, \frac{\alpha}{2}}\left(\bar{Q}_T\right)}.$$
\end{lemma}

Note that Lemma \ref{int-op} is a simplified version of the corresponding results in \cite{garroni1992green}, where the authors consider additional integral operators of higher orders.

		\bibliographystyle{plain}
		\bibliography{Levy_OU_PD}
		
		%\section{Appendix}
		%\subsection{Important results for It\^o-L\'evy processes}

	\end{document}